\documentclass[twoside]{article}
\usepackage{algorithm}

\usepackage{amsfonts}
\usepackage{color}
\usepackage{appendix}
\usepackage{graphicx}
\usepackage{tabularx}
\usepackage{longtable}
\usepackage[noend]{algpseudocode}
\makeatletter

\usepackage[accepted]{format}

\usepackage{natbib}
\usepackage{amsthm}

\newtheorem{thm}{Theorem}[section]
\newtheorem{lem}[thm]{Lemma}
\newtheorem{prop}[thm]{Proposition}
\newtheorem{assump}{Assumption}
\newtheorem{defn}[thm]{Definition}
\newtheorem{corollary}[thm]{Corollary}

\def\ee{\mathbb{E}}
\def\E{\mathbb{E}}
\def\P{\mathbb{P}}

\def\R{\mathbb{R}}

\def\bmu{\boldsymbol{\mu}}
\def\bpi{\boldsymbol{\pi}}
\def\aa{\mathcal{A}}

\def\ff{\mathcal{F}}

\def\ss{\mathcal{S}}
\def\qq{\mathcal{Q}}

\def\zz{\mathcal{Z}}
\def\bP{\mathbf{P}}
\def\tmix{T_{\text{mix}}}
\def\ts{\tilde{s}}
\def\hs{\hat{s}}
\def\tz{\tilde{z}}
\def\hz{\hat{z}}
\def\tmu{\tilde{\mu}}

\def\mix{\text{mix}}

\DeclareMathOperator*{\argmax}{arg\,max}

\def\one{\mathbf{1}}
\newcommand{\old}[1]{}
\newcommand{\sol}[1]{}
\newcommand{\norm}[1]{\left\lVert#1\right\rVert}

\newcommand{\multiline}[1]{%
  \begin{tabularx}{\dimexpr\linewidth-\ALG@thistlm}[t]{@{}X@{}}
    #1
  \end{tabularx}
}

\begin{document}

\twocolumn[

\aistatstitle{Reinforcement Learning for SBM Graphon Games with Re-Sampling}

\aistatsauthor{Peihan Huo \And Oscar Peralta \And Junyu Guo \And Qiaomin Xie \And Andreea Minca }

\aistatsaddress{ Cornell University \And  Cornell University \And Tsinghua University  \And UW-Madison \And Cornell University} ]

\begin{abstract}
   The Mean-Field approximation is a tractable approach for studying large population dynamics. However, its assumption on homogeneity and universal connections among all agents limits its applicability in many real-world scenarios. Multi-Population Mean-Field Game (MP-MFG) models have been introduced in the literature to address these limitations. When the underlying Stochastic Block Model is known, we show that a Policy Mirror Ascent algorithm finds the MP-MFG Nash Equilibrium. In more realistic scenarios where the block model is unknown, we propose a re-sampling scheme from a graphon integrated with the finite $N$-player MP-MFG model. We develop a novel learning framework based on a Graphon Game with Re-Sampling (GGR-S) model, which captures the complex network structures of agents' connections. We analyze GGR-S dynamics and establish the convergence to dynamics of MP-MFG. Leveraging this result, we propose an efficient sample-based $N$-player Reinforcement Learning algorithm for GGR-S without population manipulation, and provide a rigorous convergence analysis with finite sample guarantee. 
\end{abstract}

\section{Introduction}
\paragraph{Motivating Example} This paper draws inspiration from a real-world dilemma: during the COVID-19 pandemic, individuals faced a daily choice. Should they practice social distancing and be safer or interact with other people and become more exposed? These decisions, repeated daily, effectively constituted a stochastic game. The outcome was not solely reliant on the individuals' health and actions but also hinged on the behavior of those around, alongside the unpredictable nature of virus transmission. Adapting and learning from their experiences, people refined their strategies in the wake of the pandemic. \emph{So, how do individuals learn their optimal strategies, and how do these decisions impact the population at large?}

The nature of these problems, characterized by sequential decision-making and game-theoretical elements, positions Reinforcement Learning (RL) as a promising paradigm for finding effective solutions; moreover, the Mean-Field Game (MFG) framework proposed by \cite{caines_and_huang_MFG, lasry_lions_2007} has been widely adopted to resolve the intractability of large population dynamics (See \cite{laurière2022learning} for more comprehensive reviews on RL in MFG). While MFG provides for a scalable approach for studying learning in large populations, it may not be a perfect model for our motivating example due to the following reasons: I. \emph{Homogeneity Assumption.} The MFG framework assumes all agents in the population are identical and exchangeable. However, in the context of the virus, there are variations among individuals in terms of susceptibility and the severity of symptoms when infected \citep{COVID19}. II. \emph{Universal and Time-Invariant Connections.} In the example, individuals are  only affected by people they interact with over time, rather than the entire population as implied by the MFG framework.

To overcome the limitations of the MFG framework while still ensuring scalability, we turn to the Multi-Population Mean-Field Game (MP-MFG) framework. In MP-MFG, large populations of agents are divided into multiple homogeneous populations based on relevant features within the modeling context. Several notable studies contribute to this framework: \cite{ghosh2020model} consider the MP-MFG with two distinct populations. \cite{subramanian2022multi} introduce the MP-MFG learning model, where transitions and rewards depend only on mean-field actions. \cite{perolat2021scaling} analyze a scalable learning algorithm for an MP-MFG model where state transitions do not rely on population impact, which results in low computational complexity but limits the model's applicability. In the MP-MFG model we study, we incorporate the Stochastic Block Model (SBM) \citep{SBM} to represent the connections between populations. Moreover, state transitions and rewards depend on the SBM-weighted aggregate of all population mean-field state distributions. This formulation encompasses more sophisticated dynamics where each population receives a different level of impact from other populations, thus mitigating the homogeneity issue often associated with the MFG framework. 

On the other hand, graphon games (GG) have been developed in literature to capture the locality and uniqueness of each agent's interactions with others \citep{carmona_graphon_game}. A graphon game serves as a versatile model for infinite population games, and equilibria of graphon game have been shown to approximate Nash equilibria of finite network games sampled from the graphon \citep{Parise_graphon_game}. Recent works including \cite{Caines_2021,Gao_2021,aurell2021graphonepidemics,aurell2021graphonII} have established strong connections between graphons and MFGs, offering valuable insights into using graphons' unique property of capturing heterogeneous interactions.

Building upon these previous efforts, we propose a novel Graphon Game with Re-Sampling (GGR-S) model that seamlessly integrates repeated sampling from a Stochastic Block Model Graphon into our MP-MFG learning framework. In this model, each agent's state transition and reward depend on the states of their temporary `neighbors' - the people to whom they are connected during a particular time step. The connections are sampled from the SBM graphon which represents the underlying hidden true network connections between different populations of agents. This approach effectively captures the probabilistic nature of our motivating example: based on their personal traits, individuals' daily social interactions are realizations of certain ground truth distributions captured by the graphon. We remark that \cite{fabian2023learning} consider a related Graphon MFG model in the finite-horizon setting. However, their model does not include the re-sampling feature; and in their theoretical analysis, the dependence on the mean-field impact was dropped, thus limiting the model's applicability to complex real-world phenomenon.

In this work, our goal is to learn an approximate MP-MFG Nash Equilibrium via Reinforcement Learning from a finite $N$-player game when the underlying Stochastic Block Model, transition kernel, and reward function are unknown. We aim to design a single-path trajectory learning algorithm, unlike many existing methods that involve controlling the population distribution (for example, see \cite{anahtarci2022qlearning, zaman2023oraclefree}). In particular, we employ the Policy Mirror Ascent (PMA) algorithm for policy update, which has been applied in single-agent Markov Decision Processes (MDP) \citep{LanPMD} and standard MFG \citep{Yardim_PMA}. Furthermore, for policy evaluation, we use the Conditional Temporal Difference (CTD) algorithm \citep{kotsalis2021simple}, which was originally designed for one-player setting and has been extended to \emph{homogeneous} $N$-player games \citep{Yardim_PMA}.
Our algorithm achieves efficient learning, without population manipulation, and removes the restriction that all agents in a large population are fully connected at all times. This reflects the real-world learning process and provides a more adequate framework for learning from realistic data-sets.

The main contributions of this paper are summarized as follows: I. \emph{Model Development}: We introduce a novel learning model - Graphon Game with Re-Sampling (GGR-S). This model addresses the homogeneity limitation of MFG models and incorporates network structures to more accurately represent real-world population dynamics. Moreover, the re-sampling mechanism in the finite $N$-player game is a more realistic data-collection model. II. \emph{Convergence Analysis of GGR-S}: We conduct rigorous analyses of the dynamics of GGR-S, considering transitions and rewards dependent on the graphon-induced population impact. In particular, we quantify the expected deviation of the empirical state distribution of the finite-agent GGR-S model from that of the MP-MFG, which implies the convergence of the GGR-S dynamics to that of the MP-MFG as population size tends to infinity. Moreover, we characterize the mixing property of the Markov chains in GGR-S model, which leads to the design and analysis of sample-efficient Conditional TD-learning (CTD) for policy evaluation. III. \emph{Efficient Learning from Single-path Trajectories}: The convergence results on GGR-S inspire us to develop a Policy Mirror Ascent algorithm with CTD from single-path trajectories in the GGR-S framework to learn the MP-MFG NE. Moreover, we prove that our algorithm learns a $\epsilon+\mathcal{O}(1/\sqrt{\min_{i\in[K]} N_i})$-NE with a sample complexity of $\Tilde{\mathcal{O}}(\epsilon^{-2})$, where $N_i$ is the number of agents in population $i.$

\section{Multi-Population Mean-Field Game}
In this section, we formally define the MP-MFG model and outline the iterative equilibrium-solving scheme for MP-MFG with complete information. Then, we introduce sample-based learning algorithm.

\paragraph{Notations}
We use $\ss$ and $\aa$ to represent the finite state space and action space, respectively, for all agents. Let $\Delta(\cdot)$ denote the probability simplex on a finite space. The set of stationary policies for each agent is denoted as $\Pi := \{\pi: \ss \to \Delta(\aa)\},$ where $\pi(a|s)$ is the probability of taking action $a$ in state $s$. The set of policy profiles for $K$ entities is represented by $\Pi^K$. Denote the regularization function as $h: \Delta(\aa) \to R_{\geq 0}$. Let $W_K \in (0,1]^{K\times K}$ be a symmetric matrix with entries $W_K(i, j):=p_{ij}$. The 0-1 adjacency matrix is represented as $W^{[N]} \in \{0,1\}^{N\times N}$.

\subsection{Multi-Population Mean-Field Game}
\label{subsection 2.1: MP-MFG}

We consider a Multi-Population Mean Field Game with $K$ distinct populations, each consisting of an infinite number of agents. Agents within each population are identical and exchangeable. Each population exerts an influence on other populations, with varying degrees of strength encapsulated by the matrix $W_K$, which we will hereafter refer to as the Stochastic Block Model (SBM). Precisely, the impact of population $i$ on population $k$ is quantified by $p_{ki}:= W_K(k,i)$. Individual agents within population $k$ are indexed by $(k,l)$, where $l$ identifies a specific agent within population $k$.

In the sequential decision-making problem, agent $(k,l)$ starts with an initial state $s^{k,l}_0$ sampled from population-$k$ mean-field state distribution $\mu^k_0$, and takes actions according to a policy $\pi^k \in \Pi$ that is prescribed by the representative agent of population $k$. This process induces a random sequence of states $\{s^{k,l}_t\}$ and rewards $\{r^{k,l}_t\}$ for each agent $(k,l)$, which evolve as follows: $a^{k,l}_t \sim \pi^k(\cdot|s_t^{k,l}), s^{k,l}_{t+1} \sim P(\cdot|s^{k,l}_t, a^{k,l}_t, z_t^{k}), r^{k,l}_t = R(s^{k,l}_t, a^{k,l}_t, z^{k}_t),$
where $P$ and $R$ are the state transition function and reward function, respectively.
Here $z_t^k$ is the aggregated impact for population $k$ at time step $t$, i.e., the weighted sum of mean-field distributions of all $K$ populations,
\begin{align}
    z^k_t = \frac{1}{K}{\sum_{i=1}^K W_K(k,i)\cdot \mu_{t}^i}, \label{eq:z_def}
\end{align}
with $\mu_t^i \in \Delta(\ss)$ is the mean-field state distribution of population $i$ at time $t$. Let $\bmu = (\mu^1, \dots, \mu^K)\in \Delta^K(\ss)$ denote the mean-field ensemble of all $K$ populations, and let $\bpi = (\pi^1,\dots,\pi^K) \in \Pi^K$ denote the policy profile for all $K$ populations. 

Suppose all $K$ populations adapt the policy profile $\bpi$ under the mean-field ensemble $\bmu$. Then, for population $k$, we define the regularized state-action value function, $Q^k \in \qq := \{Q: \ss\times \aa \to \R\}$, as: for all initial state-action pairs $(s,a)$,
\begin{align*}
    Q^k(s,a|\bpi, \bmu)&=\ee\left[\sum_{t=0}^\infty \gamma^t\left(R(s^{k,l}_t,a^{k,l}_t,z_t^k) + h(\pi^k(s^{k,l}_t)\right)\right],
\end{align*} 
where $(s_0^{k,l}, a_0^{k,l}) = (s,a)$,\,$a^{k,l}_{t+1} \sim \pi^k(s_t),\, s^{k,l}_{t+1} \sim P(\cdot|s^{k,l}_t,a^{k,l}_t,z_t^k)$ with $z^k_t$ defined in Eq. \eqref{eq:z_def}, and $\gamma \in (0,1)$ is the discount factor. Furthermore, we define for each population $k$, $q^k(s,a|\bpi,\bmu) = Q^k(s,a|\bpi, \bmu) - h(\pi^k(s));$ and the value function $ V^k(s|\bpi,\bmu) = \mathbb{E}_{a \sim \pi^k(\cdot|s)}\left[Q^k(s,a|\bpi, \bmu)\right]$. Here, $h:\Delta(\aa) \to \R_{\geq 0}$ is a $\rho$-strongly concave regularizer. Commonly used $\rho$-strongly-concave regularizers include the entropy function, $h(x) = -x\log x$, and the negative squared $\ell_2$ norm, $h(x) = -\norm{x}_2^2$. In general, adding a regularization term facilitates the convergence of RL algorithms, as explored in \cite{geist2019theory} and \cite{cen2021fast}; for MFGs, see \cite{xie21g} and \cite{anahtarci2022qlearning} for using regularizers to ensure solution uniqueness and achieve fast algorithm convergence.

Each population aims to find the optimal policy that maximizes their value function $V^k$ while interacting with the mean-field ensemble $\bmu$. This objective gives rise to the definition of the Nash Equilibrium (NE) in Multi-Population Mean-Field Game (MP-MFG).

\begin{defn}[MP-MFG NE]
    \label{defn: MP MFG-NE} 
A set of policies $\bpi^* = (\pi^{*,1},\dots,\pi^{*,K}) \in \Pi^K$ and a mean-field ensemble $\bmu^* = (\mu^{*,1}, \dots, \mu^{*, K}) \in \Delta^K(\ss)$ pair $(\bpi^*, \bmu^*)$ is called an MP-MFG NE if the following conditions are satisfied for all $k\in [K]$:
\begin{itemize}
  \setlength\itemsep{0em}
    \item Population Consistency:  $$\mu^{*, k}(s) = \sum_{s',a'}\mu^{*, k}(s') \pi^{*,k}(a'|s')P(s|s',a',z^{*,k}).$$
    \item Policy Optimality: $$V^k(\pi^{*,k}, \bmu^*) = \max_{\pi \in \Pi}V^k(\pi,\bmu^*).$$
\end{itemize}
\end{defn}

Intuitively, a Nash Equilibrium is reached when each $\pi^{*,k}$ is the optimal policy for population $k$, and the mean-field ensemble $\bmu^*$ stabilizes under the policy profile $\bpi^*$ (i.e., $\bmu^*$ remains fixed under $\bpi^*$). 
Furthermore, the MP-MFG NE approximates the Nash equilibrium for a finite-agent multi-population game, given that there are sufficiently large number of agents in each population (see \cite{Bensoussan_book} for details).

\subsection{Solution to MP-MFG with Complete Information}
\label{subsection: Solution to MP-MFG with Complete Information}
We now establish applicability of the fixed-point iteration method, specifically the Policy Mirror Ascent (PMA) algorithm for finding the NE of MP-MFG, when complete information is available, i.e., the transition kernel $P$, reward function $R$, and the SBM $W_K$ are known. We also introduce the definitions of relevant operators and their associated properties that will be useful in subsequent sections, with detailed proofs of these results deferred to Appendix \ref{appendix: Solution to MP-MFG with Complete Information}.

\paragraph{Notations}
We use the discrete metric $d(x,y) = \boldsymbol{1}\{x \neq y\}$ for state space $\ss$ and action space $\aa$, and we equip $\Delta(\ss)$, $\Delta(\aa)$ with vector 1-norm $\norm{\cdot}_1.$ 
For policies $\pi,\pi'\in \Pi$, let $\norm{\pi-\pi'}_1:=\sup_{s\in \ss} \norm{\pi(s)-\pi'(s)}_1;$ for policy profiles $\bpi, \bpi' \in \Pi^K$, $\norm{\bpi-\bpi'}_1 := \max_{k\in [K]}\norm{\pi^k - \pi^{k\prime}}_1$; and for mean-field ensembles  $\bmu, \bmu' \in \Delta^K(\ss),$ $\norm{\bmu-\bmu'}_1 := \max_{k\in [K]}\norm{\mu^k - \mu^{k\prime}}_1$. Assume that $h$ is a $\rho$-strongly concave regularizer, then we define $u_{\max} := \argmax_{u \in \Delta(\aa)} h(u)$ and $h_{\max}:= h(u_{\max})$. Let $\pi_0 \in \Pi$ such that $\pi_0(\cdot) := u_{\max}$ and $Q_{\max} := \frac{1+h_{\max}}{1-\gamma}$. For some constant $c>0$, define strategy set $\mathcal{U}_{c} := \{u\in \Delta(\aa): h(u) \geq h_{\max} - c\}$ and policy set $\Pi_{c} := \{\pi \in \Pi: \pi(s) \in \mathcal{U}_{c}, \forall s\in \ss\}.$

We first outline the PMA iterative scheme as a three-step process: at each iteration $i$,
\begin{enumerate}
    \item Population Update: find the stable mean-field ensemble $\bmu_{i}$ induced by a given policy profile $\bpi_{i}$.
    \item Policy Evaluation: find the $Q$ functions $\{Q_i^k\}$ associated with the fixed $\bmu_i$ and policy profile $\bpi_i$;
    \item Policy Improvement: update policy profile $\bpi_{i+1}$ via mirror ascent mechanism with $\{Q_i^k\}$.
\end{enumerate}

To guarantee convergence of the PMA iterative scheme, a smoothness assumption on transition kernel and reward is needed. This type of assumption is standard in the MFG literature \citep{anahtarci2022qlearning, Yardim_PMA}. 
\begin{assump}[Lipschitz Continuous Transition and Reward]
\label{Lipschitz Continuity of $P,R$}
There exist constants $p_\mu, p_s, p_a, r_\mu, r_s, r_a \geq 0$ such that for any population $k$, $s, s'\in \ss,$ $ a,a'\in \aa,$ $\bmu,\bmu'\in \Delta^K(\ss)$, and $z^k, z^{k\prime}$ as defined in Eq.~\eqref{eq:z_def}, the transition kernel $P$ is jointly Lipschitz continuous:
\begin{multline*}
    \norm{P(\cdot|s,a,z^k) - P(\cdot|s',a',z^{k\prime})}_1\\ \leq p_{\mu} \norm{z^k - z^{k\prime}}_1 + p_s d(s,s') + p_a d(a,a'),
\end{multline*}
and the reward function $R$ is jointly Lipschitz continuous as well:
\begin{multline*}
    \left|R(s,a,z^k) - R(s',a',z^{k\prime})\right|\\ 
    \leq r_{\mu}\norm{z^k - z^{k\prime}}_1 + r_s d(s,s') + r_a d(a,a').
\end{multline*}
\end{assump}   

First, we define the operators to represent the population state evolution process. 
\begin{defn}[Population Update Operators]\label{defn: Population Update Operator}
Under the setting of Section \ref{subsection 2.1: MP-MFG}, the $k$-th population update operator $\Gamma_{pop}[k]: \Delta^K(\ss) \times \Pi \to \Delta(\ss)$ is defined as: for each $\bmu \in \Delta^K(\ss)$, $\bpi \in \Pi^K$, $s\in \ss$,
\begin{multline*}
    \Gamma_{pop}[k](\bmu,\bpi)(s) :=\\ \sum_{s'\in \ss}\sum_{a'\in \aa} \mu^k(s')\pi^k(a'|s')P(s|s',a',z^k). 
\end{multline*}
Then, define the total population update operator which updates the mean-field ensemble as $ \Gamma_{pop}(\bmu,\bpi) = \left(\Gamma_{pop}[1](\bmu,\bpi), \dots, \Gamma_{pop}[K](\bmu,\bpi)\right).$ 
\end{defn}

Note that $\Gamma_{pop}[k]$ represents the one-step mean-field evolution. Additionally, we can show that $\Gamma_{pop}[k](\cdot,\bpi)$  is Lipschitz continuous w.r.t.~$\bmu$ with constant $L_{pop} = \frac{1}{2}p_s + p_a + p_\mu$, and it follows that $\Gamma_{pop}(\cdot,\bpi)$ is also Lipschitz continuous with $L_{pop} = \frac{1}{2}p_s + p_a + p_\mu$ (See Lemma \ref{lemma: Lipschitz k-th Population Updates} and \ref{Lipschitz continuity of Gamma_pop} in Appendix). To guarantee the convergence of the mean-field ensemble to a unique limit, the population operator needs to be contractive, as stated in the following assumption.

\begin{assump}[Contractive Population Update]
\label{assumption : stable population}
The population update operator $\Gamma_{pop}(\cdot,\bmu)$ is contractive for each $\bpi\in \Pi^K$, i.e. $L_{pop} < 1.$
\end{assump}

Assumption~\ref{assumption : stable population} implies the existence of the operator $\Gamma_{pop}^\infty: \Pi^K \to \Delta^K(\ss)$, which generates the stable mean-field ensemble under a fixed policy profile $\bmu,$
i.e., $\Gamma_{pop}(\Gamma_{pop}^\infty(\bpi),\bpi) = \Gamma_{pop}^\infty(\bpi)$ for any policy profile $\bpi$. 

Next, we define the Policy Mirror Ascent (PMA) operators, including the policy evaluation operator that maps policy profiles and mean-field ensembles to the corresponding $Q$-functions, and the policy improvement operator that update the policy profiles based on $Q$-functions. We remark that in an online learning setting, the approximation of the policy evaluation operator using solely policy profile roll-outs becomes instrumental, facilitating the execution of single-path trajectory learning in the following section.

\begin{defn}[Policy Mirror Ascent Operators]
\label{PMA Operator on Stationary Population Distribution}
For each population $k\in [K]$, we define the policy evaluation operator $\Gamma_q[k]: \Pi \times \Delta^K(\ss) \to \mathcal{Q}$ as $\Gamma_q[k](\bpi, \bmu) = q^k(\cdot, \cdot|\bpi,\bmu) \in \qq$. We define the policy improvement operator $\Gamma^{pma}_{\eta}:\qq \times \Pi \rightarrow \Pi$ as: for each $q\in \qq,\pi\in \Pi,s\in \ss$
\begin{multline*}
    \Gamma^{pma}_{\eta}(q,\pi)(s)\\ := \argmax_{u \in \mathcal{U}_{L_h}}\langle u,q(s,\cdot)\rangle + h(u) - \frac{1}{2\eta}\norm{u - \pi(s)}_2^2,
\end{multline*}
where $\eta$ is a chosen learning rate, and $L_h = r_a + \frac{\gamma r_s p_a}{2-\gamma p_s}$.
\end{defn}

Consequently, the PMA three-step iteration for the entire system can be summarized as an operator $\Gamma_\eta: \Pi^K \to \Pi^K$ as $
\Gamma_\eta(\bpi) = \left(\Gamma_\eta[1](\bpi), \dots, \Gamma_\eta[K](\bpi)\right),$ where for each population $k\in [K]$, $\Gamma_\eta[k]: \Pi^K \to \Pi$, where for each $\bpi=(\pi^{1},\cdots,\pi^{K})\in \Pi^K,$ $$\Gamma_\eta[k](\bpi) := \Gamma_\eta^{pma}(\Gamma_q[k](\bpi,\Gamma^\infty_{pop}(\bpi)),\pi^k).$$ 

We remark that $\Gamma_\eta$ is Lipschitz continuous with Lipschitz constant $L_\eta$, and that the contractivity of $\Gamma_\eta$ can be ensured through a proper selection of strongly-concavity parameter $\rho$ and learning rate $\eta$ (see Lemma \ref{lemma: fixed point of gamma_eta} and \ref{Lipschitz Continuity of Gamma_eta} in Appendix). 
Furthermore, we show that the MP-MFG NE policy is contained in $\Pi^K_{L_h}$(see Lemma \ref{lemma: Sufficiency of Pi Delta} of the Appendix), which implies that the fixed point of $\Gamma_\eta$ is indeed the MP-MFG-NE. 

Finally, we conclude the dicussion of the complete information case by proving the linear convergence of PMA iteration to MP-MFG NE policy.

\begin{prop}[Convergence to MP-MFG NE with Complete Information]
\label{theorem: Convergence to MP-MFG NE in the Exact Case}
Suppose that the MP-MFG NE is given by $(\bpi^*, \bmu^*)$ and for a given learning rate $\eta > 0$, $L_\eta < 1$. Assume for all $k \in [K]$, $\pi^k_0 = \pi_{\max}$ and consider the updates $\bpi_{t+1} = \Gamma_\eta(\bpi_t)$ for all $t \geq 0$. It is guaranteed that for any $T\geq 1$, 
\begin{align*}
    \norm{\bpi_T - \bpi^*}_1 \leq L_\eta^T\norm{\bpi_0 - \bpi^*}_1 \leq 2L^T_{\eta}.
\end{align*}
\end{prop}

\subsection{Sample-Based Learning with Finite Players for MP-MFG }\label{section: Sample-Based Learning with Finite Players (MP-MFG)}

In practice, the transition kernel $P$ and the reward function $R$ are usually unknown. We then need to simultaneously learn the MP-MFG NE from data. A popular setting that has been studied in MFG literature \citep{subramanian,guo2023general} assumes access to a simulator, which allows one to sample state transition and reward from an arbitrary mean-field ensemble and state-action pair. With such an simulator, we can approximate the three steps of PMA in Section~\ref{subsection: Solution to MP-MFG with Complete Information} through querying the simulator oracle. 
In Appendix \ref{section: The Simulator-Oracle-Based Learning}, we provide such a simulator-based learning algorithm for MP-MFG, along with a proof of its linear convergence.

 However, the assumption of having access to a simulator imposes significant limitations. In particular, for online learning, we aim to relax two restrictions: (i) the game process can be restarted at will; (ii) there are an infinite number of agents playing the game.

To address these limitations, we look into a more practical scenario where the only source of information consists of trajectories from finite $N$ agents. Let $N_k$ denote the number of agents within population $k$. Precisely, for each $\tau>0$, let $\hat\ff_{\tau}$ denote the $\sigma$-algebra, $\hat\ff_\tau := \ff\left(\{\hs^{k,l}_t, a_t^{k,l}, r_t^{k,l}\}_{t=1}^{\tau}, k\in [K], \ell \in [N_k]\right)$. Let $\hat{\bmu}_{t}$ denote the empirical mean-field ensemble at time $t$, with $\hat \mu^k_t = \frac{1}{N_k} \sum_{l=1}^{N_k} \delta_{\hat{s}_{t}^{k,l}}$, and let $\hz_t$ denote the empirical aggregated impact based on the \emph{known} SBM $W_K$, defined as \begin{align}
    \hz^k_t = \frac{1}{K}{\sum_{i=1}^K W_K(k,i)\cdot \hat \mu_{t}^i}. \label{eq: hat_z-def}
\end{align} In essence, we learn from empirical data gathered from the trajectories of the $N$ players. 

For online learning with trajectories, a mixing condition is usually assumed to ensure sufficient exploration. In line with the assumptions considered by \cite{Yardim_PMA}, we emphasize that the mixing condition can be effectively reduced to a combination of an assumption regarding non-degenerate policy and an assumption concerning the reachability of states.

\begin{assump}[Non-Degenerate Policies]
\label{assumption: non-degenerate policies}
Assume that there exists $\zeta$ such that: I. For any population $k \in [K]$, the initial policy $\pi^k_0(a|s) \geq \zeta, \forall s\in \ss, a\in \aa$; II. For any $\pi \in \Pi$, $q\in \qq$ satisfying $\pi(a|s) \geq \zeta$, $0 \leq q(s,a) \leq Q_{\max}, \forall s \in \ss, a\in \aa$, it holds that $\Gamma_\eta^{pma}(q,\pi)(a|s) \geq \zeta, \forall s\in \ss, a\in \aa.$
\end{assump}

The next assumption dictates that any state can be visited with non-zero probability under a non-degenerate policy within finite time steps. 
\begin{assump}[Reachability Under Non-degenerate Policies]
\label{assumption: reachability}
Under the settings of MP-MFG in Section \ref{subsection 2.1: MP-MFG}, for any policy profile $\bpi \in \Pi^K$ satisfying $\pi(a|s) \geq \zeta > 0, \forall s \in \ss$, $a \in \aa$, and any initial states $\{ s_0^{k,l}\} \in \ss^N$, there exist $T_{mix} > 0$, $\delta_{mix} > 0$ such that for all $(k,l)$,$$\P(\hat s^{k,l}_{T_{mix}} = s'|\{s_0^{k,l}\}) \geq \delta_{mix}, \forall s' \in \ss.$$
\end{assump}

Under the mixing assumptions and within the framework outlined in Section \ref{subsection: Solution to MP-MFG with Complete Information}, one can show that the centralized learning algorithm proposed by \cite{Yardim_PMA} aptly applies to sample-based learning in the MP-MFG model. We highlight the key steps of the analysis in Appendix \ref{appendix: Sample-Based Learning with Finite Players (MP-MFG)}. 

\section{Learning Graphon Game with Re-Sampling}
While the MP-MFG framework allows for effective learning algorithms, it relies on the unrealistic assumption that the SBM is known and that all agents maintain full connectivity accordingly, even when population sizes become large. To tackle this challenge, we introduce a novel re-sampling procedure during the learning process of finite $N$-player games. This procedure ensures that, at each time step, agents are only connected to a subset of the population, representing their `neighbors'. The SBM is unknown, and agents only observe an aggregate impact of their neighbors. Our primary objective is to demonstrate that after introducing this re-sampling process, the resulting dynamics will not deviate significantly from the fully connected MP-MFG case, ensuring that the PMA-CTD algorithm ultimately learns the MP-MFG NE.

\subsection{GGR-S Model}
\label{subsection: GGR-S Model}
\label{section: network games sampled from graphons}
A graphon is a symmetric, measurable, real-valued function $W$ on $[0,1] \times [0,1]$, and it can be used to represent the probability distribution over the space of dense networks. In this paper, the relevant class of graphons is the Stochastic Block Model (SBM) Graphon \citep{Parise_graphon_game}. In particular, we partition $[0,1]$ into $K$ disjoint intervals $\{\mathcal{I}_k\}_{k=1}^K$, each of length $|\mathcal{I}_k| = L_k$, and agents with labels $x$ and $y$, located in the intervals $\mathcal{I}_k$ and $\mathcal{I}_j$ respectively, are connected with probability $W(x,y) = W_{K}(x,y) = p_{k,j} = p_{j,k} > 0$. We then define the procedure of sampling from a SBM graphon \citep{lovász2012}. 

\begin{defn}[Adjacency Matrix Sampled from SBM Graphon]
\label{defn: network games sampled from SBM graphon}
    Given an SBM graphon $W$ and $N$ agents from $K$ distinct populations indexed by $(k,l)$, and without loss of generality, assume that the indices $(k,l)$ are ordered such that each $(k,l)$ can be mapped to a distinct label $t^{k,l}$ on the interval $[0,1]$. Then, we can generate the 0-1 adjacency matrix $W^{[N]} \in \{0,1\}^{N\times N}$ of a simple sampled network by randomly connecting agents $(k,l)$ and $(i,j)$ with Bernoulli probability $W(t^{k,l},t^{i,j}) = p_{k,i}$.
\end{defn}

Incorporating this sampling scheme into the MP-MFG setting, we obtain a more realistic and versatile model. 

 \begin{defn}[Graphon Game with Re-Sampling Model]
\label{defn: Locally Centralized Learning Model}
Given an SBM graphon $W$ and $N$ agents from $K$ distinct populations, at each time step $t$, re-sample the 0-1 adjacency matrix $W^{[N]}_t$ as outlined in Definition \ref{defn: network games sampled from SBM graphon}. Then, each agent $(k,l)$, starting from the current state $\ts_t^{k,l}$, 
takes an action $a^{k,l}$ according to some policy $\pi^k_t$, prescribed by the representative agent of their respective population, receives a reward $r_t^{k,l}$ and transitions to the next state $\ts^{k,l}_{t+1}$. That is, for any agent $(k,l)$ at time $t$, $a^{k,l}_t \sim \pi_t^k(\cdot|\ts_t^{k,l}), \ts^{k,l}_{t+1} \sim P(\cdot|\ts^{k,l}_t, a^{k,l}_t, \tz_t^{k,l}), r^{k,l}_t = R(\ts^{k,l}_t, a^{k,l}_t, \tz^{k,l}_t),$
where $\tz_t$, the empirical neighbor impact, is defined as: 
 \begin{align}
        \tilde{z}_t^{k,l} = \frac{1}{K}\sum_{i=1}^K\frac{1}{N_i}\left[\sum_{j=1}^{N_i} W_t^{[N]}(t^{k,l},t^{i,j})\cdot \delta_{\tilde s_t^{i,j}}\right]. \label{eq: tilde_z-def}
\end{align}
\end{defn}

We remark that the key difference between $N$-player MP-MFG and GGR-S model is that the empirical aggregated impact $\tz$ in GGR-S only depends on the \emph{neighbors} of each agent at each step of transition. The GGR-S model lifts the restriction that all agents must be connected at all times. 

SBM graphons effectively model the population's underlying community structure. The reasonable conjecture that different types of people have varying interactions justifies the use of the SBM graphon model in our motivating example. Additionally, the re-sampling scheme models the randomness of individual's interactions over time. Moreover, if the SBM graphon $W$ satisfies that for any pair of agents $(k,l), (i,j)$, $W(t^{k,l},t^{i,j}) = W_K(k,i)$, where $W_K(k,i)$ is the population connectivity matrix defined in Section \ref{subsection 2.1: MP-MFG}, then the MP-MFG can be viewed as the limiting case of the model in Definition \ref{defn: network games sampled from SBM graphon} as the number of agents goes to infinity. This connection has also been explored in the static game setting by \cite{carmona2019graphonI}. 

\subsection{Dynamics of GGR-S}
\label{section: dynamics of GGR-S}
In this section, we analyze the dynamics of $N$-player GGR-S and prove its convergence to the fully connected MP-MFG case. This result will allow us to show that we can effectively estimate the iterative scheme $\Gamma_\eta$ in Section \ref{subsection: Solution to MP-MFG with Complete Information} using only sample paths from a finite number of agents even when the agents' connections change dynamically over time, as detailed in Definition \ref{defn: network games sampled from SBM graphon}. Moreover, we show that this is possible without introducing additional assumptions other than the mixing conditions Assumption \ref{assumption: non-degenerate policies} and Assumption \ref{assumption: reachability} in the MP-MFG case. 

\paragraph{Notations} In the GGR-S framework, let the filtrations (excluding the empirical neighbour impact) be $\ff_\tau := \ff\left(\{\ts^{k,l}_t, a_t^{k,l}, r_t^{k,l}\}_{t=1}^{\tau}, k\in [K], l \in [N_k]\right)$. Also, we define the one-step observation space as $\Omega := \ss \times \aa \times [0,1] \times \ss \times \aa ,$ where $\omega_t^{k,l} = (\tilde s_t^{k,l}, a_t^{k,l}, r_{t}^{k,l}, \ts_{t+1}^{k,l}, a_{t+1}^{k,l}) \in \Omega$ denotes the one-step observation at time $t$ for agent ${(k,l)}$. Denote the finite set of possible empirical state distribution for population $k$ by $\tilde{\Delta}_{N,\ss} \subset \Delta(\ss)$, and the finite set of possible empirical neighbor impact by $\zz_N$.

We start with quantifying the error generated by the re-sampling scheme by computing the distance between the empirical aggregated impact and the empirical neighbor impact. 

\begin{lem}[One-Step Error Propagation Through Aggregated Impact]
\label{lemma: One-Step Error Propagation Through Aggregated Impact}
Assume that at any time $t \geq 1$, agents follow a given policy profile $\bpi$. Let $\tilde z_t^{k,l}$ be defined as in Eq. \eqref{eq: tilde_z-def}, and $\tilde \bmu_t$ denote the resulting the mean-field ensemble at time $t$ under GGR-S; meanwhile, let $\hat z_t^k$ be defined as in Eq. \eqref{eq: tilde_z-def}, and $\hat \bmu_t$ denote the resulting the mean-field ensemble at time $t$ under MP-MFG.
Then, for all $t \geq 0$,
\begin{multline*}
    \ee\left[\norm{\tilde z_t^{k,l} - \hat z_t^{k}}_1\big|\ff_t,\hat\ff_t\right]\\ \leq p^*_k  \norm{\tilde \bmu_t - \hat \bmu_t}_1 + 2\sqrt{\frac{2|\ss|}{K \cdot \min_i{N_i}}},
\end{multline*}
with $p^*_k = \max_{i\in [K]}W_K(k,i)$. 
Furthermore,
\begin{multline*}
    \ee\left[ \norm{\tilde \bmu_{t+1} - \hat \bmu_{t+1}}_1\big|\ff_t,\hat\ff_t\right]\\ \leq   2(1+p_\mu)\sqrt{\frac{2|\ss|}{\min_i N_i}} + \tilde L_{pop} \norm{\tilde \bmu_t - \hat \bmu_t}_1,
\end{multline*}
where $\tilde L_{pop} = p_*p_\mu + \frac{1}{2}p_s + p_a$ with $p_* = \max_k p^*_k$.
\end{lem}
\emph{Proof: }See Section \ref{proof for lemma: One-Step Error Propagation Through Aggregated Impact} of the Appendix.

Lemma \ref{lemma: One-Step Error Propagation Through Aggregated Impact}  facilitates the analysis of error propagation between mean-field ensembles of MP-MFG and GGR-S over multiple time steps.

\begin{prop}[Multi-Step Error Propagation Bound]
\label{prop: Multi-Step Error Propagation Bound}
    Under the settings of Lemma \ref{lemma: One-Step Error Propagation Through Aggregated Impact}, for any $t, \tau \geq 0$, 
    \begin{multline*}
         \ee\left[ \norm{\tilde \bmu_{t+\tau} - \hat \bmu_{t+\tau}}_1\big|\ff_t,\hat\ff_t\right]\\ \leq  \frac{(2+2p_\mu)(1-\tilde L_{pop}^\tau)}{1-\tilde L_{pop}}\sqrt{\frac{2|\ss|}{\min_i N_i}} + \tilde L_{pop}^\tau  \norm{\tilde \bmu_t - \hat \bmu_t}_1,
    \end{multline*}
    where $\tilde L_{pop} = p_*p_\mu + \frac{1}{2}p_s + p_a$ with $p_* = \max_k p_k^*$.
\end{prop}
\emph{Proof: }See Section \ref{appendix: proof for prop: Multi-Step Error Propagation Bound} of the Appendix. 

Note that $\tilde L_{pop} = p_*p_\mu + \frac{1}{2}p_s + p_a \leq p_\mu + \frac{1}{2}p_s + p_a = L_{pop} < 1$ by Assumption \ref{assumption : stable population}. This means that, given enough time, the dynamics of GGR-S converge to MP-MFG, with a bias of approximately $\mathcal{O}\left(\frac{1}{\sqrt{\min_i N_i}}\right)$. Therefore, the evolution of the empirical mean-field ensemble in finite-agent GGR-S resembles the evolution of the mean-field ensemble with infinite agents.
\begin{lem}[Empirical Population Bound in GGR-S]
\label{lemma: Empirical Population Bound in the Locally Centralized Case}
Under the settings in Proposition \ref{prop: Multi-Step Error Propagation Bound}, then for all $\tau, t \geq 0$, assuming that at time $t$, GGR-S and MP-MFG share the same start, i.e. $\ff_t = \hat\ff_t$, then
\begin{multline*} 
    \E\left[\norm{\tilde{\bmu}_{t+\tau} - \Gamma^\tau_{pop}(\tilde \bmu_t, \bpi)}_1\big| \ff_t\right]\\ \leq \frac{(3+2p_\mu)(1- \tilde L_{pop}^{\tau})}{1- L_{pop}} \sqrt{\frac{2|\ss|}{\min_i N_i}},
\end{multline*}
where $\tilde L_{pop} = p_*p_\mu + \frac{1}{2}p_s + p_a$ with $p_* = \max_k p_k^*$.
\end{lem}
\emph{Proof: } See Section \ref{appendix: proof for Empirical Population Bound in GGR-S} of the Appendix. 

Another important implication of Proposition \ref{prop: Multi-Step Error Propagation Bound} is that if Assumption \ref{assumption: reachability} holds for the MP-MFG framework, then all states are also reachable in GGR-S. 
\begin{prop}[Reachability Under Non-Degenerate Policies in GGR-S]
\label{prop: Reachability Under Non-Degenerate Policies in the Locally Centralized Case}
For any policy profile $\bpi \in \Pi^K$ satisfying Assumption \ref{assumption: non-degenerate policies}, and any initial states $\{s_0^{k,l}\}_{k,l} \in \ss$, for $T_{\text{mix}}, \delta_{\text{mix}}$ that satisfies Assumption \ref{assumption: reachability}, i.e. $\P(\hat{s}^{k,l}_{\tmix} = s'|\{s_0^{k,l}\}_{k,l}) > \delta_\mix,$ it holds that for any agent $(k,l)$,
$$\P(\tilde{s}^{k,l}_{\tmix}= s'|\{s_0^{k,l}\}_{k,l}) > \delta'_\mix,$$
where $\delta'_\mix = \frac{1}{2}\delta_\mix - \frac{(1+p^*_k p_\mu)(2+2p_\mu) +2p_\mu}{1-\tilde L_{pop}}\sqrt{\frac{2|\ss|}{\min_i N_i}}$.
\end{prop}
\emph{Proof: }See Section \ref{appendix: proof for Reachability Under Non-Degenerate Policies} of the Appendix.

Note that we can guarantee mixing with $\delta_\mix' > 0$ by selecting population sizes $\{N_i\}_{i \in [K]}$ to be sufficiently large.  
Next, we see that in GGR-S, each agent's state visitation distribution converges to their respective stable mean-field, up to a population bias term. 

\begin{prop}(Convergence of State Visitation Distribution to Stable Mean-Field)
\label{Prop: Convergence of State Visitation Probabilities to Stable Mean-Field} Under the same settings as Proposition \ref{prop: Reachability Under Non-Degenerate Policies in the Locally Centralized Case}, let $\{s_0^{k,l}\}_{(k,l)}$ be arbitrary initial states of agents, and let $\bpi$ be a policy profile. Then, for any $T\geq 1$ and any agent $(k,l)$, it holds that
\begin{multline*}
    \norm{\P\left(\tilde s_T^{k,l} = \cdot\big|\{s_0^{k,l}\}_{(k,l)}\right) - \Gamma_{pop}^\infty[k](\bpi)}_1\\ 
    \leq C_\mix \rho_\mix^{T} + \frac{2p_{\mu}T_{\mix}(3p_*+2p_*p_\mu+2)}{{\delta'_{\mix}}^{2}(1-
    L_{pop})}  \sqrt{\frac{2|\mathcal{S}|}{\min_i N_i}},
\end{multline*}
 where $\rho_{\mix}:=\max\{L_{pop},(1-\delta'_{\mix})^{1/T_{\mix}}\}$, \(C_{\mix}:=\left(2+\frac{2p_{\mu}L_{pop}}{\delta'_{\mix}(1-L_{pop})|\theta-L_{pop}^{T_{\mix}}|}+\frac{2p_{\mu}L_{pop}}{\delta'_{\mix}(1-L_{pop})}\right)/(\rho_{\mix}^{T_{\mix}})\) with $\delta_{\mix}'$ defined in Proposition \ref{prop: Reachability Under Non-Degenerate Policies in the Locally Centralized Case}. 
\end{prop}
\emph{Proof: }See Section \ref{appendix: proof for Prop: Convergence of State Visitation Probabilities} of the Appendix. 

\subsection{Sample-Based Learning with Finite Agents (GGR-S)}
Recall that our goal is to learn the MP-MFG NE from the trajectories of $N$-players. Our work in Section \ref{section: dynamics of GGR-S} indicates that the empirical mean-field ensemble $\tilde \bmu$ of the $N$ players in GGR-S will converge to the stable mean-field ensemble after a sufficient mixing time. This allows us to consider the Conditional Temporal Difference (CTD) learning algorithm \citep{kotsalis2023simple} for $Q$-function estimation. Specifically, CTD updates the $Q$-function at periodic intervals of length $I_\mix$ instead of every iteration. By waiting $I_\mix$ time steps, we can ensure sufficient convergence of the empirical mean-field ensemble within each periodic interval. In particular, before updating the $Q^k$-function for each population $k$, the dynamics of finite-agent system proceed for $I_\mix$ steps and the last observation $\omega_{I_\mix}^{k,1}$ from the representative agent of population $k$ is then used for the update.

We start with defining the stochastic Temporal Difference operator used in our algorithm. 
\begin{defn}
\label{defn: stochastic TD operator}
    The stochastic TD operator is defined as follows: for all $Q \in \qq,$ $\omega:= (s,a, r,s',a') \in \Omega$, 
    \begin{multline*}
    \tilde F^\pi(Q,\omega) = (Q(s,a) - r - h(\pi(s)) - \gamma Q(s',a'))\mathbf{e}_{s,a}.
\end{multline*}
\end{defn}

Given complete information, we can execute TD-learning with respect to the abstract MDP induced by the stable mean-field ensemble $\Gamma^\infty_{pop}(\bpi)$. However, when dealing with real-world data, these limiting distributions are unobservable. Yet, we can approximate them with empirical dynamics, as justified by Lemma \ref{lemma: Empirical Population Bound in the Locally Centralized Case} and Proposition \ref{prop: Reachability Under Non-Degenerate Policies in the Locally Centralized Case}. This further implies that the iterative update of $\tilde F^\pi$ provides a robust estimate of the $Q^k$ functions. In Section \ref{appendix: CTD learning in GGR-S} of the Appendix, we provide the modified CTD learning algorithm with GGR-S dynamics, along with a finite-time convergence analysis suggesting an error bound of $\mathcal{O}((\min_i N_i)^{-1/2})$.

We now present the complete algorithm (Algorithm \ref{algo:GGRS learning}) for learning the MP-MFG NE through trajectories of finite $N$-players in the GGR-S framework. This algorithm uses samples $\tilde s^{k,l}_{t+1} \sim P(\cdot|\ts_t^{k,l},a_t^{k,l}, \tz^{k,l}_t), r_t^{k,l} = R(\cdot|\ts_t^{k,l},a_t^{k,l}, \tz^{k,l}_t)$ of GGR-S, where $\tz^{k,l}_t$ is the empirical neighbor impact which the algorithm cannot control. Hence, the algorithm is single-path trajectory and free of population manipulation. In summary, there are two main steps in the algorithm: 
\begin{enumerate}
    \item Policy evaluation: Under a fixed policy, estimate the state-action value function $Q^k$ for each population via CTD with a waiting time of $I_\mix$ steps;
    \item Policy improvement: The representative agent of each population $k$ updates their policy via Policy Mirror Ascent operator $\Gamma^{pma}_\eta$ and then synchronize the updated policy among all agents.  
\end{enumerate}

\begin{algorithm}
\caption{GGR-S PMA-CTD Learning}
\label{algo:GGRS learning}
    \begin{algorithmic}[1]
    \Require Initial policy-state pair $(\bpi_0, \{s^{k,l}_0\})$, CTD learning rate $\{\beta_n\}_n$, CTD iteration $I_{ctd}$, Mixing time $I_\mix$, PMA iteration $M$  
    \For{ $m \in 1,\dots, M$}
    \State Set $\tilde Q^k_0(\cdot,\cdot) \gets Q_{\max}$, for all $k$.
        \For{$n \in 0,1,\dots I_{ctd}-1$}
            \For{$t \in 1, \dots, I_\mix$}
                \State Re-sample $W^{[N]}_t$ from $W_K$.
                \State \multiline{Compute for all $(k,l)$: $ \tilde{z}_t^{k,l} = \frac{1}{K}\sum_{i=1}^K\frac{1}{N_i}\left[\sum_{j=1}^{N_i} W_t^{[N]}(t^{k,l},t^{i,j}) \delta_{\ts_t^{i,j}}\right]$.}
                \State  \multiline{Simulate for all $(k,l)$, $a_t^{k,l} \sim \pi_m^k(\ts_t^{k,l})$, $\ts_{t+1}^{k,l} \sim P(\cdot|\ts_t^{k,l}, a_t^{k,l}, \Tilde{z}^{k,l}_t)$, $r_{t+1}^{k,l} = R(\cdot|\ts_t^{k,l}, a_t^{k,l}, \tz^{k,l}_t)$.} 
            \EndFor
            \State \multiline{Observe $\omega_{n}^{k} = (s_{t-2}^{k,1}, a_{t-2}^{k,1}, r_{t-2}^{k,1}, s_{t-1}^{k,1}, a_{t-1}^{k,1})$.}
            \State \multiline{CTD update: $\tilde Q^k_{n+1} = \tilde Q^k_n - \beta_n \tilde F^{\pi_{m}^{k}}(\tilde Q^k_n, \omega^k_{n})$, for all $k$.}
        \EndFor
        \State PMA step: $\pi^k_{m+1} = \Gamma_\eta^{pma}(\tilde{Q}^k_{I_{ctd}}, \pi^k_m)$, for all $k$.
    \EndFor
    \State Return policy profile $\bpi_m$.
    \end{algorithmic}
\end{algorithm}

\begin{thm}[GGR-S PMA-CTD Learning]
\label{Theorem: GGR-S PMA-CTD Learning}
    Under the settings of Definition \ref{defn: Locally Centralized Learning Model}, suppose that Assumption \ref{Lipschitz Continuity of $P,R$}, \ref{assumption : stable population}, \ref{assumption: non-degenerate policies}, \ref{assumption: reachability} hold, and that $\eta > 0$ is a PMA learning rate that satisfies $L_\eta < 1$. Then, run Algorithm \ref{algo:GGRS learning} with learning rates $\beta_n = {2}\left({\frac{4(1+\gamma)^2}{(1-\gamma)(\delta'_{\mix}\zeta)} + (1-\gamma)\delta_{\mix}'\zeta(n - 1)}\right)^{-1}$ and $M >  \mathcal{O}(\log(\epsilon^{-1}))$, $I_{ctd} > \mathcal{O}(\epsilon^{-2})$, $I_\mix > \mathcal{O}(\log (\epsilon^{-1}))$. Let $\bpi^*$ be the MP-MFG NE policy profile. Then the output $ \bpi_M$ of Algorithm \ref{algo:GGRS learning} satisfies 
    \begin{align*}
        \ee\left[\norm{\bpi_M - \bpi^*}_1 \right] \leq \epsilon + \mathcal{O}\left(\frac{1}{\sqrt{\min_i N_i}}\right). 
    \end{align*}
\end{thm}

\emph{Proof: }See Section \ref{appendix: proof for theorem GGRS PMACTD learning} of the Appendix. 

In Algorithm \ref{algo:GGRS learning}, there are three loops, and it requires a total of $M \times 
I_{ctd} \times I_\mix $ steps. Thus, Theorem \ref{Theorem: GGR-S PMA-CTD Learning} yields a sample complexity of $\mathcal{O}\left(\epsilon^{-2}\log^2(\epsilon^{-1})\right)$; and consistent with prior analyses, increasing the population sizes effectively reduces the error of the learned policy.

\section{Discussion}

In this paper, we tackle the challenges of population-wise heterogeneity and local network properties through the introduction of a Graphon Game with Re-Sampling (GGR-S) model. This innovative re-sampling scheme enhances the representation of sophisticated and realistic dynamics within large populations. Furthermore, we present an efficient single-path Reinforcement Learning algorithm tailored to learning optimal policies within the GGR-S framework. 

Our work opens up several exciting directions for future research. From a game-theoretical standpoint, it would be valuable to explore scenarios where network connections are re-sampled from a general graphon, extending beyond the current utilization of SBM graphons. Additionally, the development of decentralized or distributed learning schemes holds the potential to further enhance the applicability of these models in various real-world contexts.

\bibliographystyle{apalike}
\bibliography{references}

\newpage
\onecolumn
\aistatstitle{Learning Multi-Population Mean Field Game with Graphon Re-Sampling:
Supplementary Materials}
\appendix
\section{Solution to MP-MFG with Complete Information}
\label{appendix: Solution to MP-MFG with Complete Information}
In this section, we detail the left-out definitions and proofs in Section \ref{subsection 2.1: MP-MFG}. These results build up to the proof of convergence of the Policy Mirror Ascent algorithm. 

Firstly, we review several general results used throughout the paper:
\begin{lem}[\cite{Georgii+2011}]
\label{Yardim Lemma 8}
Assume E is a finite set, $F : E \to \R$ a real valued function, and $\mu,\nu$ are two probability measures on E.
Then,$$\norm{\sum_{e}F(e)\mu(e)-\sum_{e}F(e)\nu(e)}_1 \leq \frac{\sup_{e}F(e)-\inf_{e}F(e)}{2}\norm{\mu-\nu}_1.$$
\end{lem}
\begin{lem}[\cite{Kontorovich_2008}]
\label{Yardim Lemma 9}
Assume E is a finite set, $g: E \to \R^p$ a vector value function, and $\nu, \mu$ two probability measures on $E$. Then, 
\begin{align*}
    \norm{\sum_e g(e)\mu(e) - \sum_e g(e)\nu(e)}_1 \leq \frac{\lambda_g}{2}\norm{\mu-\nu}_1,
\end{align*}
where $\lambda_g := \sup_{e,e'}\norm{g(e) - g(e')}_1.$
\end{lem}

We provide an alternative definition of transition and reward function dependent on the mixed strategy for a given state instead of the action. This definition is useful in many proofs in the coming sections. 
\begin{defn}[Definition of $\Bar{P}$ and $\Bar{R}$] For the state transition function $P$, the reward function $R$, and for all $k \in [K]$, define $\Bar{P}:\ss \times \Delta(\aa) \times \Delta(\ss) \to \Delta(\ss)$ as the state-policy transition probabilities distributions as 
$$\Bar{P}(\cdot|s,u,z^k) := \sum_{a\in \aa}u(a) P(\cdot|s,a,z^k),$$
and $\Bar{R}:\ss \times \Delta(\aa) \times \Delta(\ss) \to [0,1]$ as 
$$\Bar{R}(s,u,z^k) := \sum_{a\in \aa}u(a) R(\cdot|s,a,z^k).$$

\end{defn}

\begin{lem}[Lipschitz Continuity of $\Bar{P},\Bar{R}$]
\label{Lipschitz Continuity of Pbar Rbar}
For all $k \in [K]$, we have for all $s,s' \in \ss,$ $u,u'\in \Delta(\aa),$ and $z^k,z^{k\prime}$ defined in Eq.~\eqref{eq:z_def},
\begin{align*}
    \left|\Bar{R}(s,u,z^k) - \Bar{R}(s', u',z^{k\prime})\right| &\leq r_{\mu} \norm{z^{k}-z^{k\prime}}_1 + r_s d(s',s') + \frac{r_a}{2}\norm{u - u^{\prime}}_1,\\
    \norm{\Bar{P}(\cdot|s,u,z^k) - \Bar{P}(\cdot|s',u^\prime,z^{k\prime})}_1 &\leq p_{\mu} \norm{z^{k}-z^{k\prime}}_1 + p_s d(s',s') + \frac{p_a}{2}\norm{u - u^{\prime}}_1.
\end{align*}
and moreover, it follows that for any $\bmu = (\mu^1, \dots, \mu^K)\in \Delta^K(\ss)$ and $\bmu'=(\mu^{1\prime}, \dots, \mu^{K\prime}) \in \Delta^K(\ss)$,
\begin{align*}
    \left|\Bar{R}(s,u,z^k) - \Bar{R}(s', u',z^{k\prime})\right| &\leq r_{\mu} \norm{\bmu-\bmu^{\prime}}_1 + r_s d(s',s') + \frac{r_a}{2}\norm{u - u^{\prime}}_1,\\
    \norm{\Bar{P}(\cdot|s,u,z^k) - \Bar{P}(\cdot|s',u^\prime,z^{k\prime})}_1 &\leq p_{\mu} \norm{\bmu-\bmu^{\prime}}_1 + p_s d(s',s') + \frac{p_a}{2}\norm{u - u^{\prime}}_1.
\end{align*}
\end{lem}

\begin{proof}
By triangle inequality,
\begin{align*}
    &|\Bar{R}(s,u,z^{k})-\Bar{R}(s',u,z^{k\prime})|\\
    &\qquad=\left|\left(\sum_{a\in \mathcal{A}}u(a)R(s,a,z^{k})- \sum_{a\in \mathcal{A}} u(a)R(s',a,z^{k\prime})\right)+\left(\sum_{a\in \mathcal{A}} u(a)R(s',a,z^{k\prime})-\sum_{a\in \mathcal{A}}u'(a)R(s',a,z^{k\prime})\right)\right|\\
    &\qquad\leq \sum_{a\in \mathcal{A}}u(a)\left| R(s,a,z^{k})-R(s',a,z^{k\prime})\right|+\left|\sum_{a\in \mathcal{A}} u(a)R(s',a,z^{k\prime})-\sum_{a\in \mathcal{A}}u'(a)R(s',a,z^{k\prime})\right|.
\end{align*}
Using Lemma \ref{Yardim Lemma 8} and Lemma \ref{Yardim Lemma 9}, we can obtain:
\begin{align*}
    \sum_{a\in \mathcal{A}}u(a)\left| R(s,a,z^{k})-R(s',a,z^{k\prime})\right| \leq r_{\mu}\norm{z^{k}-z^{k\prime}}_1+r_{s}d(s,s'),
\end{align*}
and 
$\left|\sum_{a\in \mathcal{A}}
    u(a)R(s',a,z^{k\prime})-\sum_{a\in \mathcal{A}}u'(a)R(s',a,z^{k\prime})\right| \leq \frac{r_a}{2}\norm{u^{k}-u^{k\prime}}_1.$
Summing the two inequalities together, we get 
\begin{align*}
    \sum_{a\in \mathcal{A}}u(a)\left| R(s,a,z^{k})-R(s',a,z^{k\prime})\right|&+\left|\sum_{a\in \mathcal{A}} u(a)R(s',a,z^{k\prime})-\sum_{a\in \mathcal{A}}u'(a)R(s',a,z^{k\prime})\right| \\
    &\leq r_{\mu}\norm{z^{k}-z^{k\prime}}_1+r_{s}d(s,s')+\frac{r_a}{2}\norm{u^{k}-u^{k\prime}}_1.
\end{align*}
Similarly, for $\Bar{P}$ we have:
\begin{multline*}
       \norm{\Bar{P}(\cdot|s,u,z^{k})-\Bar{P}(\cdot|s',u',z^{k\prime})}_1\\ \leq \norm{\sum_{a\in \mathcal{A}}u(a)P(\cdot|s,u,z^{k})-\sum_{a\in \mathcal{A}}u(a)P(\cdot|s',u',z^{k\prime})}_1 +\norm{\sum_{a\in \mathcal{A}}u(a)P(\cdot|s',u',z^{k\prime})-\sum_{a\in \mathcal{A}}u'(a)P(\cdot|s',u',z^{k\prime})}_1.
\end{multline*}
Using Jensen's inequality for the first term and Lemma \ref{Yardim Lemma 9} for the second term, we then complete the proof for the main conclusion. 

Next, we show that $\norm{z^{k}-{z^{k}}'}_1 \leq \norm{\bmu-\bmu'}_1.$ Indeed,
using the definition of $z^{k}$ and ${z^{k}}'$ in Eq.~\eqref{eq:z_def} and the properties of matrix norm, we have 
\begin{align*}
     \norm{z^{k}-{z^{k}}'}_1 \leq \frac{1}{K}\sum_{k=1}^{K}p_{ki}\norm{\mu^{i}-{\mu}^{i\prime}}_1\leq \frac{1}{K}\sum_{k=1}^{K}p_{ki}\norm{\bmu-{\bmu}^{\prime}}_1 \leq \norm{\bmu-{\bmu}^{\prime}}_1, 
\end{align*}
Therefore, we can also obtain 
\begin{align*}
    \left|\Bar{R}(s,u,z^k) - \Bar{R}(s', u',z^{k\prime})\right| &\leq r_{\mu} \norm{\bmu-\bmu^{\prime}}_1 + r_s d(s',s') + \frac{r_a}{2}\norm{u - u^{\prime}}_1,\\
    \norm{\Bar{P}(\cdot|s,u,z^k) - \Bar{P}(\cdot|s',u^\prime,z^{k\prime})}_1 &\leq p_{\mu} \norm{\bmu-\bmu^{\prime}}_1 + p_s d(s',s') + \frac{p_a}{2}\norm{u - u^{\prime}}_1.
\end{align*}
\end{proof}

\begin{lem} \label{lemma: Lipschitz k-th Population Updates}
The k-th population update operator $\Gamma_{pop}[k]$ is Lipschitz continuous with 
    \begin{align*}
        \norm{\Gamma_{pop}[k](\bmu,\bpi) - \Gamma_{pop}[k](\bmu',\bpi^{\prime})}_1 \leq L_{pop,\mu}\norm{\bmu-\bmu'}_1 + \frac{p_a}{2}\norm{\pi^k - \pi^{k\prime}}_1,
    \end{align*}
    where $L_{pop,\mu}:= (\frac{p_s}{2}+p_a+ p_{\mu})$, for all $\bpi,\bpi^{\prime} \in \Pi^{K}, \bmu \in \Delta^K(\ss)$.
\end{lem}
\begin{proof}
By triangle inequality, 
\begin{align*}
    \norm{\Gamma_{pop}[k](\bmu,\bpi) - \Gamma_{pop}[k](\bmu',\bpi^{\prime})}_1 = &\norm{\sum_s \mu^k(s) \Bar{P}\left(\cdot|s,\pi^k(s), z^k\right)- \sum_s \mu^{k\prime}(s) \Bar{P}\left(\cdot|s,\pi^{k\prime}(s), z^{k\prime}\right)}\\
    \leq &\underbrace{\norm{\sum_s \mu^k(s) \Bar{P}\left(\cdot|s,\pi^k(s), z^k\right) - \sum_s \mu^k(s) \Bar{P}\left(\cdot|s,\pi^{k\prime}(s), z^{k\prime}\right)}_1}_{\text{A}}\\
    &+  \underbrace{\norm{\sum_s \mu^k(s) \Bar{P}\left(\cdot|s,\pi^{k\prime}(s), z^{k\prime}\right) - \sum_s \mu^{k\prime}(s) \Bar{P}\left(\cdot|s,\pi^{k\prime}(s), z^{k\prime}\right)}_1}_{\text{B}}.\\
\end{align*}
By Jensen's inequality and Lemma \ref{Lipschitz Continuity of Pbar Rbar}, we bound the first norm: 
\begin{align*}
    A \leq \sum_{s\in \ss} \mu^k(s) \norm{\Bar{P}\left(\cdot|s,\pi^k(s), z^k\right) - \Bar{P}\left(\cdot|s,\pi^{k\prime}(s), z^{k\prime}\right)}_1 \leq p_{\mu} \norm{\bmu-\bmu'}_1 + \frac{p_a}{2}\norm{\pi^k - \pi^{k\prime}}_1.
\end{align*}
Then, we use Lemma \ref{Yardim Lemma 9} and Lemma \ref{Lipschitz Continuity of Pbar Rbar} to bound the second norm: 
\begin{align*}
    B &\leq \norm{\mu^k - \mu^{k\prime}}_1 \frac{\sup_{s,s'\in \ss} \norm{\Bar{P}\left(\cdot|s,\pi^{k\prime}(s), z^{k\prime}\right) - \Bar{P}\left(\cdot|s',\pi^{k\prime}(s'), z^{k\prime}\right)}_1}{2}\\ 
    &\leq \norm{\mu^k - \mu^{k\prime}}_1 \frac{p_sd(s,s') + p_a\norm{\pi^{k\prime}(s) - \pi^{k\prime}(s')}_1}{2}\\
    &\leq (\frac{p_s}{2} + p_a)\norm{\mu^k - \mu^{k\prime}}_1\\
    &\leq (\frac{p_s}{2} + p_a)\norm{\bmu - \bmu'}_1,
\end{align*}
from which the lemma follows. 
\end{proof}

Since $\Gamma_{pop}[k]$ updates the k-th block, we can define the total population update as a collection of k-th population update: 
\begin{defn}[Total Population Update]
\label{defn: total population update}
The total population update operator $\Gamma_{pop} : \Delta^K(\ss) \times \Pi^K \to \Delta^K(\ss)$ is defined as: 
\begin{align*}
    \Gamma_{pop}(\bmu,\bpi) = \left(\Gamma_{pop}[1](\bmu,\bpi), \dots, \Gamma_{pop}[K](\bmu,\bpi)\right),
\end{align*}
where $\bmu = (\mu^1, \dots, \mu^K) \in \Delta^K(\ss) $ is the collection of population mean-fields, and $\bpi = (\pi^1, \dots, \pi^K)$ is the collection of population policies. 
\end{defn}

\begin{lem}[Lipschitz Continuity of $\Gamma_{pop}$]
\label{Lipschitz continuity of Gamma_pop}
    The total population update operator $\Gamma_{pop}$ is Lipschitz continuous with
    \begin{align*}
        \norm{\Gamma_{pop}(\bmu,{\bpi}) - \Gamma_{pop}(\bmu',{\bpi}')}_1 \leq L_{pop}\norm{\bmu-\bmu'}_1 + \frac{p_a}{2}\norm{\bpi - \bpi'}_1,
    \end{align*}
    where $L_{pop} = L_{pop,\mu}$ with $L_{pop,\mu}$ defined in Lemma \ref{lemma: Lipschitz k-th Population Updates}. 
\end{lem}

\begin{proof}
The result follows from Definition \ref{defn: Population Update Operator} and Lemma \ref{lemma: Lipschitz k-th Population Updates},
\begin{align*}
    \norm{\Gamma_{pop}(\bmu,\bpi) - \Gamma_{pop}(\bmu',\bpi')}_1 &= \max_{k\in [K]}\norm{\Gamma_{pop}[k](\bmu,\bpi) - \Gamma_{pop}[k](\bmu',\bpi^{\prime})}_1\\
    &\leq \max_{k\in [K]}\left[ L_{pop,\mu}\norm{\bmu-\bmu'}_1 + \frac{p_a}{2}\norm{\pi^k - \pi^{k\prime}}_1 \right]\\
    &= L_{pop,\mu}\norm{\bmu-\bmu'}_1 + \frac{p_a}{2}\max_{k\in [K]} \norm{\pi^k - \pi^{k\prime}}_1\\
    &= L_{pop,\mu}\norm{\bmu-\bmu'}_1 + \frac{p_a}{2} \norm{\bpi - \bpi^\prime}_1,
\end{align*}
which completes the proof. \end{proof}
Given the definition of $\Gamma_{pop}$ as the fixed point of the population update operator in Definition \ref{defn: Population Update Operator}, now we give a more specific definition. 
\begin{defn}[Stable Population Operator $\Gamma_{pop}^{\infty}$]
\label{defn: stable population operator gammainfty}
For any $k \in [K],$ under Assumption \ref{assumption : stable population}, the stable total population operator $\Gamma_{pop}^{\infty}: \Pi^K \to \Delta^K(\ss)$ is the unique total population distribution such that: 
\begin{align*}
    \Gamma_{pop}(\Gamma_{pop}^{\infty}(\bpi), \bpi) = \Gamma_{pop}^{\infty}(\bpi),
\end{align*}
i.e., the fixed point of $\Gamma_{pop}(\cdot,\bpi):\Delta^K(\ss) \to \Delta(\ss).$
\end{defn}
It follows that $\Gamma^\infty_{pop} = \lim_{n\to \infty} \Gamma^n_{pop}$ and that $\Gamma^\infty_{pop}$ is Lipschitz continuous:

\begin{lem}[Lipschitz Continuity of $\Gamma_{pop}^\infty$]
\label{Lipschitz contuity gamma_infty_pop}
The operator $\Gamma_{pop}^\infty: \Pi^K \to \Delta(\ss)$ is Lipschitz continuous: 
\begin{align*}
    \norm{\Gamma_{pop}^\infty(\bpi) - \Gamma_{pop}^\infty(\bpi')}_1 \leq L_{pop,\infty}\norm{\bpi - \bpi'}_1,
\end{align*}
where $L_{pop,\infty}:= \frac{p_a}{2(1-L_{pop})}$.
\end{lem}

\begin{proof}
    By Definition 4 and Lemma \ref{Lipschitz continuity of Gamma_pop},
    \begin{align*}
        \norm{ \Gamma_{pop}^{\infty}(\bpi) -  \Gamma_{pop}^{\infty}(\bpi')}_1 &= \norm{ \Gamma_{pop}(\Gamma_{pop}^{\infty}(\bpi), \bpi) -  \Gamma_{pop}(\Gamma_{pop}^{\infty}(\bpi'), \bpi')}_1 \\
        &\leq L_{pop}\norm{\Gamma_{pop}^{\infty}(\bpi)-\Gamma_{pop}^{\infty}(\bpi')}_1 + \frac{p_a}{2} \norm{\bpi - \bpi^\prime}_1\\
        \Rightarrow  \norm{ \Gamma_{pop}^{\infty}(\bpi) -  \Gamma_{pop}^{\infty}(\bpi')}_1 &\leq \frac{p_a}{2(1-L_{pop})}\norm{\bpi - \bpi^\prime}_1,
    \end{align*}
    which gives Lipschitz constant $L_{pop,\infty} = \frac{p_a}{2(1-L_{pop})}$.
\end{proof}
Using Bellman Expectation Equations, we can rewrite the definition of the value functions and we state the Bellman Expectation Equations for these definitions of value functions: 
\begin{defn}[Bellman Expectation Equations]
\label{bellman expectation equations}
    For all $k \in [K]$, for any $\bpi = (\pi^{1},\cdots,\pi^{K}) \in \Pi^{K}$ and $\bmu \in \Delta^K(\ss)$, the value functions $V^k, Q^k, q^k$ satisfy: 
    \begin{align*}
        V^k(s|\bpi,\bmu) &= \mathbb{E}_{a \sim \pi^k(\cdot|s)}\left[R(s,a,z^k) + h(\pi^k(s)) + \gamma \mathbb{E}_{s'\sim P(\cdot|s,a,z^k)}\left[V^k(s'|\bpi,\bmu)\right]\right].\\
        Q^k(s,a|\bpi,\bmu) &= R(s,a,z^k) + h(\pi^k(s)) + \gamma \mathbb{E}_{s'\sim P(\cdot|s,a,z^k), a'\sim \pi^k(\cdot|s')}\left[Q^k(s',a'|\bpi,\bmu)\right].\\
        q^k(s,a|\bpi,\bmu) &= R(s,a,z^k) + \gamma \mathbb{E}_{s'\sim P(\cdot|s,a,z^k), a'\sim \pi^k(\cdot|s')}\left[q^k(s',a'|\bpi,\bmu)+ h(\pi^k(s'))\right].
    \end{align*}
\end{defn}

We also include the Bellman Optimality equations for completeness: 
\begin{defn}[Bellman Optimality Equations]
\label{bellman optimality equations defn}
    For all $k \in [K]$, for any $\bpi = (\pi^{1},\cdots,\pi^{K}) \in \Pi^{K}$ and $\bmu \in \Delta^K(\ss)$, the optimal value functions $V^{*,k} ,Q^{*,k}, q^{*,k}$ satisfy: 
    \begin{align*}
        V^{*,k}(s|\bmu) &= \max_{\pi  \in \Delta(\aa)}\left[\mathbb{E}_{a\sim \pi(\cdot|s)} \left( R(s,a,z^k) + h(\pi(s)) + \gamma \mathbb{E}_{s'\sim P(\cdot|s,a,z^k)}\left[V^{*,k}(s'|\bmu)\right]\right)\right].\\
        Q^{*,k}(s,a|\bmu) &= R(s,a,z^k) + h(\pi^{*,k}(s)) + \gamma \mathbb{E}_{s'\sim P(\cdot|s,a,z^k)}\left[V^{*,k}(s'|\bmu)\right].\\
        q^{*,k}(s,a|\bmu) &= R(s,a,z^k) + \gamma \mathbb{E}_{s'\sim P(\cdot|s,a,z^k)}\left[V^{*,k}(s'|\bmu)\right].
    \end{align*}
\end{defn}

\begin{lem}
\label{boundedness in s of V^k_h}
    Suppose $\bpi \in \Pi^{K}_{\Delta_h}$ and $\bmu \in \Delta^K(\ss)$, then for any $s_1,s_2 \in \ss$, 
    \begin{align*}
        \left|V^k(s_1|\bpi,\bmu)- V^k(s_2|\bpi,\bmu)\right| \leq L_{V,s}:= \frac{r_s + r_a + \Delta_h}{1- \gamma \min\{1,\frac{p_s+p_a}{2}\}}.
    \end{align*}
\end{lem}
\begin{proof}
    By  Lemma \ref{Lipschitz Continuity of Pbar Rbar}, Lemma \ref{Yardim Lemma 8}, and Definition \ref{bellman expectation equations},
    \begin{align*}
        |V^k(s_1|\bpi,\bmu) &- V^k(s_2|\bpi,\bmu) | \\
        \leq &\left|\Bar{R}(s_1,\pi^k(s_1),z^k) - \Bar{R}(s_2,\pi^k(s_2),z^k)\right| + \left|h(\pi^k(s_1)) - h(\pi^k(s_2))\right|\\ 
        &+ \gamma\left|\sum_{s'\in \ss}\left(\Bar{P}(s'|s_1, \pi^k(s_1),z^k)-\Bar{P}(s'|s_2, \pi^k(s_2),z^k)\right)V^k(s'|\bpi,\bmu)\right|\\
        \leq & r_s + r_a + \Delta_h\\
        &+ \frac{\gamma\sup_{s,s^\prime \in \ss}\left|V^k(s|\bpi,\bmu) - V^k(s|\bpi,\bmu)\right|}{2}\norm{\Bar{P}(s'|s_1,\pi^k(s_1),z^k) - \Bar{P}(s'|s_2, \pi^k(s_2), z^k)}_1\\
        \leq &r_s + r_a + \Delta_h + \frac{\gamma \min\{2,p_s+p_a\}}{2}\sup_{s,s'}|V^k(s|\bpi,\bmu) - V^k(s'|\bpi,\bmu)|,
    \end{align*}
    and taking the supremum on the left side proves the lemma. 
\end{proof}

\begin{lem}[Lipschitz Continuity of Value Function $V^k$]
\label{lemma: Lipschitz Continuity of Value Function V^k}
    Assume that $\Delta_h > 0$ arbitrary. For any $\bpi,\bpi^{\prime}\in \Pi^{K}_{\Delta_h}$, $k,k' \in [K]$ and $\bmu,\bmu^\prime \in \Delta^K(\ss),$
    \begin{align*}
        \norm{V^k(\cdot|\bpi,\bmu) - V^k(\cdot|\bpi^{\prime},\bmu^\prime)}_\infty \leq L_{V,\pi} \norm{\pi^k - \pi^{k\prime}}_1 + L_{V,\mu} \norm{\bmu - \bmu^{\prime}}_1,
    \end{align*}
    for $$L_{V,\pi} = \frac{4r_a + \gamma p_a L_{V,s}}{4(1-\gamma)}, L_{V,\mu} = \frac{2r_{\mu} + \gamma p_{\mu} L_{V,s}}{2(1-\gamma)},$$ where $L_{V,s}$ is defined in Lemma \ref{boundedness in s of V^k_h}.
\end{lem}

\begin{proof}
    For an arbitrary state $s\in \ss$,
    \begin{align*}
        |V^k(s|\bpi,\bmu) &- V^k(s|\bpi^{\prime},\bmu^\prime)| \leq \left|\Bar{R}(s,\pi^k(s),z^k) - \Bar{R}(s,\pi^k(s),z^{k\prime})\right| + \left|h(\pi^k(s)) - h(\pi^{k\prime}(s))\right|\\
        &+  \gamma\left|\sum_{s'\in \ss}\left(\Bar{P}(s'|s, \pi^k(s),z^k)V^k(s'|\bpi,\bmu)-\Bar{P}(s'|s, \pi^{k\prime}(s),z^{k\prime})V^k(s'|\bpi^{\prime},\bmu^\prime)\right)\right|\\
        \leq &r_a\norm{\pi^k - \pi^{k\prime}}_1 + r_{\mu} \norm{\bmu - \bmu^\prime}_1 + \Delta_h+ \gamma \frac{L_{V,s}}{2}\left(p_{\mu}\norm{\bmu -\bmu'}_1 + \frac{p_a}{2}\norm{\pi^k - \pi^{k\prime}}_1\right)\\
        &+ \gamma \sup_{s\in \ss}\left|V^k(s|\bpi,\bmu) - V^k(s|\bpi^{\prime},\bmu^\prime)\right|, \text{ by Lemma \ref{Lipschitz Continuity of Pbar Rbar} and \ref{boundedness in s of V^k_h}},
    \end{align*}
    and taking supremum over the left hand side completes the proof. 
\end{proof}

\begin{lem}[Lipschitz continuity of $\Gamma_q$]
\label{lemma: Lipschitz Gamma qk}
    For arbitrary $\Delta_h$, for all $k$, there exists $L_{q,\pi}, L_{q,\mu}$ depending on $\Delta_h$ such that for all $\bpi,\bpi^{\prime} \in \Pi_{\Delta_h}^{K}$ and $\bmu,\bmu^\prime \in \Delta^K(\ss)$,
    \begin{align*}
        \norm{\Gamma_q[k](\bpi,\bmu) - \Gamma_q[k](\bpi^{\prime},\bmu')}_\infty \leq L_{q,\pi} \norm{\pi^k - \pi^{k\prime}}_1 + L_{q,\mu}\norm{\bmu- \bmu'}_1.
    \end{align*}
    with $L_{q,\pi}=\gamma L_{V,\pi},L_{q,\mu}=r_{\mu}+\gamma L_{V,\mu}$.
\end{lem}

\begin{proof}
    By Definition \ref{bellman expectation equations}, $q^k$ can be rewritten in terms of $V^k$: 
    \begin{align*}
        q^k(\cdot,\cdot|\bpi,\bmu) = R(s,a, z^k) + \gamma \E_{s'\sim P(\cdot|s,a,z^k)}V^k(s'|\bpi,\bmu),
    \end{align*}
    so by triangle inequality,
    \begin{align*}
        \left|\Gamma_q[k](\bpi,\bmu) - \Gamma_q[k](\bpi^{\prime},\bmu')\right| \leq &|R(s,a,z^k) - R(s,a,z^{k\prime})|\\
        &+ \gamma \left|\ee_{s'\sim P(\cdot|s,a,z^k)}V^k(s'|\bpi,\bmu) - \ee_{s'\sim P(\cdot|s,a,z^{k\prime})}V^k(s'|\bpi^{\prime},\bmu')\right|,
    \end{align*}
    and the result follows from Lemma \ref{lemma: Lipschitz Continuity of Value Function V^k}.
\end{proof}

\begin{lem}[Sufficiency of $\Pi_{\Delta_h}$]
\label{lemma: Sufficiency of Pi Delta}
Let $\bmu \in \Delta^K(\ss)$ be arbitrary, and $\bpi^* = (\pi^{*,1},\dots, \pi^{*,K}) \in \Pi^K$ the optimal response policy collection such that for any $k$, and for all $s\in \ss$, 
$$V^k(s|\pi^{*,k},\bmu) = \max_{\pi \in \Pi} V^k(s|\pi,\bmu).$$ Then, $\pi^{*,k} \in \Pi_{L_h}^*$ where $L_h := r_a +  \frac{\gamma r_s p_a}{2-\gamma p_s}$.
\end{lem}

\begin{proof}
    For arbitrary $k$, by Bellman Optimality equation of $V^{*,k}_h$ (Definition \ref{bellman optimality equations defn}), it then follows from Lemma \ref{Yardim Lemma 8}, 
    \begin{align*}
        |V^{*,k}
    (s_1|\bmu) &- V^{*,k}(s_2|\bmu)|\\
        \leq &\Bigg|\sup_{u \in \Delta(\aa)}\left(\Bar{R}(s_1,u,z^k) + h(u) + \gamma \ee_{s\sim \Bar{P}(\cdot|s_1,u,z^k)}\left[V^{*,k}(s|\bmu)\right]\right)\\
        &- \sup_{u \in \Delta(\aa)}\left(\Bar{R}(s_2,u,z^k) + h(u) + \gamma \ee_{s\sim \Bar{P}(\cdot|s_2,u,z^k)}\left[V^{*,k}(s|\bmu)\right]\right) \Bigg| \\
        \leq &\sup_{u \in \Delta(\aa)}\left|\Bar{R}(s_1,u,z^k) -\Bar{R}(s_2,u,z^k) + \gamma \sum_{s\in \ss}\left(\Bar{P}(s|s_1,u,z^k) - \Bar{P}(s|s_2,u,z^k)\right)V^{*,k}(s|\bmu)\right|\\
        \leq &r_s + \frac{\gamma p_s}{2}\sup_{s,s'}|V^{*,k}(s|\bmu) - V^{*,k}(s'|\bmu)|,
    \end{align*}
    and taking supremum over the left-hand side yields: 
    \begin{align*}
        \sup_{s_1,s_2\in\ss}|V^{*,k}(s_1|\bmu) &- V^{*,k}(s_2|\bmu)| \leq \frac{r_s}{1-\frac{\gamma}{2} p_s}.
    \end{align*}
    Then, by Bellman Optimality Equation for $q^{*,k}$ (Definition \ref{bellman optimality equations defn}), we have: for all $s\in \ss$,
    \begin{align*}
        \sup_{a,a'\in \aa}|q^{*,k}(s,a|\bmu) - q^{*,k}(s,a'|\bmu)| \leq &|R(s,a,z^k) - R(s,a',z^k)|\\
        &+ \gamma\sum_{s'\in \ss}\left|\left(P(s,a,z^k) - P(s,a',z^k)\right)V^{*,k}(s'|\bmu)\right|\\
        \leq &r_a + \gamma \frac{p_a}{2}\frac{r_s}{1-\frac{\gamma}{2}p_s}.
    \end{align*}
    By optimality of $\pi^{*,k}$ as defined in Definition \ref{defn: MP MFG-NE} ,
    \begin{align*}
        V^{*,k}(s|\bmu) = \sum_{a\in \aa}q^{*,k}(s,a|\bmu) + h(\pi^{*,k}(s)) = \max_{u\in\Delta(\aa)}\langle q^{*,k}(s,\cdot |\bmu), u \rangle + h(u) \geq \langle q^{*,k}(s,\cdot |\bmu), u_{\max} \rangle + h_{\max},
    \end{align*}
     so this means that
    \begin{align*}
        h_{\max} - h(\pi^{*,k}(s)) \leq \langle q^{*,k}(s,\cdot |\bmu), \pi^{*,k}(s) - u_{\max}\rangle \leq \sup_{a,a'\in \aa}|q^{*,k}(s,a|\bmu) - q^{*,k}(s,a'|\bmu)|,
    \end{align*}
    which completes the proof of the lemma. 
\end{proof}

\begin{lem}[Lipschitz Continuity of $\Gamma^{pma}_\eta$]
\label{Lipschitz Continuity of mirror descent}
\label{lipschitz continuity of mirror descent q function generator}
     $\Gamma^{pma}_\eta$ is Lipschitz continuous, i.e. for all $q,q'\in \qq$, $\pi,\pi' \in \Pi$, it holds that 
     \begin{align*}
          \norm{\Gamma^{pma}_\eta(q,\pi) -  \Gamma^{pma}_\eta(q',\pi')}_1 \leq L_{md,\pi}\norm{\pi - \pi'}_1 + L_{md,q}\norm{q-q'}_\infty,
     \end{align*}
     where $L_{md,\pi} = \frac{1}{|\aa|^{-1}+\eta \rho}$ and $L_{md,q} = \frac{2\eta \sqrt{|\aa|}}{1 + 2\eta\rho \sqrt{|\aa|}}$.
\end{lem}
Since we update the state-action value function by population, the proof for this lemma follows from \cite{Yardim_PMA} (See Section E.3).

\begin{lem}[Fixed Point of $\Gamma_\eta$ as MPMFG-NE]
\label{lemma: fixed point of gamma_eta}
Let $\eta > 0$ be arbitrary. A policy-mean-field tuple $(\bpi^*,\bmu^*)$ is a MP-MFG NE if and only if , $\bpi^{*} = \Gamma_\eta(\bpi^{*})$ and $\bmu^* = \Gamma_{pop}^\infty(\bpi^*)$.
\end{lem}

\begin{lem}[Lipschitz Continuity of $\Gamma_\eta$]
\label{Lipschitz Continuity of Gamma_eta}
For any $\eta>0$, the operator $\Gamma_\eta: \Pi^K \to \Pi^K$ is Lipschitz continuous with constant $L_{\eta_\eta}$ on $(\Pi,\norm{\cdot}_1)$, where $L_\eta = L_{md,\pi} + L_{md,q}(L_{q,\pi}+L_{q,\mu}L_{pop,\infty})$.
\end{lem}

\begin{proof}
By  Lemma \ref{Lipschitz Continuity of mirror descent},
\begin{align*}
    \norm{\Gamma_\eta(\bpi) - \Gamma_\eta(\bpi')}_1 &= \max_{ k\in [K]}\norm{\Gamma_\eta^k(\bpi) - \Gamma_\eta^k(\bpi')}_1\\
    &\leq L_{md,\pi}\norm{\pi^k - \pi^{k\prime}}_1 + L_{md,q}\norm{\Gamma_q[k](\bpi,\Gamma^\infty_{pop}(\bpi)) - \Gamma_q[k](\bpi^{\prime},\Gamma^\infty_{pop}(\bpi'))}_\infty,
\end{align*}
where we take $k = \argmax \norm{\Gamma_\eta^k(\bpi) - \Gamma_\eta^k(\bpi')}_1$, and then by Lemma \ref{lemma: Lipschitz Gamma qk} and Lemma \ref{Lipschitz contuity gamma_infty_pop}, we obtain: 
\begin{align*}
     \norm{\Gamma_\eta(\bpi) - \Gamma_\eta(\bpi')}_1 &\leq L_{md,\pi}\norm{\bpi - \bpi'}_1 + L_{md,q}\left(L_{q,\pi} \norm{\pi^k - \pi^{k'}}_1 + L_{q,\mu}\norm{\Gamma^\infty_{pop}(\bpi)- \Gamma^\infty_{pop}(\bpi')}_1.\right) \\
     &\leq L_{md,\pi}\norm{\bpi - \bpi'}_1 + L_{md,q}(L_{q,\pi}+L_{q,\mu}L_{pop,\infty}) \norm{\bpi - \bpi'}_1\\
     &\leq \left(L_{md,\pi} + L_{md,q}(L_{q,\pi}+L_{q,\mu}L_{pop,\infty})\right)\norm{\bpi - \bpi'}_1,
\end{align*}
so the Lipschitz constant is given by $L_\eta = L_{md,\pi} + L_{md,q}(L_{q,\pi}+L_{q,\mu}L_{pop,\infty})$.

Furthermore, note that by Lemma \ref{lipschitz continuity of mirror descent q function generator}: 
\begin{align*}
    L_\eta =  \frac{1}{|\aa|^{-1}+\eta \rho} + \frac{2\eta \sqrt{|\aa|}}{1 + 2\eta\rho \sqrt{|\aa|}}(L_{q,\pi}+L_{q,\mu}L_{pop,\infty}) \leq \frac{L_{q,\pi}+L_{q,\mu}L_{pop,\infty}}{\rho} + \frac{1}{\eta \rho},
\end{align*}
which completes the proof of the lemma. 
\end{proof}
Note that in order for the mapping $\Gamma_\eta$ to be contractive, we need $L_\eta < 1$, which can be obtained with sufficiently large $\rho.$ Indeed, we can select $\eta$ such that as $\eta \to \infty$, $L_\eta \to ({L_{q,\pi}+L_{q,\mu}L_{pop,\infty}})/{\rho}$, and this means that if we pick a sufficiently large $\rho$ such that $ \rho > L_{q,\pi}+L_{q,\mu}L_{pop,\infty}$, we can always obtain $L_\eta<1.$ Conversely, if $\rho$ is sufficiently large, then we can obtain a contraction by setting the learning rate $\eta > (\rho - L_{q,\pi}+L_{q,\mu}L_{pop})^{-1}$.

\subsection{Proof for Proposition 1}
This propositionon 1 directly follows recursively from Lemma \ref{lemma: fixed point of gamma_eta} and \ref{Lipschitz Continuity of Gamma_eta}.

\section{The Simulator-Oracle-Based Learning}
\label{section: The Simulator-Oracle-Based Learning}

In this section, we consider the problem of learning the MP-MFG NE when the probability transition $P$ and reward $R$ are unknown. We need to simultaneously learn the MP-MFG system and the NE. 

We start with assuming the availability of a simulator which enables us to sample from $P$ and $R$ at any time $t$ and the aggregated impact $z_t^k$ from any arbitrary mean-field ensemble $\bmu$. Restrictive as it may seem, this assumption is reasonable in many circumstances such as online ad auction \cite{guo2023general}.

\begin{assump}[Simulator Oracle]
\label{Assumption: simulator oracle}
For any policy profile $\bpi$ and any arbitrary mean-field ensemble $\bmu_t$, for any population $k$, given the current state $s_t^k \in \ss$ at time $t$, we have access to a simulator which returns the next state $s^{k,l}_{t+1} \sim P(\cdot|s^{k,l}_t,\pi^k(s^{k,l}_t), z_t^k)$ and $r_t = R(s^{k,l}_t, \pi^k(s^{k,l}_t), z^k_t)$. Denote this simulator as $(s',r) = \mathcal{G}(s,\pi,z)$.  
\end{assump}

We remark that our simulator oracle is less restrictive than the simulator in \cite{guo2023general} which is enabled to generate the next mean field distribution whereas our version only generates the next state and reward. Instead, our definition of simulator is more comparable to the weak simulator as defined in \cite{guo2023general}.

Now, consider the scenario where there is an omniscient coordinator who knows the connectivity matrix $W_K$ and the mean-field ensemble $\bmu_t$ and the policy profile $\bpi_t$ at any time $t$. Then, at each time $t$, this coordinator will feed the policy profile $\bpi_t$ into the simulator $\mathcal{G}$, and recursively generate a $\epsilon_{pop}$-stable population under the policy profile $\bpi_t$ with an approximate population update operator $\hat{\Gamma}_{pop}$ (which approximates $\Gamma_{pop}$ in Definition \ref{defn: total population update}). Then, the coordinator can choose from the standard Q-learning algorithm, offline learning algorithms (fitted Q-learning, see e.g.,  \cite{Anahtarci2019FittedQI}), or deep Q-learning algorithms (for example, DQN in \cite{mnih2013playing}, DDQN in \cite{DDQN-2015}, Mellowmax Q-Learning in \cite{asadi2017alternative}, and Neural Q-learning in \cite{cai2020neural}) to find a $q$-function approximation $\hat{q}^k_t$ for each population $k$. Note that the coordinator can also choose an offline algorithm such as Fitted Q-Learning as suggested in \cite{Anahtarci2019FittedQI}, and we comment that this choice depends on the size of state and action space, as well as the availability of memory. 

Finally, they will perform one-step of the total population Policy Mirror Ascent as defined in Definition 
\ref{PMA Operator on Stationary Population Distribution} to update the policy profile $\bpi_t$ to $\bpi_{t+1}$. The coordinator will then repeat this process until they obtain the $\epsilon_\pi$-stable policy profile.

First, we present the algorithm for the operator $\hat{\Gamma}_{pop}$ which approximates $\Gamma_{pop}$ and generates a $\epsilon_{pop}$-stable mean-field ensemble, and denote the element-wise approximation of $\Gamma_{pop}[k]$ by $\hat{\Gamma}^k_{pop}$. 

\begin{algorithm}
    \caption{Approximate Population Update Operator $\hat{\Gamma}_{pop}$}
    \label{Algo: Stable Population hat-Gamma_pop}
    \begin{algorithmic}
        \Require Simulator $\mathcal{G}$. 
        \Require Number of samples $N$, mean-field ensemble $\bmu$, policy profile $\bpi$.
        \For{Population $k = 1, \dots ,K$}
             \State Compute $$z^k = \frac{1}{K}\left({\sum_{l=1}^K W_{kl}\cdot \mu^l}\right).$$
            \State Generate a dataset of $N$ i.i.d. samples $\{(s_i,a_i, s'_i)\}_{i=1}^N$ via simulator $\mathcal{G}$: $s'_i \sim P(\cdot|s_i, a_i, z^k).$
            \State Define $$\hat{P}(s'|s,a,z^k) = \frac{\sum_{i=1}^N \one\{(s_i,a_i,s_i') = (s,a,s')\}}{N(s,a)}, N(s,a) = \sum_{i=1}^N \one\{(s_i,a_i) = (s,a)\}$$

            \State Set $\mu^k_{new}(\cdot) = \sum_{s\in \ss}\sum_{a\in \aa}\mu^k(s)\pi^k(a|s)\hat{P}(\cdot|s,a,z^k)$.
        \EndFor
        \State Return $\bmu_{new} = (\mu^1_{new}, \dots, \mu^K_{new})$.
       
    \end{algorithmic}
     \label{complete information case}
\end{algorithm}
Before introducing the main algorithm, we present an error analysis of the approximate population update operator $\hat{\Gamma}_{pop}.$ For simplicity, we assume that $N$ is large enough that $N(s,a) > 0, \forall s\in \ss, a\in \aa.$

\begin{lem}[Error Analysis of $\hat{\Gamma}_{pop}$]
\label{Error Analysis of Gamma_pop_hat}
For any $(\epsilon, \delta) \in (0,1)^2$, with probability at least $1-\delta$,
\begin{align*}
    \norm{\hat{\Gamma}_{pop}(\bmu, \bpi) - \Gamma_{pop}(\bmu,\bpi)}_1 \leq \epsilon,
\end{align*}
if the sample size $$N \geq \Tilde O(\frac{S^2A}{\epsilon^2}),$$ where $S=|\ss|$ and $A= |\aa|$. 
\end{lem}

\begin{proof}
    First, we bound the model error in Algorithm \ref{Algo: Stable Population hat-Gamma_pop}, i.e. $\norm{\hat{P}(\cdot|s,a) - P(\cdot|s,a)}_1$ for all state-action pair $s\in \ss$, $a \in \aa$. 

    Notice that for all $s,a$, $$\norm{\hat{P}(\cdot|s,a) - P(\cdot|s,a)}_1 = \max_{f:\ss \to [-1,1]}(\hat{P}(\cdot|s,a) - P(\cdot|s,a))^\top f.$$

    Then, for a fixed function $f$, Hoeffding's inequality (\cite{AJKS}, Lemma A.1.) implies that for a fixed $s,a$, with probability at least $\frac{\delta}{SA}$,
    \begin{align*}
        |(\hat{P}(\cdot|s,a) - P(\cdot|s,a))^\top f| \geq \sqrt{\frac{2\ln(2SA/\delta)}{N(s,a)}}.
    \end{align*}
     Then by union bound, for all $s,a$
    \begin{align*}
        \P(|(\hat{P}(\cdot|s,a) - P(\cdot|s,a))^\top f| &\leq \sqrt{\frac{2\ln(2SA/\delta)}{N(s,a)}})\\ &= 1 - \P\left(\bigcup_{s,a}\left\{\left|\left(\hat{P}(\cdot | s,a) - P^*(\cdot | s,a)\right)^\top f\right| \geq  \sqrt{\frac{2\ln(2SA/\delta)}{N(s,a)}})\right\}\right)\\
          &\geq 1 - \sum_{s,a} \P\left(\left|\left(\hat{P}(\cdot | s,a) - P^*(\cdot | s,a)\right)^\top f\right| \geq  \sqrt{\frac{2\ln(2SA/\delta)}{N(s,a)}}\right)\\
          &\geq 1 - \sum_{s,a}\frac{\delta}{SA}\\
          &= 1- \delta.
    \end{align*}

    We can then use a standard $\epsilon$-net argument and get a covering $N_\epsilon$ with $|N_\epsilon| < (1+2\sqrt{S}/\epsilon)^S$ such that for any $f \in [-1,1]^S$, there exists a $f' \in N_\epsilon$ such that $\norm{f-f'}_2 \leq \epsilon$. This implies that 
    \begin{align*}
        \left|\left(\hat{P}(\cdot | s,a) - P^*(\cdot | s,a)\right)^\top f\right| \leq 2\epsilon + \left| \left(\hat{P}(\cdot | s,a) - P^*(\cdot | s,a)\right)^\top f'\right|.
    \end{align*}
    Thus, by Hoeffding's inquality and union bound, it follows that for all $s,a$ and $f$, 
    \begin{align*}
        \P\left(\left|\left(\hat{P}(\cdot | s,a) - P^*(\cdot | s,a)\right)^\top f\right| \leq \sqrt{\frac{2S\ln(2SA(1+2\sqrt{S}/\epsilon)/\delta)}{N(s,a)}} +2\epsilon\right) \geq 1-\delta,
    \end{align*}

    Finally, take $\epsilon = \frac{2}{N}$, where $N$ is the total number of samples, which means that $N \geq N(s,a)$. Therefore, we obtain: with probability at least $1-\delta$, for all $s,a$
    \begin{align*}
        \norm{\hat{P}(\cdot | s,a) - P^*(\cdot | s,a) }_1 &= \max_{f \in [-1,1]^S}\left|(\hat{P}(\cdot | s,a) - P^*(\cdot | s,a))^\top f\right|\\
        &\leq \sqrt{\frac{2S\ln(2SA(1+2\sqrt{S}/\epsilon)/\delta)}{N(s,a)}} +\frac{4}{N}\\
        &\leq O(\sqrt{\frac{2S\ln(2SA(1+\sqrt{S}N)/\delta)}{N(s,a)}})
        = \Tilde{O}\left(\sqrt{\frac{S}{N(s,a)}}\right),
    \end{align*}
    where $O(\cdot)$ ignores absolute constant, $\Tilde O(\cdot)$ ignores the log term. 

    Having bounded the model error, now we proceed to analyze the estimation error:
    \begin{align*}
        \P\left(\norm{\hat{\Gamma}_{pop}(\bmu, \bpi) - \Gamma_{pop}(\bmu,\bpi)}_1 \leq \epsilon\right) &= \P\left(\norm{\hat{\Gamma}^k_{pop}(\bmu, \pi^k) - \Gamma_{pop}(\bmu,\pi^k)}_1 \leq \epsilon\right), \text{ where $k$ is the maximizer}\\
        &= \P\left(\norm{\sum_{s\in \ss}\sum_{a\in \aa}\mu^k(s)\pi^k(a|s)\left(\hat{P}(\cdot|s,a,z^k) - P(\cdot|s,a,z^k)\right)}_1 \leq \epsilon \right)\\
        \text{ by triangle inequality, }&\geq \P\left(\sum_{s\in \ss}\sum_{a\in \aa}\mu^k(s)\pi^k(a|s)\norm{\hat{P}(\cdot|s,a,z^k) - P(\cdot|s,a,z^k)}_1 \leq \epsilon \right)\\
        &\geq \P\left(\max_{s,a}\norm{\hat{P}(\cdot|s,a,z^k) - P(\cdot|s,a,z^k)}_1 \leq \epsilon \right),
    \end{align*}
    Using the model error bound, we can obtain that when $N(s,a) \geq \Tilde O(\frac{S}{\epsilon^2}),$
    \begin{align*}
         \P\left(\norm{\hat{\Gamma}_{pop}(\bmu, \bpi) - \Gamma_{pop}(\bmu,\bpi)}_1\leq \epsilon\right) &\geq 1-\delta,
    \end{align*}
    and in order to get the total sample complexity, we can assume that for all $s,a$, $N(s,a) \geq \Tilde O(\frac{S}{\epsilon}).$, so we have: 
    \begin{align*}
        N = \sum_{s\in \ss} \sum_{a\in \aa} N(s,a) \geq \Tilde O\left(\frac{S^2A}{\epsilon^2}\right),
    \end{align*}
    and this implies a polynomial sample complexity. 
\end{proof}

Recall that in the complete information case, we have shown that under the assumption of $L_{pop,\mu} < 1$, the total population update operator $\Gamma_{pop}$ is a contraction. Now, we would like the approximator $\hat{\Gamma}_{pop}$ to have the same contractive property. 

\begin{lem}[Contractive Property of $\hat{\Gamma}_{pop}$]
\label{lemma: contractive property of hat{Gamma}_pop}
    For a fixed policy profile $\bpi$, the approximate population update operator $\hat{\Gamma}_{pop}$ is a contraction, i.e. 
    \begin{align*}
        \norm{\hat{\Gamma}_{pop}(\bmu, \bpi) - \hat{\Gamma}_{pop}(\bmu', \bpi)}_1 \leq \norm{\bmu - \bmu'}_1,
    \end{align*}
    when the model error of $\epsilon$ as defined in Lemma \ref{Error Analysis of Gamma_pop_hat} satisfies
    \begin{align*}
        \epsilon < \frac{(1-L_{pop,\mu})\norm{\bmu - \bmu'}_1}{2}.
    \end{align*}
\end{lem}

\begin{proof}
By Lemma \ref{Lipschitz continuity of Gamma_pop},
\begin{align*}
    \norm{\hat{\Gamma}_{pop}(\bmu, \bpi) - \hat{\Gamma}_{pop}(\bmu', \bpi)}_1  \leq &\norm{\hat{\Gamma}_{pop}(\bmu, \bpi) - {\Gamma}_{pop}(\bmu, \bpi)}_1\\
    &+ \norm{{\Gamma}_{pop}(\bmu, \bpi) - {\Gamma}_{pop}(\bmu', \bpi)}_1 + \norm{{\Gamma}_{pop}(\bmu', \bpi) - \hat{\Gamma}_{pop}(\bmu', \bpi)}_1\\
    \leq &2 \epsilon + L_{pop,\mu}\norm{\bmu - \bmu'}_1,
\end{align*}
which implies the upper bound for the model error $\epsilon$. 
\end{proof}

Assuming the total population update is a contractive mapping, for a fixed policy profile $\bpi_t$, we can iterate $\hat{\Gamma}_{pop}$ until ``almost" stable i.e. $\norm{\bmu_{new}- \hat{\bmu}_{t}}_1 \leq \epsilon_{pop}$ (See Algorithm \ref{algo:Simulator-based PMA Reinforcement Learning for MP-MFG NE} for more details). Denote the number of iterations required by $T^t_{pop}$. Then, using the error bound for single-step population update, we can develop the following lemma which provides an error bound for the stable population generator $\Gamma_{pop}^\infty$: 

\begin{lem}[Error Analysis for $\Gamma^\infty_{pop}$]
\label{lemma: error analysis for gamma_infty_pop}
    The $\epsilon_{pop}$-approximate stable population $\hat{\bmu}_{T_{pop}^t}$ generated recursively by $\hat{\Gamma}_{pop}$ satisfies, with probability $1-\delta$, 
    \begin{align*}
        \norm{\Gamma_{pop}^\infty(\bpi) - \hat{\bmu}_T}_1 \leq \frac{\epsilon_{pop}}{1- L_{pop}} + \epsilon_{pop},
    \end{align*}
    when the sample sizes for Algorithm \ref{Algo: Stable Population hat-Gamma_pop} are sufficiently large. 
\end{lem}

\begin{proof}
For brevity, we use $T$ to replace $T^t_{pop}$, as the proof works for any policy profile $\bpi \in \Pi^K$. 

By Lemma \ref{Error Analysis of Gamma_pop_hat}, take all the sample sizes $N \geq \tilde O\left(\frac{4S^2A}{(1-L_{pop}\epsilon_{pop})^2}\right)$ so that with probability $1-\delta$, the model error $\epsilon < \frac{1}{2}{(1-L_{pop})\epsilon_{pop}}$, and then by Lemma \ref{lemma: contractive property of hat{Gamma}_pop} and the design of the algorithm, $\norm{\hat{\Gamma}_{pop}^\infty(\bpi) - \hat{\bmu}_{T}}_1 \leq \epsilon_{pop},$ so by Lemma \ref{Lipschitz continuity of Gamma_pop},
\begin{align*}
    \norm{\Gamma_{pop}^\infty(\bpi) - \hat{\bmu}_T}_1 &\leq \norm{\Gamma_{pop}^\infty(\bpi) - \hat{\Gamma}_{pop}^\infty(\bpi)}_1 + \norm{\hat{\Gamma}_{pop}^\infty(\bpi) - \hat{\bmu}_T}_1\\
    &\leq \norm{\Gamma^\infty_{pop}(\bpi) - \hat{\Gamma}_{pop}^\infty(\bpi)}_1 + \epsilon_{pop},
\end{align*}
where the first term is bounded by: 
\begin{align*}
    \norm{\Gamma^\infty_{pop}(\bpi) - \hat{\Gamma}_{pop}^\infty(\bpi)}_1 &= \norm{\Gamma_{pop}(\Gamma^\infty_{pop}(\bpi), \bpi) - \hat{\Gamma}_{pop}(\hat{\Gamma}^\infty_{pop}(\bpi),\bpi)}_1\\
    &\leq \norm{\Gamma_{pop}(\Gamma^\infty_{pop}(\bpi), \bpi) - \Gamma_{pop}(\hat{\Gamma}^\infty_{pop}(\bpi), \bpi)}_1 +  \norm{\Gamma_{pop}(\hat{\Gamma}^\infty_{pop}(\bpi), \bpi) - \hat{\Gamma}_{pop}(\hat{\Gamma}^\infty_{pop}(\bpi), \bpi)}_1\\
    &\leq L_{pop} \norm{\Gamma^\infty_{pop}(\bpi) - \hat{\Gamma}^\infty_{pop}(\bpi)}_1 + \epsilon_{pop},\\
    \Rightarrow \norm{\Gamma^\infty_{pop}(\bpi) - \hat{\Gamma}_{pop}^\infty(\bpi)}_1 &\leq \frac{\epsilon_{pop}}{1- L_{pop}}.
\end{align*}
This implies that: with probability $1-\delta$, 
\begin{align*}
    \norm{\Gamma_{pop}^\infty(\bpi) - \hat{\bmu}_T}_1 \leq \frac{\epsilon_{pop}}{1- L_{pop}} + \epsilon_{pop},
\end{align*}
which completes the proof.   
\end{proof}
Now, we present the learning algorithm based on a simulator (a generative model). 

\begin{algorithm}
\caption{Simulator-based PMA Reinforcement Learning for MP-MFG NE $\hat{\Gamma}_{\mathcal{G}}$}
\label{algo:Simulator-based PMA Reinforcement Learning for MP-MFG NE}
    \begin{algorithmic}[1]
    \Require Learning parameter $\eta$. Error tolerance $\epsilon_{pop}$, $\epsilon_{\pi}$, $\epsilon_q$. Connectivity matrix $W_K$. Reset mean-field distribution $\bmu_0$. 
    \Require Simulator $\mathcal{G}$. Initialize $\hat{\bpi}_0 = \bpi_0 $.
    \While {$\Delta_\pi > \epsilon_{\pi}$}
        
        \State \textbf{(Stable Population)} Set $\Delta_{pop} = 1.$
        \While{$\Delta_{pop} > \epsilon_{pop}$}
            \State Initialize mean-field ensemble $\hat{\bmu}_t = \bmu_0.$
            \State Update $\bmu_{new} = \hat{\Gamma}_{pop}(\hat{\bmu}_t,\hat{\bpi}_t)$ by Algorithm \ref{Algo: Stable Population hat-Gamma_pop}.
            \State Compute $\Delta_{pop} = \norm{\bmu_{new}- \hat{\bmu}_t}_1.$
            \State Set $\hat{\bmu}_t = \bmu_{new}$.  
        \EndWhile
        \State Compute $z^k_{t} = \frac{1}{K}\left({\sum_{l=1}^K W_{kl}\cdot \hat{\mu}_{t}^l}\right).$
        \State \textbf{(Policy Evaluation)} For all $k$, get an $\epsilon_q$-approximation $\hat{q}_t^k$ with a Q-learning algorithm.
        
        \State \textbf{(Policy Mirror Ascent)} Obtain $\hat{\bpi}_{t+1} = (\hat{\pi}^1_{t+1}, \dots, \hat{\pi}_{t+1}^K),$ where $\hat{\pi}_{t+1}^k = \Gamma_\eta^{pma}(\hat{q}_{t}^k, \hat{\pi}_t^k).$

        \State Compute $\Delta_{\pi} = \norm{\hat{\bpi}_{t+1} - \hat{\bpi}_t}_1$.
        \State $t \gets t+1$.
    \EndWhile
   
    \end{algorithmic}
     \label{Simulator case}
\end{algorithm}

In this algorithm, approximation is used in two parts: approximately generate the stable population mean-field ensemble as prescribed in Algorithm \ref{Algo: Stable Population hat-Gamma_pop}, and approximate the state-action value function $q^k$ with an appropriate $Q$-learning algorithm. Error propagates in each iteration of the algorithm so it is necessary to consider the following error analysis:

\begin{thm}[Convergence of Simulator-Based Algorithm]
\label{theorem: Convergence of Simulator-Based Algorithm}
    Suppose that the error in the Policy Evaluation step in Algorithm \ref{algo:Simulator-based PMA Reinforcement Learning for MP-MFG NE} satisfies: with probability $1-\delta$, for all $k \in [K]$, the $q$-function estimation is $\epsilon_q$ accurate, i.e., $\norm{\hat{\Gamma}_q[k](\pi^k, \bmu) - \Gamma_q[k](\pi^k,\bmu)}_\infty \leq \epsilon_{q}(\delta),$
    then the output $\bpi_T$ of Algorithm \ref{algo:Simulator-based PMA Reinforcement Learning for MP-MFG NE} satisfies with probability $1-\delta$,
    \begin{align*}
         \norm{\bpi_{T} - \bpi^*}_1 &\leq L_{md,q}\left(\frac{\epsilon_{pop}L_{q,\mu}}{1-L_{pop}} + \epsilon_{pop} + \epsilon_q\right) + 2L^{T}_{\Gamma_\eta}
    \end{align*}
\end{thm}

\begin{proof}

    Recall also that the complete information case PMA operator is defined as $\Gamma_\eta: \Pi^K \to \Pi^K$, and then for any $t = 0, \dots, T$, where $T$ denotes the number of iterations to obtain an $\epsilon_\pi$-approximate stable policy profile. Denote the $\epsilon_{pop}$-approximate stable policy generated under each policy profile $\bpi_t$ by $\hat{\bmu}^*_t$. By Lemma \ref{Lipschitz continuity of Gamma_pop}, Lemma \ref{lemma: Lipschitz Gamma qk}, Lemma \ref{lipschitz continuity of mirror descent q function generator}, Lemma \ref{Error Analysis of Gamma_pop_hat}, and Lemma \ref{lemma: error analysis for gamma_infty_pop}, with probability $1-\delta$,
    \begin{align*}
        \norm{\Gamma_\eta(\bpi_t) - \hat{\Gamma}_{\mathcal{G}}(\bpi_t)}_1 &= \max_{k \in [K]} \norm{\Gamma^k(\bpi_t) -  \hat{\Gamma}^k_{\mathcal{G}}(\bpi_t)}_1\\
        &= \norm{\Gamma^k(\bpi_t) -  \hat{\Gamma}^k_{\mathcal{G}}(\bpi_t)}_1, \text{ where k is the maximizer,}\\
        &= \norm{\Gamma_\eta^{pma}(\Gamma_q[k](\pi^k_t,\Gamma^\infty_{pop}(\bpi_t)),\pi^k) - \Gamma_\eta^{pma}(\hat{\Gamma}_q[k](\pi^k_t,\bmu^*_t),\pi^k)}_1\\
        &\leq L_{md,q}\norm{\Gamma_q[k](\pi^k,\Gamma^\infty_{pop}(\bpi))-\hat{\Gamma}_q[k](\pi^k,\bmu^*_t)}_\infty\\
        &\leq L_{md,q} \left(\norm{\Gamma_q[k](\pi^k,\Gamma^\infty_{pop}(\bpi))-{\Gamma}_q[k](\pi^k,\bmu^*_t)}_\infty + \norm{{\Gamma}_q[k](\pi^k,\bmu^*_t) - \hat{\Gamma}_q[k](\pi^k,\bmu^*_t)}_\infty\right)\\
        &\leq L_{md,q}\left(L_{q,\mu} \norm{{\Gamma}^\infty_{pop}(\bpi)- \bmu^*_t}_1 + \epsilon_q\right)\\
        &\leq L_{md,q}\left(\frac{\epsilon_{pop}L_{q,\mu}}{1-L_{pop}} + \epsilon_{pop} + \epsilon_q\right).
    \end{align*}
    Then, we consider the final output of the algorithm: with probability $1-\delta$, by Theorem \ref{theorem: Convergence to MP-MFG NE in the Exact Case},
    \begin{align*}
        \norm{\bpi_{T} - \bpi^*}_1 &\leq \norm{\Gamma_\eta(\bpi_{T}) - \Gamma_{\mathcal{G}}(\bpi_{T})}_1 + \norm{\Gamma_\eta(\bpi_{T}) - \bpi^*}_1\\
        &\leq L_{md,q}\left(\frac{\epsilon_{pop}L_{q,\mu}}{1-L_{pop}} + \epsilon_{pop} + \epsilon_q\right) + 2L^{T}_{\Gamma_\eta},
    \end{align*}
    which completes the proof. 
\end{proof}

\section{Sample-Based Learning with Finite Players (MP-MFG)}
\label{appendix: Sample-Based Learning with Finite Players (MP-MFG)}
In Section \ref{subsection: Solution to MP-MFG with Complete Information}, we have established that under mild technical constraints, the iterative scheme $\Gamma_\eta$ contracts to MP-MFG NE in the complete information case. Subsequently, as discussed in \ref{section: Sample-Based Learning with Finite Players (MP-MFG)}, it follows that the centralized PMA-CTD algorithm proposed by \cite{Yardim_PMA} can be extended to MP-MFG case assuming that $W_K$ is known. For brevity, we only include some key analyses which are relevant to our main analysis of GGR-S dynamics.

The first result we establish is a bound on the expected difference between the empirical mean-field ensemble and the mean-field ensemble updated via $\Gamma_{pop}$ with complete information.  

\begin{lem}[Empirical Population Bound in MP-MFG]
\label{lemma: empirical population bound}
Assume that at any time $t \geq 0$, each agent $i$ in population $k$ follows $a$ given (arbitrary) policy $\pi^k \in \Pi$ prescribed by the central learner, so that, 
\begin{align*}
    a_t^i \sim \pi^k(s_t^{k,l}), s_{t+1}^{k,l} \sim P(\cdot|s_t^{k,l}, a_t^{k,l}, \hat{z}^k_t), \forall t \geq 0, i =1,\dots, N.
\end{align*}
Let $\bpi \in \Pi^K$ denote an arbitrary policy profile. For all $\tau, t \geq 0$, it holds that: 
\begin{align*} 
    \ee\left[\norm{\hat{\bmu}_{t+\tau} - \Gamma^\tau_{pop}(\hat{\bmu}_t, \bpi)}_1\big| \ff_t\right] \leq \frac{1-L_{pop}^\tau}{1-L_{pop}}\cdot \sqrt{\frac{2|\ss|}{\min_i N_i}}.
\end{align*}
\end{lem}

\begin{proof}
    We start with $\tau = 1$, let $\ff_t := (\ff_t^1, \dots, \ff_t^K)$, where $\ff_t$ is defined in Definition \ref{defn: Locally Centralized Learning Model}. 
    
    We prove that the empirical population estimator is unbiased. Indeed, note that for any population $k$,
    \begin{align*}
        \ee[\hat{\mu}^k_{t+1}|\ff_t] &= \ee\left[\frac{1}{N_k}\sum_{l=1}^{N_k}\delta_{\hat s_{t+1}^{k,l}}\bigg|\ff_t\right] = \sum_{l=1}^{N_k} \frac{1}{N_k}\Bar{P}\left(\cdot\big|s_t^{k,l}, \pi^k(s_t^{k,l}),\hat{z}^{k}_t\right) = \Gamma_{pop}[k](\hat{\mu}_t^k,\pi^k).
    \end{align*}

   Then, we compute the variance at time step $t+1$. For any population $k$, 
   \begin{align*}
        \ee\left[\norm{\hat{\mu}_{t+1}^k - \ee[\hat{\mu}^k_{t+1} |\ff_t]}_2^2\big|\ff_t\right] = \frac{1}{N_k^2} \sum_{i=1}^{N_k} \ee\left[\norm{\delta_{s_{t+1}^{k,l}}- \ee\left[\delta_{s_{t+1}^{k,l}}\big|\ff_t\right]}_2^2\big|\ff_t \right]\leq \frac{2}{N_k}
   \end{align*}
   since $\norm{\delta_{s_{t+1}^{k,l}}- \ee\left[\delta_{s_{t+1}^{k,l}}\big|\ff_t\right]}_2^2 \leq 2$. Then, we get: 
   \begin{align*}
       \ee\left[\norm{\hat{\mu}_{t+1}^k - \ee[\hat{\mu}_{t+1}^k|\ff_t]}_1\big|\ff_t\right] &= \sqrt{\ee\left[\norm{\hat{\mu}_{t+1}^k - \ee[\hat{\mu}_{t+1}^k|\ff_t]}_1\right]^2} \leq \sqrt{(\sqrt{|\ss|} \ee\left[\norm{\hat{\mu}_{t+1}^k - \ee[\hat{\mu}_{t+1}^k|\ff_t]}_2)^2\right]} \leq \sqrt{\frac{2|\ss|}{N_k}},
   \end{align*}
   by Jensen's inequality and that $\norm{x}_1 \leq \sqrt{n}\norm{x}_2$ for all $x \in \R^n$. Thus, it follows that, $$\ee\left[\norm{\hat{\bmu}_{t+1} - \Gamma_{pop}(\hat{\bmu}_t, \bpi)}_1|\ff_t\right] = \ee\left[\max_k \norm{\hat{\mu}_{t+1}^k - \ee[\hat{\mu}_{t+1}^k|\ff_t]}_1 |\ff_t\right] \leq \sqrt{\frac{2|\ss|}{\min_i N_i}}.$$

   When $\tau > 1,$ note that at $\tau +1$,
   \begin{align*}
       \ee\left[\norm{\hat{\bmu}_{t+\tau} - \Gamma^{\tau}_{pop}(\hat{\bmu}_t, \bpi)}_1|\ff_t \right] \leq &\ee\left[\norm{\hat{\bmu}_{t+\tau} - \Gamma_{pop}(\hat{\bmu}_{t+\tau-1}, \bpi)}_1|\ff_t \right]\\
       &+ \ee\left[\norm{\Gamma_{pop}(\hat{\bmu}_{t+\tau-1}, \bpi) - \Gamma^{\tau}_{pop}(\hat{\bmu}_{t}, \bpi)}_1|\ff_t \right]\\
       \leq &\ee\left[\ee\left[\norm{\hat{\bmu}_{t+\tau} - \Gamma_{pop}(\hat{\bmu}_{t+\tau-1}, \bpi)}_1|\ff_{t+\tau-1} \right]|\ff_t\right]\\
       &+ \ee\left[\norm{\Gamma_{pop}(\hat{\bmu}_{t+\tau-1}, \bpi) - \Gamma^{\tau}_{pop}(\hat{\bmu}_{t}, \bpi)}_1|\ff_t \right]\\
       \leq &\sqrt{\frac{2|\ss|}{\min_i N_i}} + L_{pop}\cdot \ee\left[\norm{\hat{\bmu}_{t+\tau-1} - \Gamma_{pop}^{\tau-1}(\hat{\bmu}_t,\bpi)}_1|\ff_t\right],
   \end{align*}
   where the second to last equality is by tower property of conditional expectation and the fact that $\ff_\tau \supset \ff_{t}$, and the last inequality is due to Lemma \ref{Lipschitz continuity of Gamma_pop}. Then, by recursion on the second term of the last inequality, we can obtain: 
   \begin{align*}
          \ee\left[\norm{\hat{\bmu}_{t+\tau} - \Gamma^{\tau}_{pop}(\hat{\bmu}_t, \bpi)}_1 \big| \ff_t \right] \leq \frac{1-L_{pop}^\tau}{1-L_{pop}}\cdot \sqrt{\frac{2|\ss|}{\min_i N_i}},
   \end{align*}
   which concludes the proof. 
\end{proof}

It follows immediately from Lemma \ref{lemma: empirical population bound} that the distance between the empirical population aggregated impact and the theoretical population aggregated impact is bounded in expectation, since the aggregated impact are linear combinations of elements of $\bmu$ and $\hat{\bmu}$. 
\begin{corollary}[Empirical Aggregated Impact Bound]
\label{corollary: empirical population impact bound}
Denote by $\bmu_{t+\tau} = \Gamma^\tau(\hat{\bmu}_t,\bpi)$ the mean-field ensemble at time $t+\tau$ in the complete information case, i.e. updated via $\Gamma_{pop}$. For all $k\in [K]$ and $t,\tau > 0$, consider $z^k_{t+\tau}$ as defined in Eq. \eqref{eq:z_def} with respect to $\bmu_{t+\tau}$ and $\hat{z}^{k}_{t+\tau}$ as defined in Eq. \eqref{eq: hat_z-def}. Then
\begin{align*}
    \ee[\norm{\hat{z}_{t+\tau} - z_{t+\tau}}_1|\ff_t] \leq \frac{p_*(1-L_{pop}^\tau)}{1-L_{pop}}\cdot \sqrt{\frac{2|\ss|}{\min_i N_i}}.
\end{align*}
\end{corollary}
\begin{proof}
Note that $p_* = \max_{k,i} W_K(k,i)$, so
    \begin{align*}
         \ee\left[\norm{\hat{z}_{t+\tau} - z_{t+\tau}}_1|\ff_t\right] &=  \ee\left[\norm{\frac{1}{K}\sum_{i=1}^K W_K(k,i)\cdot \hat{\mu}^i_{t+\tau} - \frac{1}{K}\sum_{i=1}^K W_K(k,i)\cdot \mu^i_{t+\tau}}_1\bigg|\ff_t\right]\\
         &\leq \ee\left[\norm{\frac{1}{K}\sum_{i=1}^K p_*\cdot \hat{\mu}^i_{t+\tau} - \frac{1}{K}\sum_{i=1}^K p_*\cdot \mu^i_{t+\tau}}_1\bigg|\ff_t\right]\\
         &\leq p_*  \ee\left[\frac{1}{K}\sum_{i=1}^K\norm{ \hat{\mu}^i_{t+\tau} - \mu^i_{t+\tau}}_1\bigg|\ff_t\right]\\
         &\leq p_*  \ee\left[\norm{ \hat{\bmu}_{t+\tau} - \bmu_{t+\tau}}_1\bigg|\ff_t\right],
    \end{align*}
    and then the result follow from Lemma \ref{lemma: empirical population bound}. 
\end{proof}

With Lemma \ref{lemma: empirical population bound}, we can prescribe an expected distance between the empirical trajectory roll-out and the stable population with respect to a policy over time. 

\begin{corollary}[Convergence to stable mean-field]
Under the assumptions of Lemma \ref{lemma: empirical population bound}, and let $\bmu_\infty =  \Gamma_{pop}^\infty(\bpi)$ denote the limiting stable mean-field ensemble and let $z^k_\infty$ denote the aggregated impact for population $k$ as defined in Eq.~\eqref{eq:z_def} with respect to the stable mean-field ensemble $\bmu_\infty$; also $\hz^k_{t+\tau}$ is the empirical aggregated impact as defined in Eq. \eqref{eq: hat_z-def} with respect to $\hat \bmu_{t+\tau}$. Then, for any $t,\tau \geq 0,$ we have 
\begin{align*}
    &\ee\left[\norm{\hat{\bmu}_{t+\tau} - \bmu_\infty}_1 \big| \ff_t \right] \leq \frac{1}{1-L_{pop}}\cdot \sqrt{\frac{2|\ss|}{\min_i N_i}} + 2L_{pop}^\tau,\\
    &\ee\left[\norm{\hz^k_{t+\tau} - z^k_\infty}_1 \big| \ff_t \right] \leq \frac{p_*}{1-L_{pop}}\cdot \sqrt{\frac{2|\ss|}{\min_i N_i}} + 2L_{pop}^\tau
\end{align*}
\end{corollary}
\begin{proof}
    The corollary follows from triangle inequality and Lemma \ref{lemma: empirical population bound}.
\end{proof}

Additionally, given the information about the state of a single agent at time $t+\tau$, we can also provide an empirical population bound in terms of the mixing probability in Assumption \ref{assumption: reachability}. 
\begin{corollary}
    Given Lemma \ref{lemma: empirical population bound}, for any state $s^* \in \ss$, $\tau > T_{mix}$, and any agent $(k,l)$, we have: 
    \begin{align*}
        \ee\left[\norm{\hat{\bmu}_{t+\tau} - \Gamma^\tau_{pop}(\hat{\bmu}_t, \bpi)}_1\big|\hat{s}^{k,l}_{t+\tau} = s^*, \ff_t\right] \leq \frac{1-L_{pop}^\tau}{(1-L_{pop})\cdot \delta_{mix}}\cdot \sqrt{\frac{2|\ss|}{\min_i N_i}},
    \end{align*}
    where $\delta_{mix}$ is defined in Assumption \ref{assumption: reachability}.
\end{corollary}
\begin{proof}
    Note that 
    \begin{align*}
        \ee\left[\norm{\hat{\bmu}_{t+\tau} - \Gamma^\tau_{pop}(\hat{\bmu}_t, \bpi)}_1 \big|\ff_t\right] = \sum_{s\in \ss} \ee\left[\norm{\hat{\bmu}_{t+\tau} - \Gamma^\tau_{pop}(\hat{\bmu}_t, \bpi)}_1\big|\hat{s}^{k,l}_{t+\tau} = s, \ff_t\right]\P(\hat{s}^{k,l}_{t+\tau} = s|\ff_t),
    \end{align*}
    and the result follows. 
\end{proof}

Readers may check that with some additional computations of the learning parameters, the centralized PMA-CTD algorithm proposed by \cite{Yardim_PMA} can be extended to the fully connected MP-MFG case to achieve efficient learning with a finite sample complexity.

\section{Analysis of GGR-S Dynamics}
In GGR-S, at any time $t$, each player sense a different empirical neighbor impact; in particular, we may view individual neighbor impact $\tilde z^{k,l}_t$ as a stochastic perturbed version of the aggregated impact $\hat z^k_t$. Therefore, we first need to ensure that after sufficient mixing time, the empirical neighbor impact (GGR-S) does not deviate too much from the empirical population impact (MP-MFG).

Before presenting the proofs, we re-iterate the definitions of dynamics history information in GGR-S and MP-MFG. Let $\ff_t$ be the $\sigma$-algebra of state, action, and reward in GGR-S, $\ff_\tau := \ff\left(\{\ts^{k,l}_t, a_t^{k,l}, r_t^{k,l}\}_{t=1}^{\tau}, k\in [K], l \in [N_k]\right)$; and $\hat\ff_t$ the $\sigma$-algebra in MP-MFG, $\hat\ff_\tau := \ff\left(\{\hs^{k,l}_t, a_t^{k,l}, r_t^{k,l}\}_{t=1}^{\tau}, k\in [K], \ell \in [N_k]\right)$.

\subsection{Proof for Lemma \ref{lemma: One-Step Error Propagation Through Aggregated Impact}}
\label{proof for lemma: One-Step Error Propagation Through Aggregated Impact}
\begin{proof}
\label{appendix proof for lemma: One-Step Error Propagation Through Aggregated Impact}
First, let $\{\tilde s_{t}^{i,j}\}$ denote the states of agents in GGR-S, and let $\{\hat s_{t}^{i,j}\}$ denote the states of agents in MP-MFG. 

Note that
     \begin{align*}
        \ee\left[\tilde z^{k,l}_{t}\big|\ff_{t}\right] &= \frac{1}{K}\sum_{i=1}^K \frac{1}{N_i}\sum_{j=1}^{N_i}\ee\left[W_{t}^{[N]}((k,l),(i,j))\right]\cdot \ee\left[\delta_{\tilde s_{t}^{i,j}} \big|\ff_{t}\right]\\
        &= \frac{1}{K}\sum_{i=1}^K p_{k,i} \cdot \frac{1}{N_i}\sum_{j=1}^{N_i} \delta_{\tilde s_{t}^{i,j}} = \frac{1}{K}\sum_{i=1}^K p_{k,i}\cdot \tilde \mu^i_t
    \end{align*}
and similarly, 
    \begin{align*}
         \ee\left[\hat z^{k,l}_{t}\big|\hat\ff_{t}\right] &= \ee\left[\frac{1}{K}\sum_{i=1}^K \frac{p_{k,i}}{N_i}\sum_{j=1}^{N_i}\delta_{s^{i,j}_t}\big|\ff_{t}\right] = \frac{1}{K}\sum_{i=1}^K  {p_{k,i}}\cdot \hat \mu^k_t
    \end{align*}
and for all $k$,
\begin{align*}
     \norm{\ee\left[\tilde z^{k,l}_{t}\big|\ff_{t}\right] - \ee\left[\hat{z}^k_{t}\big|\hat\ff_{t}\right]}_1 &= \norm{ \frac{1}{K}\sum_{i=1}^K \sum_{s \in \ss}\mathbf{e}_s (p_{k,i}\cdot \tilde \mu^i_t(s)) - \frac{1}{K}\sum_{i=1}^K  \sum_{s \in \ss}\mathbf{e}_s ({p_{k,i}}\cdot \hat \mu^k_t(s))}_1 \\
     &\leq \frac{1}{K}\sum_{i=1}^K p_{k,i}\cdot \norm{\sum_{s\in \ss} \mathbf{e}_s(\tilde \mu^i_t(s) - \hat \mu^i_t(s))}_1\\
     &\leq p_{k}^*\cdot \norm{\tilde \bmu_t - \hat \bmu_t}_1,
\end{align*}
where $p_{k}^* = \max_{i} p_{k,i}$. 
Then, by independence of state transition of each agent, we can also decompose and bound the $l_2$-variance of $\tilde z$: 
    \begin{align*}
         \ee\left[\norm{\tilde z_{t}^{k,l} - \ee\left[\tilde z_{t}^{k,l}\big|\ff_{t}\right]}_2^2\big|\ff_t\right] &=   \frac{1}{K^2}\sum_{i=1}^K \frac{1}{N_i^2}\sum_{l=1}^{N_k}\ee\left[\norm{W_{t}^{[N]}((k,l),(i,j))\cdot \mathbf{e}_{s^{k,l}_{t}}- p_{k,i}\cdot \ee\left[\mathbf{e}_{s^{k,l}_{t}}\big|\ff_{t}\right]}_2^2\big|\ff_t\right]\\
         &\leq \frac{2}{K \cdot \min_i{N_i}},
    \end{align*}
    since $W_{t}^{[N]}((k,l),(i,j)),p_{k,i} \in [0,1]$ and $\norm{\mathbf{e}_{s^{k,l}_{t}}- \ee[\mathbf{e}_{s^{k,l}_{t}}]}_2^2 \leq 2.$ Since $\norm{\cdot}_1 \leq |\ss|\norm{\cdot}_2$, it then means: 
    \begin{align*}
        \ee\left[\norm{\tilde z_{t}^{k,l} - \ee\left[\tilde z_{t}^{k,l}\big|\ff_{t}\right]}_1\big|\ff_{t}\right] \leq \sqrt{\frac{2|\ss|}{K \cdot \min_i{N_i}}},
    \end{align*}
    and similarly,
    \begin{align*}
        \ee\left[\norm{\hat z_{t}^{k,l} - \ee\left[\hat z_{t}^{k,l}\big|\hat \ff_t\right]}_1\big|\hat\ff_{t}\right] \leq \sqrt{\frac{2|\ss|}{K \cdot \min_i{N_i}}}.
    \end{align*}
    Putting everything together, we obtain by triangle inequality: for all $k$,
   \begin{align*}
         \ee\left[\norm{\tilde z^{k,l}_{t} - \hat{z}^k_{t}}_1\big|\ff_{t},\hat\ff_t\right] &\leq
         \ee\left[ \norm{\tilde z^{k,l}_{t} - \ee[\tilde z^{k,l}_{t}|\ff_{t}] + \ee[\tilde z^{k,l}_{t}|\ff_{t}] - \ee[\hat z^k_{t}|\hat\ff_{t}] + \ee[\hat z^k_{t}|\hat\ff_{t}] - \hat{z}^k_{t}}_1\big|\ff_{t},\hat\ff_t \right]\\
         &\leq  \ee\left[ \norm{\tilde z^{k,l}_{t} - \ee[\tilde z^{k,l}_{t}|\ff_t]}_1 + \norm{\ee[\tilde z^{k,l}_{t}|\ff_t] - \ee[\hat z^k_{t}|\hat\ff_t]}_1 + \norm{\ee[\hz^k_{t}|\hat\ff_t] - \hat{z}^k_{t}}_1\big|\ff_t,\hat\ff_t \right]\\
         &\leq 2\sqrt{\frac{2|\ss|}{K\min_i{N_i}}} + p_{k}^*\cdot \norm{\tilde \bmu_t - \hat \bmu_t}_1.
    \end{align*}
    Next, we bound the difference between $\tilde \mu^k_{t+1}$ and $\hat \mu^k_{t+1}$. Here, note that $\tz_t^{k,l}$ is random under $\ff_t$ while $\hat{z}_t$ is not. Therefore, let $\nu(\tz|s)$ denote the conditional probability of $\tz \in \zz_N$ given the state $s$. Then, it follows from law of total probability that  
    \begin{align*}
        \norm{\ee[\tilde \mu^k_{t+1}|\ff_t] - \ee[\hat \mu^k_{t+1}|\hat \ff_t]}_1 &= 
         \norm{\sum_{s\in \ss, \tilde z \in \zz_N}\Bar{P}\left(\cdot\big|s, \pi^k(s),\tilde z\right)\nu(\tilde z | s)\tilde \mu_t^k(s) - \sum_{s\in \ss}\Bar{P}\left(\cdot\big| s, \pi^k(s),\hat{z}^{k}_t\right)\hat \mu_t^k(s)}_1\\
        &\leq \norm{\sum_{s\in \ss, \tilde z \in \zz_N}\Bar{P}\left(\cdot\big|s, \pi^k(s),\tilde z\right)\nu(\tilde z | s)\tilde \mu_t^k(s) - \sum_{s\in \ss}\Bar{P}\left(\cdot\big| s, \pi^k(s),\hat{z}^{k}_t\right)\tilde \mu_t^k(s)}_1\\
        &+ \norm{\sum_{s\in \ss}\Bar{P}\left(\cdot\big|s, \pi^k(s),\hat{z}^{k}_t\right)\tilde \mu_t^k(s) - \sum_{s\in \ss}\Bar{P}\left(\cdot\big| s, \pi^k(s),\hat{z}^{k}_t\right)\hat \mu_t^k(s)}_1,
    \end{align*}
    where the first term can be bounded via triangle inequality: 
   \begin{align*}
        &\norm{\sum_{s\in \ss, \tilde z \in \zz_N}\Bar{P}\left(\cdot\big|s, \pi^k(s),\tilde{z}\right) \nu(\tilde z | s) \tilde \mu_t^k(s) - \sum_{s\in \ss}\Bar{P}\left(\cdot\big| s, \pi^k(s),\hat{z}^{k}_t\right) \tilde \mu_t^k(s)}_1 \\  &\leq \sum_{s\in \ss}\tilde \mu_t^k(s)\norm{\sum_{\tilde z \in \zz_N}\Bar{P}\left(\cdot\big|s, \pi^k(s),\tilde z\right) \nu(\tilde z | s)  - \Bar{P}\left(\cdot\big|s, \pi^k(s),\hat{z}^{k}_t\right)}_1\\
        &\leq  p_\mu \max_{l \in [N_k]}  \norm{\tilde{z}^{k,l}_t - \hat{z}^{k}_t}_1, 
    \end{align*}
    by the fact that $\ee[z] \leq \max z$ and Lemma \ref{Lipschitz Continuity of Pbar Rbar}. 
    
    The second term can be bounded, using Lemma \ref{Yardim Lemma 9}:
    \begin{align*}
        \norm{\sum_{s\in \ss}\Bar{P}\left(\cdot\big|s, \pi^k(s),\hat{z}^{k}_t\right)\tilde \mu_t^k(s) - \sum_{s\in \ss}\Bar{P}\left(\cdot\big| s, \pi^k(s),\hat{z}^{k}_t\right)\hat \mu_t^k(s)}_1 \leq  \frac{\lambda}{2}\norm{\tilde \mu_t^k - \hat \mu_t^k}_1,
    \end{align*}
    with 
    \begin{align*}
        \frac{\lambda}{2} = \frac{\sup_{s,s'\in \ss}\norm{\Bar{P}\left(\cdot\big|s, \pi^k(s),\hat{z}^{k}_t\right) - \Bar{P}\left(\cdot\big|s', \pi^k(s'),\hat{z}^{k}_t\right)}_1}{2} &\leq \frac{1}{2}\left(p_sd(s,s') + p_a\norm{\pi^k(s) - \pi^k(s')}_1\right) \leq \frac{p_s+2p_a}{2},
    \end{align*}
    since $\norm{\pi^k(s) - \pi^k(s')}_1 \leq 2$ by triangle inequality.  

    Thus, putting everything together, we obtain: 
    \begin{align*}
        \norm{\ee[\tilde \mu^k_{t+1}|\ff_t] - \ee[\hat \mu^k_{t+1}|\hat \ff_t]}_1 &\leq p_\mu \max_{l \in [N_k]} \norm{\tilde{z}^{k,l}_t - \hat{z}^{k}_t}_1 + \frac{p_s+2p_a}{2}\norm{\tilde \mu_t^k - \hat \mu_t^k}_1.
    \end{align*}
    Again, by independence of state evolution of each agent, we can decompose the $l_2$-variance: 
    \begin{align*}
        \ee\left[\norm{\tilde \mu_{t+1}^k - \ee\left[\tilde \mu_{t+1}^k\big|\ff_t\right]}_2^2\big|\ff_t\right] = \frac{1}{N_k^2} \sum_{l=1}^{N_k}\ee\left[\norm{\mathbf{e}_{s^{k,l}_{t+1}}- \ee[\mathbf{e}_{s^{k,l}_{t+1}}|\ff_t]}_2^2\big|\ff_t\right] \leq \frac{2}{N_k},
    \end{align*}
    and similarly, 
   $\ee\left[\norm{\hat \mu_{t+1}^k - \ee\left[\hat \mu_{t+1}^k\big|\ff_t\right]}_2^2\big|\ff_t\right] \leq \frac{2}{N_k}.$
    By Jensen's inequality and the fact that $\norm{\cdot}_1 \leq \sqrt{\ss}\norm{\cdot}_2$, we obtain, for all $k$: 
    \begin{align*}
        \ee\left[\norm{\tilde \mu_{t+1}^k - \ee\left[\tilde \mu_{t+1}^k\big|\ff_t\right]}_1\big|\ff_t\right]  \leq \sqrt{\frac{2|\ss|}{\min_i N_i}}, \quad \ee\left[\norm{\hat \mu_{t+1}^k - \ee\left[\hat \mu_{t+1}^k\big|\hat\ff_t\right]}_1\big|\hat\ff_t\right] \leq \sqrt{\frac{2|\ss|}{\min_i N_i}}.
    \end{align*}
    Then, we may bound, via triangle inequality: 
    \begin{align*}
         \ee\left[\norm{\tilde \mu^k_{t+1} - \hat{\mu}^k_{t+1}}_1\big|\ff_t,\hat\ff_t\right] &\leq
         \ee\left[ \norm{\tilde \mu^k_{t+1} - \ee[\tilde \mu^k_{t+1}|\ff_t] + \ee[\tilde \mu^k_{t+1}|\ff_t] - \ee[\hat \mu^k_{t+1}|\hat \ff_t] + \ee[\hat \mu^k_{t+1}|\hat \ff_t] - \hat{\mu}^k_{t+1}}_1 \big|\ff_t,\hat\ff_t \right]\\
         &\leq \ee\left[\norm{\tilde \mu^k_{t+1} - \ee[\tilde \mu^k_{t+1}|\ff_t]}_1\big|\ff_t,\hat\ff_t\right]+\ee\left[\norm{\ee[\tilde \mu^k_{t+1}|\ff_t] - \ee[\hat \mu^k_{t+1}|\hat \ff_t]}_1\big|\ff_t,\hat\ff_t\right]\\
         &+\ee\left[\norm{\hat \mu^k_{t+1} - \ee[\hat \mu^k_{t+1}|\hat \ff_t]}_1\big|\ff_t,\hat\ff_t\right]\\
         &\leq 2\sqrt{\frac{2|\ss|}{\min_i N_i}} + \ee\left[ p_\mu \max_{l \in [N_k]}\norm{\tilde{z}^{k,l}_t - \hat{z}^{k}_t}_1 + \frac{p_s+2p_a}{2}\norm{\tilde \mu_t^k - \hat \mu_t^k}_1\big|\ff_t,\hat\ff_t\right]\\
         &\leq 2\sqrt{\frac{2|\ss|}{\min_i N_i}} + 2p_\mu \sqrt{\frac{2|\ss|}{K\min_i{N_i}}} + (p_*p_\mu + \frac{1}{2}p_s + p_a) \norm{\tilde \bmu_t - \hat \bmu_t}_1\\
         &\leq (2+2p_\mu)\sqrt{\frac{2|\ss|}{\min_i N_i}} + (p_*p_\mu + \frac{1}{2}p_s + p_a) \norm{\tilde \bmu_t - \hat \bmu_t}_1,
    \end{align*}
    with $p_* = \max_k p_k^*$.
\end{proof}

\subsection{Proof for Proposition \ref{prop: Multi-Step Error Propagation Bound}}
\label{appendix: proof for prop: Multi-Step Error Propagation Bound}
\begin{proof}
    First, note that it follows from Lemma \ref{lemma: One-Step Error Propagation Through Aggregated Impact} that
    \begin{align*}
         \norm{\E[\tilde \mu^k_{t+\tau+1}|\ff_{t+\tau}] - \E[\hat \mu^k_{t+\tau+1}|\ff_{t+\tau}]}_1 \leq p_\mu \max_{l \in [N_k]} m\norm{\tilde{z}^{k,l}_{t+\tau} - \hat{z}^{k}_{t+\tau}}_1 + \frac{p_s+2p_a}{2}\norm{\tilde \mu_{t+\tau}^k - \hat \mu_{t+\tau}^k}_1,
    \end{align*}
    so, since $\ff_{t} \subset \ff_{t+\tau}$, by tower property, 
    \begin{align*}
        \E\left[\norm{\tilde \mu^k_{t+\tau+1} - \hat{\mu}^k_{t+\tau+1}}_1\big|\ff_{t}\right] &\leq   \E\left[\E\left[\norm{\tilde \mu^k_{t+\tau+1} - \hat{\mu}^k_{t+\tau+1}}_1\big|\ff_{t+\tau}\right]\big|\ff_t\right]\\
          &\leq 2(1+p_\mu)\sqrt{\frac{2|\ss|}{\min_i N_i}} + \tilde L_{pop} \E\left[\norm{\tilde \mu_{t+\tau}^k - \hat \mu_{t+\tau}^k}_1\big|\ff_{t}\right]\\
             &\leq \cdots, \text{recursion}\\
          &\leq \sum_{m=0}^{\tau} \tilde L_{pop}^m\left(2\sqrt{\frac{2|\ss|}{\min_i N_i}}
         + p_\mu \cdot 2\sqrt{\frac{2|\ss|}{K\cdot \min_i{N_i}}}\right) + \tilde L_{pop}^{\tau+1} \cdot  \E\left[\norm{\tilde \bmu_t - \hat \bmu_t}_1\big|\ff_{t}\right]\\
         &\leq (2+2p_\mu)\sqrt{\frac{2|\ss|}{\min_i N_i}} \cdot \sum_{m=0}^{\tau} \tilde L_{pop}^m + \tilde L_{pop}^{\tau+1} \cdot  \norm{\tilde \bmu_t - \hat \bmu_t}_1\\
         &\leq \frac{(2+2p_\mu)(1- \Tilde{L}_{pop}^{\tau + 1})}{1-\tilde L_{pop}}\sqrt{\frac{2|\ss|}{\min_i N_i}} + \tilde L_{pop}^{\tau+1} \cdot  \norm{\tilde \bmu_t - \hat \bmu_t}_1.
    \end{align*}
    Since the above inequality holds for all population $k$, we can conclude that
    \begin{align*}
        \E\left[\norm{\tilde \bmu_{t+\tau} - \hat{\bmu}_{t+\tau}}_1\big|\ff_{t}\right] \leq \frac{(2+2p_\mu)(1- \Tilde{L}_{pop}^{\tau })}{1-\tilde L_{pop}}\sqrt{\frac{2|\ss|}{\min_i N_i}} + \tilde L_{pop}^{\tau} \cdot  \norm{\tilde \bmu_t - \hat \bmu_t}_1,
    \end{align*}
    as desired.     
\end{proof}

\subsection{Proof for Lemma \ref{lemma: Empirical Population Bound in the Locally Centralized Case} and Its Corollaries}
\label{appendix: proof for Empirical Population Bound in GGR-S}
\begin{proof}
    Suppose at any $t \geq 0$, $\ff_t = \hat\ff_t $, which implies that $\norm{\tilde \bmu_t - \hat{\bmu}_t}_1 = 0.$ Thus, it follows by triangle inequality that
    \begin{align*}
        \E\left[\norm{\tilde \bmu_{t+\tau} - \Gamma_{pop}^\tau (\tilde \mu_t, \bpi)}_1\big|\ff_t\right] &\leq \E\left[\norm{\tilde \bmu_{t+\tau} - \hat{\bmu}_{t+\tau}}_1 +\norm{\hat{\bmu}_{t+\tau} - \Gamma_{pop}^\tau (\tilde \mu_t, \bpi)}_1\big|\ff_t\right]\\
        &\leq \frac{1-\tilde L_{pop}^{\tau}}{1-\tilde L_{pop}} \left(2\sqrt{\frac{2|\ss|}{\min_i N_i}}
         + p_\mu \cdot 2\sqrt{\frac{2|\ss|}{K\cdot \min_i{N_i}}}\right)  + \frac{1-L_{pop}^\tau}{1-L_{pop}}\cdot \sqrt{\frac{2|\ss|}{\min_i N_i}},
    \end{align*}
    by Lemma \ref{lemma: empirical population bound} and Proposition \ref{prop: Multi-Step Error Propagation Bound}. The second inequality in the Lemma follows in the same manner from triangle inequality, Lemma \ref{lemma: One-Step Error Propagation Through Aggregated Impact}, and Lemma \ref{corollary: empirical population impact bound}.
\end{proof}

Lemma \ref{lemma: Empirical Population Bound in the Locally Centralized Case} has immediate consequences on the convergence of the empirical mean-field ensemble to the stable mean-field ensemble, as well as a bound on the distance between the empirical neighbor impact and the aggregated impact in the complete information case. 
    
\begin{corollary}[Convergence to Stable Mean-Field in GGR-S]
\label{corollary: Convergence to Stable Mean-Field in the Locally Centralized Model}
    Under Lemma \ref{lemma: Empirical Population Bound in the Locally Centralized Case}, for any $t,\tau \geq 0$, it holds that
    \begin{align*}
        &\E\left[\norm{\tilde{\bmu}_{t+\tau} - \Gamma^\infty_{pop}(\bpi)}_1 \big|\ff_t\right] \leq \frac{(3+2p_\mu)(1- \tilde L_{pop}^{\tau})}{1- L_{pop}} \sqrt{\frac{2|\ss|}{\min_i N_i}} + 2L_{pop}^\tau.
    \end{align*}
\end{corollary}

\begin{proof}
    The proof follows from triangle inequality and Lemma \ref{lemma: Empirical Population Bound in the Locally Centralized Case}. In particular, 
    \begin{align*}
        \E\left[\norm{\tilde{\bmu}_{t+\tau} - \Gamma^\infty_{pop}(\bpi)}_1 \big|\ff_t\right] \leq \E\left[\norm{\tilde{\bmu}_{t+\tau} - \Gamma^\tau_{pop}(\tilde \bmu_t, \bpi)}_1 \big|\ff_t\right] + \E\left[\norm{ \Gamma^\tau_{pop}(\tilde \bmu_t, \bpi) - \Gamma^\infty_{pop}(\bpi)}_1 \big|\ff_t\right],
    \end{align*}
    and in particular, since $\ff_t \subset \ff_{t+\tau-1}$ for all $\tau \geq 1$,
    \begin{align*}
        \E\left[\norm{ \Gamma^\tau_{pop}(\tilde \bmu_t, \bpi) - \Gamma^\infty_{pop}( \bpi)}_1 \big|\ff_t\right] &= \E\left[\norm{ \Gamma_{pop}(\Gamma^{\tau-1}_{pop}(\tilde \bmu_t, \bpi),\bpi) - \Gamma_{pop}(\Gamma^\infty_{pop}(\bpi),\bpi)}_1 \big|\ff_t\right]\\
        &\leq L_{pop}\cdot \E\left[\E\left[\norm{ \Gamma^{\tau-1}_{pop}(\tilde \bmu_t, \bpi) - \Gamma^\infty_{pop}(\bpi)}_1\big|\ff_{t+\tau-1}\right] \big|\ff_t\right] \text{ by Lemma \ref{Lipschitz continuity of Gamma_pop}},\\
        &\cdots, \text{by recursion}\\
        &\leq L^\tau_{pop}\cdot \E\left[\norm{ \tilde \bmu_t - \Gamma^\infty_{pop}(\bpi)}_1 \big|\ff_t\right]\\
        &\leq 2L^\tau_{pop},
    \end{align*}
    and the results follows from Lemma \ref{lemma: Empirical Population Bound in the Locally Centralized Case}.
\end{proof}

\begin{corollary}[Convergence of Empirical Neighbor Impact to Stable Aggregated Impact]
\label{corollary: Convergence of Empirical Neighbor Impact to Stable Aggregated Impact}
    Under Lemma \ref{lemma: Empirical Population Bound in the Locally Centralized Case}, let $\bmu_\infty =  \Gamma_{pop}^\infty(\bpi)$ denote the limiting stable mean-field ensemble and let $z^k_\infty = z^k_{t+\tau} = \frac{1}{K}\sum_{i=1}^K W_K(k,i)\cdot \mu^i_{\infty}$ denote the aggregated impact for population $k$ under stable mean-field ensemble. Then, for any $t,\tau \geq 0$, it holds that for all agent $(k,l)$,
    \begin{align*}
    \E\left[\norm{\tilde{z}^{k,l}_{t+\tau} - z^k_\infty}_1 \big|\ff_t\right] \leq \frac{p_*(3+2p_\mu)+2}{1- L_{pop}}\sqrt{\frac{2|\ss|}{\min_i N_i}} + 2L_{pop}^\tau.
    \end{align*}
\end{corollary}
\begin{proof}
    By triangle inequality, Lemma \ref{lemma: Empirical Population Bound in the Locally Centralized Case}, and Corollary \ref{corollary: empirical population impact bound},
    \begin{align*}
        \E\left[\norm{\tilde{z}^{k,l}_{t+\tau} - z^k_\infty}_1 \big|\ff_t\right] &\leq \E\left[\norm{\tilde{z}^{k,l}_{t+\tau} - \hz^k_{t+\tau}}_1 \big|\ff_t\right] + \E\left[\norm{\hz^{k}_{t+\tau} - z^k_\infty}_1 \big|\ff_t\right]\\
        &\leq \frac{p_*(2+2p_\mu)(1-\tilde L_{pop}^\tau)}{1-\tilde L_{pop}}\sqrt{\frac{2|\ss|}{\min_i N_i}} + 2\sqrt{\frac{2|\ss|}{K \cdot \min_i{N_i}}} + \frac{p^*(1-L_{pop}^\tau)}{1-L_{pop}} \sqrt{\frac{2|\ss|}{\min_i N_i}} + 2L_{pop}^\tau,
    \end{align*}
    and the results follows from simple algebraic manipulations. 
\end{proof}

\subsection{Proof for Proposition \ref{prop: Reachability Under Non-Degenerate Policies in the Locally Centralized Case} and Its Corollary}
\label{appendix: proof for Reachability Under Non-Degenerate Policies}
\begin{proof}
 Denote by $\P_W$ the probability measure under the randomness of re-sampling, that is, the conditional probability given everything except the re-sampling values. First, we bound the distance between the state visitation probabilities of the MP-MFG case and the GGR-S case: assume that the two systems start at the same initial states $\{s_0^{k,l}\}_{k \in [K],l\in[N_k]}$:
    \begin{align*}
        &\E\left[\norm{\P(\hat{s}^{k,l}_{\tmix} = \cdot) - \P_W(\tilde{s}^{k,l}_{\tmix}
        = \cdot)}_1\big|\{s_0^{k,l}\}_{k,l}\right]\\ 
        &= \E\left[\norm{\sum_{s \in \ss}P\left(\cdot|s, \pi^k(s), \hat{z}^{k}_{\tmix-1} \right)\hat{\mu}^k_{\tmix -1}(s) -  \sum_{s\in \ss, \tz \in \zz_N}P\left(\cdot|s, \pi^k(s), \tz \right)\nu_{T_{\mix} - 1}(\tz|s)\tilde{\mu}^{k}_{\tmix -1}(s)}_1\bigg|\{s_0^{k,l}\}_{k,l}\right]\\
        &\leq \E\left[\norm{\sum_{s\in \ss}P\left(\cdot|s, \pi^k(s), \hat{z}^{k}_{\tmix-1} \right)\left(\hat{\mu}^k_{\tmix -1}(s) - \tmu^k_{\tmix -1}(s)\right)}_1\bigg|\{s_0^{k,l}\}_{k,l}\right]\\
        &+\E\left[\sum_{s \in \ss}\tmu^k(s)  \norm{P\left(\cdot|s, \pi^k(s), \hz^{k}_{\tmix-1}\right) - \sum_{\tz \in \zz_N} P\left(\cdot|s, \pi^k(s), \tz \right)\nu(\tz|s)}_1\bigg|\{s_0^{k,l}\}_{k,l}\right]\\
        &\leq \E\left[\norm{\hat{\mu}^k_{\tmix -1} - \tmu^{k}_{\tmix -1}}_1\big|\{s_0^{k,l}\}_{k,l}\right] + p_\mu\E\left[\max_{l \in [N_k]}\norm{\hz^{k}_{\tmix-1} - \tz^{k,l}_{\tmix-1}}_1\big|\{s_0^{k,l}\}_{k,l}\right]\\
        &\leq (1+ p^*_k p_\mu)\E\left[\norm{\hat{\mu}^k_{\tmix -1} - \tmu^{k}_{\tmix -1}}_1\big|\{s_0^{k,l}\}_{k,l}\right] + 2p_\mu\sqrt{\frac{2|\ss|}{K \cdot \min_i{N_i}}}, \text{ by Lemma \ref{lemma: One-Step Error Propagation Through Aggregated Impact}}.
    \end{align*}
    Then, by Proposition \ref{prop: Multi-Step Error Propagation Bound}, since the starting states are the same, i.e. $\norm{\tilde{\bmu}_0 - \hat{\bmu}_0}_1 = 0,$ which means that
    \begin{align*}
        \E&\left[\norm{\P\left(\hat{s}^{k,l}_{\tmix} = \cdot|\{s_0^{k,l}\}_{k,l}\right) - \P_W(\tilde{s}^{k,l}_{\tmix}= \cdot|\{s_0^{k,l}\}_{k,l})}_1\right]\\
        &\leq \frac{(2+2p_\mu)(1- \Tilde{L}_{pop}^{T_\mix - 1})(1+p^*_k p_\mu)}{1-\tilde L_{pop}}\sqrt{\frac{2|\ss|}{\min_i N_i}}+2p_\mu\sqrt{\frac{2|\ss|}{K \cdot \min_i{N_i}}}\\
        &\leq \frac{(2+2p_\mu)(1+p^*_k p_\mu)}{1-\tilde L_{pop}}\sqrt{\frac{2|\ss|}{\min_i N_i}}+2p_\mu\sqrt{\frac{2|\ss|}{K \cdot \min_i{N_i}}}
    \end{align*}
    Then, by Markov Inequality, for any $\epsilon > 0$,
    \begin{align*}
        \P\left(\norm{\P\left(\hat{s}^{k,l}_{\tmix} = \cdot|\{s_0^{k,l}\}_{k,l}\right) - \P_W\left(\tilde{s}^{k,l}_{\tmix}= \cdot|\{s_0^{k,l}\}_{k,l}\right)}_1 \geq \epsilon\right) \leq \frac{\E\left[\norm{\P\left(\hat{s}^{k,l}_{\tmix} = \cdot|\{s_0^{k,l}\}_{k,l}\right) - \P_W(\tilde{s}^{k,l}_{\tmix}= \cdot|\{s_0^{k,l}\}_{k,l})}_1\right]}{\epsilon},
    \end{align*}
    which implies that with probability 
    \begin{align*}
       1 -  \frac{\E\left[\norm{\P\left(\hat{s}^{k,l}_{\tmix} = \cdot|\{s_0^{k,l}\}_{k,l}\right) - \P(\tilde{s}^{k,l}_{\tmix}= \cdot|\{s_0^{k,l}\}_{k,l})}_1\right]}{\epsilon} \geq 1- \frac{\frac{(2+2p_\mu)(1+p^*_k p_\mu)}{1-\tilde L_{pop}}\sqrt{\frac{2|\ss|}{\min_i N_i}}+2p_\mu\sqrt{\frac{2|\ss|}{K \cdot \min_i{N_i}}}}{\epsilon},
    \end{align*}
    we can guarantee that
    \begin{align*}
        \norm{\P\left(\hat{s}^{k,l}_{\tmix} = \cdot|\{s_0^{k,l}\}_{k,l}\right) - \P(\tilde{s}^{k,l}_{\tmix}= \cdot|\{s_0^{k,l}\}_{k,l})}_1 <  \epsilon,
    \end{align*}
    and this implies that for any $s' \in \ss$,
    \begin{align*}
       \P\left(\hat{s}^{k,l}_{\tmix} = s'|\{s_0^{k,l}\}_{k,l}\right) - \P_W\left(\tilde{s}^{k,l}_{\tmix}= s'|\{s_0^{k,l}\}_{k,l}\right) \leq  \norm{\P\left(\hat{s}^{k,l}_{\tmix} = \cdot|\{s_0^{k,l}\}_{k,l}\right) - \P_W\left(\tilde{s}^{k,l}_{\tmix}= \cdot|\{s_0^{k,l}\}_{k,l}\right)}_1 <  \epsilon,
    \end{align*}
    and together with Assumption \ref{assumption: reachability}, we have for any $s' \in \ss$,
    \begin{align*}
        \P_W\left(\tilde{s}^{k,l}_{\tmix}= s'|\{s_0^{k,l}\}_{k,l}\right) > \P\left(\hat{s}^{k,l}_{\tmix} = s'|\{s_0^{k,l}\}_{k,l}\right) - \epsilon > \delta_{\mix} - \epsilon.
    \end{align*}
    Without loss of generality, we may select $\epsilon = \frac{1}{2}\delta_\mix$. 

    Now, when we average over all possible re-sampling of $W$, we obtain: 
    \begin{align*}
        \P\left(\tilde{s}^{k,l}_{\tmix}= s'|\{s_0^{k,l}\}_{k,l}\right)
        &\geq \frac{1}{2}\delta_\mix \P\left(\P_W\left(\tilde{s}^{k,l}_{\tmix}= s'|\{s_0^{k,l}\}_{k,l}\right) > \frac{1}{2} \delta_\mix \right)\\
        &\geq \frac{1}{2}\delta_\mix\left(1- \frac{2}{\delta_\mix} \frac{(2+2p_\mu)(1+p^*_k p_\mu)}{1-\tilde L_{pop}}\sqrt{\frac{2|\ss|}{\min_i N_i}}+2p_\mu\sqrt{\frac{2|\ss|}{K \cdot \min_i{N_i}}}\right)\\
        &= \frac{1}{2}\delta_\mix - \left(\frac{(2+2p_\mu)(1+p^*_k p_\mu)}{1-\tilde L_{pop}}\sqrt{\frac{2|\ss|}{\min_i N_i}}+2p_\mu\sqrt{\frac{2|\ss|}{K \cdot \min_i{N_i}}}\right)\\
        &\geq  \frac{1}{2}\delta_\mix - \frac{(1+p^*_k p_\mu)(2+2p_\mu) +2p_\mu}{1-\tilde L_{pop}}\sqrt{\frac{2|\ss|}{\min_i N_i}}\\
         &:= \delta'_\mix,
         \end{align*} 
         as desired. 
\end{proof}

Using Proposition \ref{prop: Reachability Under Non-Degenerate Policies in the Locally Centralized Case}, we can bound the empirical neighbor impact and the limiting aggregated impact conditional on the state of a single agent $(k,l)$. 

\begin{corollary}
\label{appendix corollary: conditional bound on neighbor impact}
    Under the settings of Lemma \ref{lemma: Empirical Population Bound in the Locally Centralized Case} and Proposition \ref{prop: Reachability Under Non-Degenerate Policies in the Locally Centralized Case}, for any state $s \in \ss$, agent $(k,l)$, and any $\tau > T_\mix$, it follows that for $\tz^{k,l}_{t+\tau}$ as defined in Eq.~\eqref{eq: tilde_z-def} with respect to $\tilde \bmu_{t+\tau}$, and $z^k_{\infty}$ as defined in Eq. \eqref{eq:z_def} with respect to $\bmu_{\infty} := \Gamma_{pop}^\infty(\bpi)$,
    \begin{align*}
    \E\left[\norm{\tilde{z}^{k,l}_{t+\tau} - z^k_{\infty}}_1 \big|\ts^{k,l}_{t+\tau} = s, \ff_t\right] \leq \frac{p_*(3+2p_\mu)+2}{\delta'_\mix(1- L_{pop})}\sqrt{\frac{2|\ss|}{\min_i N_i}} + \frac{2L_{pop}^\tau}{\delta_\mix'}.
    \end{align*}
\end{corollary}
\begin{proof}
    The proof follows from law of total expectation. 
    \begin{align*}
         \E\left[\norm{\tilde{z}^{k,l}_{t+\tau} - z^k_{\infty}}_1 \big|\ff_t\right] = \sum_{s \in \ss} \E\left[\norm{\tilde{z}^{k,l}_{t+\tau} - z^k_{\infty}}_1 \big|\ts^{k,l}_{t+\tau} = s, \ff_t\right]\P(\ts^{k,l}_{t+\tau} = s|\ff_t),
    \end{align*}
    and since $\tau > T_\mix$, Proposition \ref{prop: Reachability Under Non-Degenerate Policies in the Locally Centralized Case} implies that $\P(\ts^{k,l}_{t+\tau} = s|\ff_t) > \delta_\mix'$ which means that for any $s\in \ss$, 
    \begin{align*}
        \E\left[\norm{\tilde{z}^{k,l}_{t+\tau} - z^k_{\infty}}_1 \big|\ts^{k,l}_{t+\tau} = s, \ff_t\right]\leq \frac{ \E\left[\norm{\tilde{z}^{k,l}_{t+\tau} - z^k_{\infty}}_1 \big|\ff_t\right]}{\P(\ts^{k,l}_{t+\tau} = s|\ff_t)} \leq \frac{p_*(3+2p_\mu)+2}{\delta'_\mix(1- L_{pop})}\sqrt{\frac{2|\ss|}{\min_i N_i}} + \frac{2L_{pop}^\tau}{\delta_\mix'},
    \end{align*}
    by Corollary \ref{corollary: Convergence of Empirical Neighbor Impact to Stable Aggregated Impact}.
\end{proof}

\subsection{Proof for Proposition \ref{Prop: Convergence of State Visitation Probabilities to Stable Mean-Field}}
\label{appendix: proof for Prop: Convergence of State Visitation Probabilities}
\begin{proof}
Denote by $z^k_\infty$ the aggregated impact for population $k$ under the stable mean-field ensemble $\bmu_\infty = \Gamma^\infty_{pop}(\bpi)$, i.e. $z^k_\infty = \frac{1}{K}\sum_{i=1}^K W_K(k,i)\cdot \mu^i_\infty$. Let $\bP_{\infty}^k$ denote the limiting probability transition of population $k$, i.e. $[\bP_{\infty}^k]_{s,s'} = \bar P(s'|s,\pi^k(s),z^k_\infty)$. Then, $\mu_\infty^k$ is the limiting distribution induced by $\bP_{\infty}^k$. Also, let $ \bP_{t}^k[s]$ denote the one-step transition for population $k$, i.e. $ [\bP_{t}^k]_{s,s'} = \P(s_t^{k,l} = s'|s^{k,l}_{t-1}=s)$. Lastly, define $M_\infty^k$ as the matrix with all rows equal to $\mu^k_\infty$. We remark that under a particular stable mean-field ensemble induced by a policy profile $\bpi$, the system is a Markov Chain. Now, the goal is to show that the state visitation of any agent of the GGR-S system under the policy profile $\bpi$ is converging to this limiting Markov Chain.  

The first two steps of the proof follow from \cite{Yardim_PMA} with some modifications. We present the modified conclusions. 
\begin{enumerate}
        \item Suppose Assumption \ref{assumption: reachability} holds, then Proposition \ref{prop: Reachability Under Non-Degenerate Policies in the Locally Centralized Case} implies that there exists a $T_\mix > 0$ such that for some $\delta_\mix'$, the matrix defined by $\bP_{(j)}^k = \prod_{t=(j-1)T_\mix + 1}^{jT_\mix} \bP_t^k$ satisfies for all $j$ that $\left[\bP_{(j)}^k\right]_{s,s'} > \delta_\mix'\mu^k_\infty(s')>0$. Then define the matrix $Q^k_{(j)}$ implicitly as $\bP_{(j)}^k = (1-\theta)M_\infty^k + \theta Q_{(j)}^k$ with $\theta:= 1-\delta_\mix'$. Then, it holds that for all $J > 0$,  
        \begin{align*}
            \prod_{j=1}^J \bP_{(j)}^k = (1-\theta^J)M_\infty^k + \theta^J \prod_{j=1}^J Q_{(j)}^k + \sum_{j=2}^J(1-\theta^{l-1})\theta^{J-j}M^k_\infty\left(\bP^k_{(j)} - (\bP_\infty^k)^{T_\mix}\right)\prod_{j'=j+1}^J Q^k_{(j')}.
        \end{align*}

        \item For some arbitrary column vector $c \in \Delta(\ss)$, $r< \tmix$, let $\bP_{t}^k[s]$ and  $\bP_{\infty}^k[s]$ denote the $s$-th row of the matrices $\bP_{t}^k$ and
        \begin{align*}
        \norm{c^\top \prod_{t=1}^{J\tmix + r}\bP_t^k - \mu^k_\infty}_1 &\leq 2\theta^J + \sum_{j=2}^J \theta^{J-j} \sum_{t=(j-1)\tmix +1}^{j\tmix}\sup_s \norm{\bP_{t}^k[s] -\bP_{\infty}^k[s]}_1 + \sum_{t= J\tmix +1}^{J\tmix + r} \sup_s \norm{\bP_{t}^k[s] -\bP_{\infty}^k[s]}_1\\
        \end{align*}
\end{enumerate}
    Lastly, we may use the above two results to prove the result. First, we need to bound the different between the one-step transition from state $s$ in trajectory and under the stable mean-field ensemble, i.e., $\sup_s \norm{\bP_{t}^k[s] -\bP_{\infty}^k[s]}_1$. Note that for any $s \in \ss$, 
    \begin{align*}
        \bP_{t}^k[s] -\bP_{\infty}^k[s] &= \P(\tilde s^{k,l}_{t+1} = \cdot\big|\tilde s^{k,l}_t = s) - P(\cdot|s,\pi^k(s), z^k_\infty)\\
        &= \sum_{\tz \in \zz_N}\left(\P\left(\tilde s^{k,l}_{t+1} = \cdot \big|\tilde{z}^{k,l}_t = \tz, \tilde s^{k,l}_t = s\right) - P(\cdot|s,\pi^k(s), z^k_\infty)\right)\nu(\tz|s)\\
        &= \sum_{\tz \in \zz_N} \left(P\left(\cdot|s, \pi^k(s), \tz\right) - P\left(\cdot|s, \pi^k(s), z^k_\infty\right)\right)\nu(\tz|s),
    \end{align*}
    where the second equality is due to law of total probability and the last equality uses Definition \ref{defn: Locally Centralized Learning Model}.

    Now, assuming that $t \geq \tmix$, it follows that
    \begin{align*}
        \norm{\bP_{t}^k[s] -\bP_{\infty}^k[s]}_1 \leq \sum_{\tz \in \zz_N} \norm{\left(\bar P\left(\cdot|s, \pi^k(s), \tz\right) - \bar P\left(\cdot|s, \pi^k(s), z^k_\infty\right)\right)}_1\nu(\tz|s) \leq \frac{p_\mu}{\delta'_\mix} \E\left[\max_{k,l}\norm{\tilde{z}_t^{k,l} - z_\infty^k}_1\big|\tilde s^{k,l}_t = s\right],
    \end{align*}

    where the last inequality follows from the law of total expectation and Lemma \ref{Lipschitz Continuity of Pbar Rbar}. Then, by Corollary \ref{corollary: Convergence of Empirical Neighbor Impact to Stable Aggregated Impact},
    \begin{align*}
         \sup_{s}\norm{\bP_{t}^k[s] -\bP_{\infty}^k[s]}_1 &\leq A\frac{p_\mu}{\delta'_\mix} \sqrt{\frac{2|\ss|}{\min_i N_i}} + \frac{2p_\mu}{\delta'_\mix}L_{pop}^t,
    \end{align*}
    where $A = \frac{p_*(3+2p_\mu)+2}{1- L_{pop}}.$
    
    From \cite{Yardim_PMA}, we know that for some arbitrary column vector $c \in \Delta(\ss)$, $r< \tmix$,
    \begin{align*}
        \norm{c^\top \prod_{t=1}^{J\tmix + r}\bP_t^k - \mu^k_\infty}_1 &\leq 2\theta^J + \sum_{j=2}^J \theta^{J-j} \sum_{t=(j-1)\tmix +1}^{j\tmix}\sup_s \norm{\bP_{t}^k[s] -\bP_{\infty}^k[s]}_1 + \sum_{t= J\tmix +1}^{J\tmix + r} \sup_s \norm{\bP_{t}^k[s] -\bP_{\infty}^k[s]}_1\\
        &\leq 2\theta^J + \left(\tmix  A\frac{p_\mu}{\delta'_\mix}\sqrt{\frac{2|\ss|}{\min_i N_i}} \right)\sum_{j=2}^J \theta^{J-j} +  \frac{2p_\mu}{\delta'_\mix}\sum_{j=2}^J \theta^{J-j}\sum_{t=(j-1)\tmix +1}^{j\tmix}L_{pop}^t\\
        &+ rA\frac{p_\mu}{\delta'_\mix}\sqrt{\frac{2|\ss|}{\min_i N_i}} + \frac{2p_\mu}{\delta'_\mix}\sum_{t= J\tmix+1}^{J\tmix + r} L_{pop}^t\\
        &\leq 2\theta^J + \left( A\frac{p_\mu}{\delta'_\mix}\sqrt{\frac{2|\ss|}{\min_i N_i}} \right)\left(\tmix \frac{1-\theta^{J-1}}{1-\theta}+ r\right)+\frac{2p_{\mu}}{\delta'_{\mix}}\sum_{j=2}^{J}\theta^{J-j}\frac{L_{pop}^{(j-1)T_{\mix}+1}\left(1-L_{pop}^{T_{\mix}}\right)}{1-L_{pop}}\\
        &+\frac{2p_{\mu}}{\delta'_{\mix}}\frac{L_{pop}^{JT_{\mix}+1}\left(1-L_{pop}^{r}\right)}{1-L_{pop}}\\
        &\leq 2\theta^{J}+2\left( A\frac{p_\mu}{(\delta'_\mix)^2}\sqrt{\frac{2|\ss|}{\min_i N_i}} \right)T_{\mix}+\frac{2p_{\mu}L_{pop}\theta^{J-1}}{\delta'_{\mix}}\frac{1-L_{pop}^{T_{\mix}}}{1-L_{pop}}\sum_{j=2}^{J}\left(\frac{L_{pop}^{T_{\mix}}}{\theta}\right)^{j-1}\\
        &+\frac{2p_{\mu}}{\delta'_{\mix}}\frac{L_{pop}^{JT_{\mix}+1}\left(1-L_{pop}^{r}\right)}{1-L_{pop}}.
    \end{align*}
    Now, when $L_{pop}^{\tmix} > \theta$, i.e. $L_{pop}<(1-{\delta^\prime}_{\mix})^{1/{T_{\mix}}}$,
        \begin{align*}
            &2\theta^{J}+\frac{2p_{\mu}L_{pop}\theta^{J-1}}{\delta'_{\mix}}\frac{1-L_{pop}^{T_{\mix}}}{1-L_{pop}}\sum_{j=2}^{J}\left(\frac{L_{pop}^{T_{\mix}}}{\theta}\right)^{j-1}+\frac{2p_{\mu}}{\delta'_{\mix}}\frac{L_{pop}^{JT_{\mix}+1}\left(1-L_{pop}^{r}\right)}{1-L_{pop}}\\
        &\leq 2\theta^{J}+\frac{2p_{\mu}L_{pop}\theta^{J-1}}{\delta'_{\mix}}\frac{1}{1-L_{pop}}\frac{\frac{L_{pop}^{T_{\mix}}}{\theta}}{1-\frac{L_{pop}^{T_{\mix}}}{\theta}}+\frac{2p_{\mu}}{\delta'_{\mix}}\frac{L_{pop}^{JT_{\mix}+1}}{1-L_{pop}}\\
        &= 2\theta^{J}+\frac{2p_{\mu}L_{pop}\theta^{J-1}}{\delta'_{\mix}}\frac{1}{1-L_{pop}}\frac{{L_{pop}^{T_\mix}}}{\theta-L_{pop}^{T_{\mix}}}+\frac{2p_{\mu}}{\delta'_{\mix}}\frac{L_{pop}^{JT_{\mix}+1}}{1-L_{pop}}\\
        &\leq \theta^{J}\left(2+\frac{2p_{\mu}L_{pop}}{\delta'_{\mix}(1-L_{pop})(\theta-L_{pop}^{T_{\mix}})}+\frac{2p_{\mu}L_{pop}}{\delta'_{\mix}(1-L_{pop})}\right)
        \end{align*}
        and similarly, when \(L_{pop}>(1-\delta'_{\mix})^{1/T_{\mix}}\),
        \begin{align*}
              &2\theta^{J}+\frac{2p_{\mu}L_{pop}\theta^{J-1}}{\delta'_{\mix}}\frac{1-L_{pop}^{T_{\mix}}}{1-L_{pop}}\sum_{j=2}^{J}\left(\frac{L_{pop}^{T_{\mix}}}{\theta}\right)^{j-1}+\frac{2p_{\mu}}{\delta'_{\mix}}\frac{L_{pop}^{JT_{\mix}+1}\left(1-L_{pop}^{r}\right)}{1-L_{pop}}\\
              & \leq 2\theta^{J}+\frac{2p_{\mu}L_{pop}\theta^{J-1}}{\delta'_{\mix}}\frac{1}{1-L_{pop}}\frac{(\frac{L_{pop}^{T_{\mix}}}{\theta})^{J}}{\frac{L_{pop}^{T_{\mix}}}{\theta}-1}+\frac{2p_{\mu}}{\delta'_{\mix}}\frac{L_{pop}^{JT_{\mix}+1}}{1-L_{pop}}\\
              &= L_{pop}^{T_{\mix}J}\left(2+\frac{2p_{\mu}L_{pop}}{\delta'_{\mix}(1-L_{pop})(L_{pop}^{T_{\mix}}-\theta)}+\frac{2p_{\mu}L_{pop}}{\delta'_{\mix}(1-L_{pop})}\right)
        \end{align*}
    Combining the two cases, we obtain 
    \begin{align*}
     \norm{c^\top \prod_{t=1}^{J\tmix + r}\bP_t^k - \mu^k_\infty}_1
       \leq \frac{2p_{\mu}T_{\mix}(p_*(3+2p_\mu)+2)}{{\delta'_{\mix}}^{2}(1-L_{pop})}  \sqrt{\frac{2|\mathcal{S}|}{\min_i N_i}}+C_{\mix}\rho_{\mix}^{T} 
    \end{align*}
    where \(\rho_{\mix}:=\max\{L_{pop},(1-\delta'_{\mix})^{1/T_{\mix}}\}\) and \(C_{\mix}:=\left(2+\frac{2p_{\mu}L_{pop}}{\delta'_{\mix}(1-L_{pop})|\theta-L_{pop}^{T_{\mix}}|}+\frac{2p_{\mu}L_{pop}}{\delta'_{\mix}(1-L_{pop})}\right)/(\rho_{\mix}^{T_{\mix}})\).
\end{proof}

\section{Learning GGR-S}
\label{appendix: learning GGR-S}
\begin{defn}[TD operators]
\label{appendix defn: TD operators}
     Let $\bpi \in \Pi^K$ and $\bmu_\infty := \Gamma_{pop}^\infty(\bpi)$ be the stable mean-field ensemble induced by the policy profile $\bpi$, and $z_\infty^k$ defined as in \eqref{eq:z_def} with respect to $\bpi_\infty$. We define the Bellman operator for population $k$, $T^{\pi^k}: \qq \to \qq$ as: $$(T^{\pi^k} Q^k)(s,a) := R(s,a,z_\infty^k) + h(\pi^k(s)) + \gamma\sum_{s',a'} P(s'|s,a,z_\infty^k)\pi(a'|s')Q^k(s',a'),$$for each $Q \in \qq$. Then the corresponding TD-learning operator under the stable mean-field ensemble is defined as: for population $k$ with policy $\pi^k \in \Pi$, $F^{\pi^k}(Q^k) := M^{\pi^k}(Q^k-T^{\pi^k} Q^k)$,
    where $M^{\pi^k} := \text{diag}(\{{\mu^k_\infty}(s) \pi^k(a|s)\}_{s,a}) \in \R^{|\ss||\aa|\times |\ss||\aa|}$ is a matrix denoting the state-action distribution induced by the policy $\pi$ at the limiting mean field distribution $\mu^k$.
\end{defn}

Note that $F$ is the temporal difference operator under the limiting stable mean-field ensemble to which we do not have access in sample-based learning. We aim to approximate $F$ with $\tilde F$ as defined in Definition \ref{defn: stochastic TD operator}. This approximation was first proposed in \cite{kotsalis2021simple} and then extended to the standard MFG case by \cite{Yardim_PMA}.

\subsection{CTD Learning in GGR-S}
\label{appendix: CTD learning in GGR-S}

First, we present the CTD learning algorithm which collects samples to perform a single update of the $Q^k$ function. 

\begin{algorithm}
\caption{GGR-S CTD Learning}
\label{algo:GGRS CTD learning}
    \begin{algorithmic}[1]
    \Require CTD learning rate $\{\beta_n\}_n$, CTD iteration $I_{ctd}$, mixing time $I_\mix$, and policy profile $\bpi$. 
    \State Set $\tilde Q^k_0(\cdot,\cdot) \gets Q_{\max}$, for all $k$.
        \For{$n \in 0,1,\dots I_{ctd}-1$}
            \For{$t \in 1, \dots, I_\mix$}
                \State Re-sample $W^{[N]}_t$ from $W_K$.
                \State \multiline{Compute for all $(k,l)$: $ \tilde{z}_t^{k,l} = \frac{1}{K}\sum_{i=1}^K\frac{1}{N_i}\left[\sum_{j=1}^{N_i} W_t^{[N]}(t^{k,l},t^{i,j}) \delta_{\ts_t^{i,j}}\right]$.}
                \State  \multiline{Simulate for all $(k,l)$, $a_t^{k,l} \sim \pi^k(\ts_t^{k,l})$, $\ts_{t+1}^{k,l} \sim P(\cdot|\ts_t^{k,l}, a_t^{k,l}, \Tilde{z}^{k,l}_t)$, $r_{t+1}^{k,l} = R(\cdot|\ts_t^{k,l}, a_t^{k,l}, \tz^{k,l}_t)$.} 
            \EndFor
            \State \multiline{Observe $\omega_{n}^{k} = (s_{t-2}^{k,1}, a_{t-2}^{k,1}, r_{t-2}^{k,1}, s_{t-1}^{k,1}, a_{t-1}^{k,1})$.}
            \State \multiline{CTD update: $\tilde Q^k_{n+1} = \tilde Q^k_n - \beta_n \tilde F^{\pi^{k}}(\tilde Q^k_n, \omega^k_{n})$, for all $k$.}
        \EndFor
    \State Return updated state-action value function $\Tilde{Q}_{I_{ctd}}$.
    \end{algorithmic}
\end{algorithm}

\begin{thm}[GGR-S CTD Learning]
\label{Theorem: GGR-S CTD Learning}
    In the GGR-S framework, suppose that Assumption \ref{assumption: reachability} holds and suppose the policy profile $\bpi$ satisfy Assumption \ref{assumption: non-degenerate policies}. Set learning rates $\beta_n = \frac{2}{\rho_F(t_0 + n - 1)}, \forall n \geq 0,$ and $I_{ctd} > \mathcal{O}(\epsilon^{-2})$, $I_\mix > \mathcal{O}(\log \epsilon^{-1})$. Let $Q_\infty^k := (\cdot, \cdot|\pi^k, \bmu_\infty)$. Then the output $\tilde Q^k_{I_{ctd}}$ of Algorithm \ref{algo:GGRS CTD learning} satisfies for any $k\in [K]$. 
    \begin{align*}
        \E\left[\norm{\tilde Q^k_{I_{ctd}} -Q_\infty^k}_\infty \right] \leq \epsilon + \mathcal{O}\left(\frac{1}{\sqrt{\min_{i}{N_i}}}\right). 
    \end{align*}
    where $t_0 = \frac{64(1+\gamma)^2}{\rho_F^2}$, $ I_{mix} > \frac{\log 1/(\rho_F)+ \log 20(1+\gamma)C_{\mix}}{\log 1/\rho_\mix}$ and $\rho_F = (1-\gamma)\delta_\mix' \zeta$.
\end{thm}

\begin{proof}
    The main idea of this proof is based on Theorem D.2 and D.7 in \cite{Yardim_PMA} which generalizes the CTD result from \cite{kotsalis2021simple} to standard MFG settings. In order to use this theorem, we need to check all the five assumptions hold in GGR-S setting.
    
    First, we summarize several key results from \cite{kotsalis2021simple} and \cite{Yardim_PMA} which facilitate this proof. 
    \begin{enumerate}
        \item Under Assumption \ref{assumption: non-degenerate policies}, the set $\{u \in \Delta(\ss): u(\cdot) \geq \zeta \}$ is compact which implies that $\nabla h$ is continuous on this set. Thus, we denote the Lipschitz constant of the regularizer $h$ as $C_h$.
        \item The TD operator $T^\pi$ is contractive and Lipschitz continuous with Lipschitz constant $L_F = 1+\gamma$ with respect to the $\ell_2$ norm on $\qq$ for all policy $\pi \in \Pi$. 
        \item \cite{kotsalis2021simple} has proven that $F^\pi$ is generalized strongly monotone with modulus $\rho_F := (1-\gamma)\delta'_\mix \zeta$. In particular, for any $Q^k \in \qq$, 
        \begin{align*}
            \langle F^\pi(Q^k), Q^k-Q^*\rangle \geq \rho_F\norm{Q^k-Q^k_\infty}_2^2,
        \end{align*}
        where $Q^k_\infty$ is the true state-action value function of the policy profile $\bpi$ under the stable mean-field ensemble $\bmu_\infty = \Gamma_{pop}^\infty(\bpi)$, i.e. $Q^k_\infty = Q^k(\cdot, \cdot|\bpi, \bmu_\infty)$.  
        \item Lemma 16 in \cite{LanPMD} proves that the distance between $\tilde{F}^{\pi^k}$-updated $Q$-function and its expected value is bounded: for the $n$-th observation $\omega_n \in \Omega$ (after $n$ rounds of waiting $T_\mix$ steps), 
        \begin{align*}
            \E\left[\tilde F^{\pi^k}\left(\tilde{Q}^k, \omega_n\right) - \E\left[\tilde F^{\pi^k}\left(\tilde{Q}^k, \omega_n\right)|\ff_{(n-1)T_\mix}\right]\bigg|\ff_{(n-1)T_\mix}\right] \leq 4(1+\gamma)^2 \E\left[\norm{\tilde{Q}^k - Q^k_\infty}_2^2\right] + \frac{4(1+L_h)^2}{(1-\gamma)^2},
        \end{align*}
        where $Q^k_\infty = Q^k(\cdot, \cdot|\bpi, \bmu_\infty)$.
    \end{enumerate}

    The above results ensure that Assumption A1 - A4 in Theorem D.2 from \cite{Yardim_PMA} hold. 
    
    Then, we need to check the last assumption for our model - convergence of the stochastic TD update $\tilde F$ to $F$ (See also Lemma 4.1 in \cite{kotsalis2023simple}): for any policy profile $\bpi$, any $\tilde{Q}^k \in \qq$, and any population $k$, $t\geq 0$, 
\begin{align*}
    \norm{\ee\left[\tilde F^{\pi^k}\left(\tilde{Q}^k, \omega_{t+\tau}\right)\big| \ff_t\right] - F^{\pi^k}(\tilde{Q}^k)}_2 \leq \norm{\ee\left[\left(R(\tilde s_{t+\tau},a_{t+\tau},z^k_\infty) - R(\ts_{t+\tau}, a_{t+\tau}, \tz^k_{t+\tau})\right)\mathbf{e}_{\ts_{t+\tau}, a_{t+\tau}}\big|\ff_t\right]}_2\\
    +\norm{\ee\left[\tilde{Q}^k(\tilde s_{t+\tau},a_{t+\tau}) - R(\tilde s_{t+\tau},a_{t},z^k_\infty) - h(\pi^k(\tilde s_{t+\tau})) - \gamma \tilde{Q}^k(\tilde s_{t+\tau +1 },a_{t+\tau+1})\mathbf{e}_{\tilde s_{t+\tau},a_{t+\tau}}\big|\ff_t\right]- F^{\pi^k}(\tilde{Q}^k)}_2, 
\end{align*}
by triangle inequality. 

The first term can be bounded as
\begin{align*}
    &\norm{\ee\left[\left(R(\tilde s_{t+\tau},a_{t+\tau},z^k_\infty) - R(s_{t+\tau}, a_{t+\tau}, \tz^k_{t+\tau})\right)\mathbf{e}_{\ts_{t+\tau}, a_{t+\tau}}\big|\ff_t\right]}_2\\
    &\quad \leq \ee\left[\norm{\left(R(\tilde s_{t+\tau},a_{t+\tau},z^k_\infty) - R(\tilde s_{t+\tau}, a_{t+\tau}, \tz^k_{t+\tau})\right)\mathbf{e}_{\ts_{t+\tau}, a_{t+\tau}}}_1\big|\ff_t\right], \\
    &\quad \leq r_{\mu} \ee\left[\norm{z^k_\infty - \tz^k_{t+\tau}}_1\big|\ff_t\right]\\
    &\quad \leq  \frac{r_{\mu}p_*(3+2p_\mu)+2}{1- L_{pop}}\sqrt{\frac{2|\ss|}{\min_i N_i}} + 2r_\mu L_{pop}^\tau
\end{align*}
where the first inequality is due to Jensen's inequality and $\norm{\cdot}_2 \leq \norm{\cdot}_1$; the second inequality is due to Assumption \ref{Lipschitz Continuity of $P,R$}; and the last inequality is due to Corollary \ref{corollary: Convergence of Empirical Neighbor Impact to Stable Aggregated Impact}.

Next, note that in the second term: let $$g^k(s,a,s',a') =\tilde{Q}^k(s,a) - R(s,a,z^k_\infty) - h(\pi^k(s)) - \gamma \tilde{Q}^k(s',a'),$$ then we have:
\begin{align*}
    &\ee\left[\tilde{Q}^k(\tilde s_{t+\tau},a_{t+\tau}) - R(\tilde s_{t+\tau},a_{t+\tau},z^k_\infty) - h(\pi^k(\ts_{t+\tau})) - \gamma \tilde{Q}^k(\ts_{t+\tau +1 },a_{t+\tau +1 }) e_{\tilde s_{t+\tau},a_{t+\tau}}\big|\ff_t\right]- F^{\pi^k}(\tilde{Q}^k)\\
    &= \sum_{\substack{s,s'\in \ss\\a,a'\in \aa}}\sum_{\tz \in \zz_N} \mathbf{e}_{s,a} g^k(s,a,s',a') P(s'|s,a,\tz)\P(\ts_{t+\tau} = s|\ff_t)\P(\tz_{t+\tau}^k = \tz|\ts_{t+\tau} = s, \ff_t)\pi^k(a|s)\pi^k(a'|s')- F^{\pi^k}(\tilde{Q}^k)\\
    &=\underbrace{\sum_{\substack{s,s'\in \ss\\a,a'\in \aa}} \mathbf{e}_{s,a} g^k(s,a,s',a')\pi^k(a|s)\pi^k(a'|s')\P(\ts_{t+\tau} = s|\ff_t)\left(\sum_{\tz \in \zz_N}  P(s'|s,a,\tz)\P\left(\tz^{k,l}_{t+\tau} = \tz|\ts_{t+\tau}^{k,l}=s,\ff_t\right)- P(s'|s,a,z_\infty^k)\right)}_\text{$A$}\\
    &\quad + \underbrace{\sum_{\substack{s,s'\in \ss\\a,a'\in \aa}} \mathbf{e}_{s,a} g^k(s,a,s',a')P(s'|s,a,z_\infty^k)\P(\ts_{t+\tau} = s|\ff_t)\pi^k(a|s)\pi^k(a'|s')- F^{\pi^k}(\tilde{Q}^k)}_\text{$B$},
\end{align*}

In particular, we start with the first term: 
\begin{align*}
    \norm{\sum_{\tz \in \zz_N}  P(\cdot|s,a,\tz)\P\left(\tz^{k,l}_{t+\tau} = \tz|\ts_{t+\tau}^{k,l}=s,\ff_t\right)- P(\cdot|s,a,z_\infty^k)}_1 &\leq p_\mu \E\left[\max_{k,l} \norm{\tz^{k,l}_{t+\tau} - z_\infty^k}_1 \bigg|\ts_{t+\tau} = s, \ff_t \right]\\
    &\leq \frac{p_\mu(p_*(3+2p_\mu)+2)}{\delta'_{\mix}(1- L_{pop})}\sqrt{\frac{2|\ss|}{\min_i N_i}} + \frac{2p_\mu L_{pop}^\tau}{\delta_\mix'},
\end{align*}
where the last inequality follows from Corollary \ref{appendix corollary: conditional bound on neighbor impact}. Furthermore, since the algorithm output $\tilde Q^k \in \left[\frac{h_{\max}-L_h}{1-\gamma}, \frac{1+h_{\max}}{1-\gamma}\right]$ with $\tilde Q^k$ initialized to $\frac{1+h_{\max}}{1-\gamma}$, it means that for vectors $\mathbf{v}, \mathbf{v}' \in \ss\times \aa \times \ss \times \aa$, 
$$\norm{g^k(\mathbf{v}) - g^k(\mathbf{v}')}_1 \leq \frac{2(1+L_h)}{1-\gamma}.$$

Therefore, applying Lemma \ref{Yardim Lemma 9} to term $A$ gives
\begin{align*}
    \norm{A}_1 \leq \frac{1+ L_h}{\delta_\mix'(1-\gamma)}\left(\frac{p_\mu p_*(3+2p_\mu)+2p_\mu}{1- L_{pop}}\sqrt{\frac{2|\ss|}{\min_i N_i}} + 2p_\mu L_{pop}^\tau\right).
\end{align*}

Next, we consider the second term: by Definition \ref{appendix defn: TD operators}, $F^{\pi^k}(Q^k) := M^{\pi^k}(Q^k-T^{\pi^k} Q^k),$ and $$\sum_{\substack{s,s'\in \ss\\a,a'\in \aa}} \mathbf{e}_{s,a}g^k(s,a,s',a')P(s'|s,a,z_\infty^k)\P(\ts_{t+\tau} = s|\ff_t)\pi^k(a|s)\pi^k(a'|s') = M_{t+\tau}(I - \gamma P_\infty)(Q^k-Q^k_\infty),$$ with $M_{t+\tau} = \text{diag}(\P(\ts_{t+\tau} = s|\ff_t)\pi(a|s))$ and $[P_{\infty}^k]_{s,s'} = \bar P(s'|s,\pi^k(s),z^k_\infty)$. Thus, $
    B = (M_{t+\tau} - M^{\pi^k})(I - \gamma P_\infty)(Q-Q^k_\infty)$. Since $M_{t+\tau}$ and $M^{\pi^k}$ are diagonal matrices the max difference of the diagonal entries is the max eigenvalue, so $\lambda_{max}(M_{t+\tau} - M^{\pi^k}) \leq \norm{\P(s^{k,l}_{t}=\cdot)-\Gamma^{\infty}_{pop}[k](\pi^{k})}_1$.
    
    Putting everything together, we obtain: 
\begin{align*}
    \norm{B}_2 &\leq \norm{\P(s^{k,l}_{t}=\cdot)-\Gamma^{\infty}_{pop}[k](\pi^{k})}_1\norm{I -\gamma P_\infty}_2\norm{Q^k-Q_\infty^k}_2\\
    &\leq (1+\gamma)\left(C_\mix \rho_\mix^{T} + \frac{2p_{\mu}T_{\mix}(3p_*+2p_*p_\mu+2)}{{\delta'_{\mix}}^{2}(1-L_{pop})}  \sqrt{\frac{2|\mathcal{S}|}{\min_i N_i}}\right)\norm{Q^k-Q_\infty^k}_2,
\end{align*}
since $\norm{\cdot}_1 \leq \norm{\cdot}_1$ and $\norm{I-\gamma P_\infty }_1 \leq 1+\gamma$. 

Then, we can conclude that
\begin{align*}
    \norm{\ee\left[\tilde F^{\pi^k}\left(\tilde{Q}^k, \omega_{t+\tau}\right)\big| \ff_t\right] - F^{\pi^k}(\tilde{Q}^k)}_2 &\leq \frac{r_{\mu}(3p_*+2p_*p_\mu+2)}{1- L_{pop}}\sqrt{\frac{2|\ss|}{\min_i N_i}} + 2r_\mu L_{pop}^\tau\\
    &\quad + \frac{1+ L_h}{\delta_\mix'(1-\gamma)}\left(\frac{p_\mu (3p_*+2p_*p_\mu+2))}{1- L_{pop}}\sqrt{\frac{2|\ss|}{\min_i N_i}} + 2p_\mu L_{pop}^\tau\right)\\
    &\quad + (1+\gamma)\left(C_\mix \rho_\mix^{T} + \frac{2p_{\mu}T_{\mix}(3p_*+2p_*p_\mu+2)}{{\delta'_{\mix}}^{2}(1-
    L_{pop})}  \sqrt{\frac{2|\mathcal{S}|}{\min_i N_i}}\right)\norm{Q^k-Q_\infty^k}_2\\
    &\leq \frac{(1+ L_h)(r_\mu+p_\mu) (3p_*+2p_*p_\mu+2))}{(1- L_{pop})\delta_\mix'(1-\gamma)}\sqrt{\frac{2|\ss|}{\min_i N_i}} + \frac{2(1+ L_h)(r_\mu+p_\mu)}{\delta_\mix'(1-\gamma)} L_{pop}^\tau\\
    &\quad + (1+\gamma)\left(C_\mix \rho_\mix^{T} + \frac{2p_{\mu}T_{\mix}(3p_*+2p_*p_\mu+2)}{{\delta'_{\mix}}^{2}(1-
    L_{pop})}  \sqrt{\frac{2|\mathcal{S}|}{\min_i N_i}}\right)\norm{Q^k-Q_\infty^k}_2.
\end{align*}

Therefore, Corollary 3.9 in \cite{kotsalis2023simple} and Theorem D.2 in \cite{Yardim_PMA} imply that when $t_0 := \frac{64(1+\gamma)^2}{\rho_F^2}$, $I_{mix} > \frac{\log 1/\rho_F+ \log 20(1+\gamma)C_{\mix}}{\log 1/\rho_\mix}$, the output $\tilde Q_{n}$ of Algorithm \ref{algo:GGRS CTD learning} satisfies:
\begin{align*}
      \E\left[\norm{\tilde Q_{n}^k -Q_\infty^k}^2_2 \right] &\leq \frac{2(t_0 + 1)(t_0 + 2)\norm{ Q_{\max}^k -Q_\infty^k}^{2}_2}{(n + t_0)(n+t_0+1)} + \frac{12n\left(\frac{8(1+L_h)^2}{(1-\gamma)^2}+ 4C_{pop}^2\right)}{\rho_F^2(n+t_0)(n+t_0+1)} + \frac{200C_{pop}^2}{\rho_F^2}
\end{align*}
where \begin{align*}
    C_{pop} &=\frac{(\delta_\mix'+4T_\mix)(1+L_h)(r_\mu + p_{\mu})(3p_*+2p_*p_\mu+2)}{{\delta'_{\mix}}^{2}(1-\gamma)(1-
    L_{pop})}  \sqrt{\frac{2|\mathcal{S}|}{\min_i N_i}} + \frac{(1+ L_h)2(p_\mu+r_{\mu}) }{\delta_\mix'(1-\gamma)}L_{pop}^{I_{\mix}}
\end{align*}
and note that $C_{pop} = \mathcal{O}\left(\frac{1}{\sqrt{\min_{i}{N_i}}}\right)$. 

Finally, using the fact that $\norm{\cdot}_\infty \leq \norm{\cdot}_2$ and taking square root, we obtain
\begin{align*}
    \E\left[\norm{\tilde Q^k_{n} -Q_\infty^k}_\infty \right] &\leq \E\left[\norm{\tilde Q_{n}^k -Q_\infty^k}_2 \right]\\ 
    &\leq \frac{2(t_0 + 2)\norm{ Q_{\max}^k -Q_\infty^k}_2}{\sqrt{(n + t_0)(n+t_0+1)}} + \frac{{8(1+L_h)}\sqrt{n}/{(1-\gamma)}}{\rho_F\sqrt{(n+t_0)(n+t_0+1)}}  + \frac{20C_{pop}}{\rho_F}\\
    &\leq \frac{4(t_0 + 2)(1+L_h)/(1-\gamma)}{\sqrt{(n + t_0)(n+t_0+1)}} + \frac{8(1+L_h)\sqrt{n}/(1-\gamma)}{\rho_F\sqrt{(n+t_0)(n+t_0+1)}}  + \frac{20C_{pop}}{\rho_F},
\end{align*}

since $\norm{\tilde Q_{n}^k -Q_\infty^k}_2 \leq \frac{2(1+L_h)}{1-\gamma}$. For brevity, we write the error bound for the output of the algorithm $\tilde Q^k_{I_{ctd}}$ as:  
\begin{align*}
     \E\left[\norm{\tilde Q^k_{I_{ctd}} -Q_\infty^k}_\infty \right] \leq  \frac{C_1}{\sqrt{(I_{ctd} + t_0)(I_{ctd}+t_0+1)}} + \frac{C_2\sqrt{I_{ctd}}}{\sqrt{(I_{ctd}+t_0)(I_{ctd}+t_0+1)}} + \frac{C_{pop,1}}{\sqrt{\min_i N_i}} + C_{pop,2}L_{pop}^{I_{\mix}}
\end{align*}
with 
\begin{align*}
    &C_1 = \frac{4(t_0 + 2)(1+L_h)}{(1-\gamma)},\\
    &C_2 = \frac{8(1+L_h)}{(1-\gamma)\rho_F} = \frac{8(1+L_h)}{(1-\gamma)^2\delta_\mix' \zeta},\\
    &C_{pop,1} = \frac{20(\delta_\mix'+4T_\mix)(1+L_h)(r_\mu + p_{\mu})(3p_*+2p_*p_\mu+2)\sqrt{2|\ss|}}{{\delta'_{\mix}}^{3}(1-\gamma)^2(1-
    L_{pop})},\\
    &C_{pop,2} = \frac{40(1+ L_h)(p_\mu+r_{\mu}) }{{\delta_\mix'}^2(1-\gamma)^2\delta'_\mix \zeta}.
\end{align*}
Then, the conclusion holds when $I_{\mix} > \mathcal{O}(\log \epsilon^{-1})$ and $I_{ctd} > \mathcal{O}(\epsilon ^{-2})$. 
\end{proof}

We remark that Theorem \ref{Theorem: GGR-S CTD Learning} proves that the output of the CTD learning algorithm well approximates the true $Q^k$ function under the limiting mean-field ensemble induced by any policy profile. 

\subsection{Proof for Theorem \ref{Theorem: GGR-S PMA-CTD Learning}}
\label{appendix: proof for theorem GGRS PMACTD learning}
\begin{proof}
    Under the same settings of Theorem \ref{Theorem: GGR-S CTD Learning}, we consider the $m$-th PMA policy update iteration: the current policy profile is $\bpi_m$, and denote the $q$-function output by the algorithm as $\tilde q^k_m$ for population $k$; and denote the true $q$-function under the stable mean-field ensemble $\Gamma_{pop}^\infty(\bpi_m)$ as $q^k_m$. Then, by Theorem \ref{Theorem: GGR-S CTD Learning}, with probability 1, 
    \begin{align*}
          &\E\left[\norm{\tilde q^k_{m} -q_m^k}_\infty\big|\pi_m \right] = \E\left[\norm{\tilde Q^k_{I_{ctd}} -Q_\infty^k}_\infty \right]\\ &\leq \frac{C_1}{\sqrt{(I_{ctd} + t_0)(I_{ctd}+t_0+1)}} + \frac{C_2\sqrt{I_{ctd}}}{\sqrt{(I_{ctd}+t_0)(I_{ctd}+t_0+1)}} + \frac{C_{pop,1}}{\sqrt{\min_i N_i}} + C_{pop,2}L_{pop}^{I_{\mix}}
    \end{align*}
    Denote the optimal policy profile by $\bpi^*$. It holds with probability $1$ that
    \begin{align*}
        \E\left[\norm{\bpi_{m+1} - \bpi^*}|\bpi_m\right] &= \E\left[\norm{\Gamma_\eta^{pma}(\tilde q^k_m, \bpi_m) - \bpi^*}_1 |\bpi_k \right]\\
        & \leq \E\left[\norm{\Gamma_\eta^{pma}(q^k_m, \bpi_m) - \bpi^*}_1 |\bpi_k \right] + \E\left[\norm{\Gamma_\eta^{pma}(\tilde q^k_m, \bpi_m) - \Gamma_\eta^{pma}(q^k_m, \bpi_m) }_1 |\bpi_k \right], \text{ triangle inequality}\\
        &\leq \E\left[\norm{\Gamma_\eta(\bpi_m) - \bpi^*}_1 |\bpi_k \right] + \E\left[\norm{\Gamma_\eta^{pma}(\tilde q^k_m, \bpi_m) - \Gamma_\eta^{pma}(q^k_m, \bpi_m) }_1 |\bpi_k \right]\\
        &\leq L_{\eta}\norm{\bpi_m - \bpi^*}_1 + L_{md,q}\E\left[\norm{\tilde q^k_{m} -q_m^k}_\infty\big|\pi_m \right],
    \end{align*}
    and it follows by the law of iterated expectations, 
    \begin{align*}
        \E\left[\norm{\bpi_{m+1} - \bpi^*}\right] \leq L_\eta\E[\norm{\bpi_m - \bpi^*}_1] + L_{md,q}\E\left[\norm{\tilde q^k_{m} -q_m^k}_\infty\right],
    \end{align*}
    and recursively, with $L_\eta < 1$, $I_{mix} > \mathcal{O}(\log \epsilon^{-1})$, $I_{ctd} > \mathcal{O}(\epsilon ^{-2})$, and $M > \mathcal{O}(\log(\epsilon^-1))$ we obtain:
    \begin{align*}
         \E\left[\norm{\bpi_{M} - \bpi^*}\right] \leq \epsilon + \mathcal{O}\left(\frac{1}{\sqrt{\min_i N_i}}\right).
    \end{align*}
    This completes the proof.     
\end{proof}

\section{Numerical Results}
In this section, we conduct numerical experiments on an toy epidemic model (similar to the SIR model proposed by \cite{cui&koeppl}) to illustrate the effectiveness of our proposed algorithms. 

\subsection{Model Setting}

We consider three population of agents:
\begin{itemize}
    \item population 1: people who recovers faster and least susceptible;
    \item population 2: people who are more susceptible;
    \item population 3: people who are most susceptible.
\end{itemize}
Then, we formulate the spread of virus among these three populations as a MP-MFG: 
\begin{itemize}
    \item State space: $s \in \ss = \{H,S\}$. $H$ represents the state of healthy and $S$ represents the state of sick. 
    \item Action space: $a \in \aa = \{Y,N\}$. Y represents an agent chooses to wear mask and N represents an agent chooses not to wear mask.
    \item Strength of Connectivity: $W$ matrix:
    \begin{align*}
        W = \begin{bmatrix}
          0.5 & 0.4 & 0.5\\
        0.4& 0.6 & 0.3\\
        0.5 & 0.3 & 0.7\\
        \end{bmatrix}
    \end{align*}
    where the $(i,j)$th entry denotes the weight of the impact of the mean field state of population $j$ on population $i$. We set the matrix to be symmetric given the definition of Graphon game.

    \item Reward function:
     \begin{align*}
         r^{k}(s,a,z^{k}) =  -2*k\cdot \one\{s=S\}(s) - \one\{a=Y\}(a)-\one\{s=S\}(s)\cdot \one\{a=N\}(a)\cdot 0.5,k=1,2,3
     \end{align*}
    \item Probability transition: 
    \begin{align*}
        P(s_{t+1} = H | s_t = S) &= 0.3\\
        P(s_{t+1} = S | s_t = H, a_t = Y) &= 0.8 \cdot  z_t(S)+0.1\\
        P(s_{t+1} = S | s_t = H, a_t = N) &= 0.55 \cdot z_t(S)+0.3
    \end{align*}
\end{itemize}
\subsection{The Complete Information Case}
We implement the PMA procedure outlined in Section \ref{subsection: Solution to MP-MFG with Complete Information}; and the results of the learned policy profiles are presented in Figure \ref{Exact_result}. Generally, the more vulnerable a type of people are to the virus, the more likely they will wear a mask. 
\begin{figure}[H]
    \centering
    \includegraphics[scale=0.25]{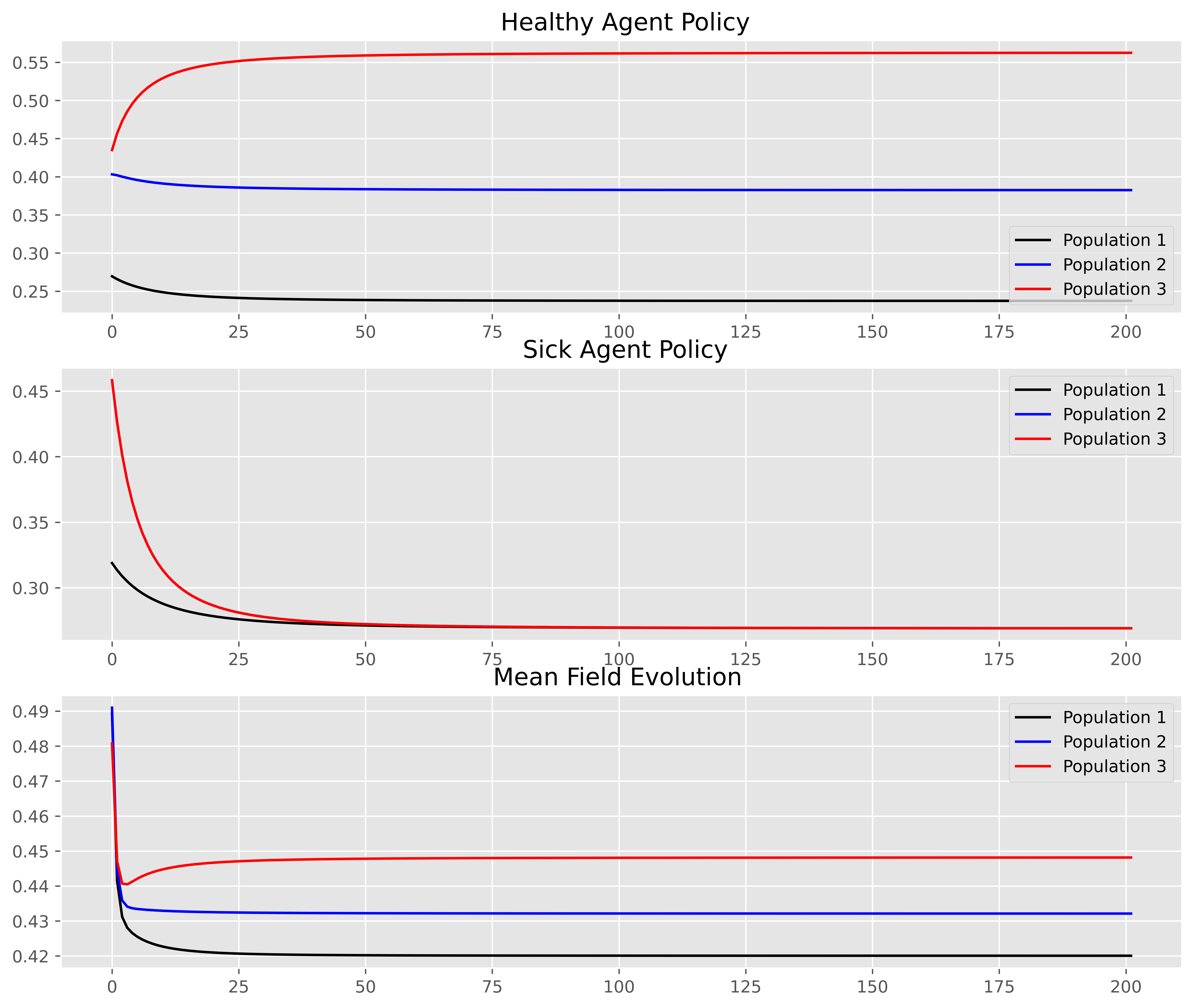}
    \caption{Result of Exact Case}
    \label{Exact_result}
\end{figure}

\subsection{Metrics}

We use three metrics to compare the output of the algorithm to that of the complete information case.
\begin{itemize}
    \item \textbf{Average Discounted Regularized Reward}.
    
    We consider the expected value of the value function for each population \(V^k(\bpi,\bmu_\infty^{\bpi})\) under the stable mean-field induced by the policy profile $\bpi$, which we denote by $\bmu_\infty^{\bpi} = \Gamma_{pop}^\infty(\bpi)$. 
    \begin{align*}
        \Bar{V}(\bpi,\bmu_\infty^{\bpi})= \E_{s\sim \bmu_\infty^{\bpi}}[V(s|\bpi,\bmu_\infty^{\bpi})]
    \end{align*}
    \item \textbf{Exploitability.}
    
    The exploitability is generally defined as 
\begin{align*}
    &\mathcal{E}(\bpi) = \max_{\bpi'}\Bar{V}(\bpi',\bmu_{\infty}^{\pi'})-\Bar{V}(\bpi,\bmu_\infty^{\bpi})\label{defn: exploitability}
\end{align*}

This measures how much an agent can gain by further optimizing the policy.
Since we do not have access to the exact best response policy, we approximate the exploitability metric. We freeze the current mean field and run $m$ iterations of policy improvement, then we can use this value to compute the exploitability.
    \item  \textbf{Policy Convergence}.
    
    Suppose in the \(m^{th}\) iteration the policy for all population is \(\bpi_m=(\pi^{1}_m,\pi^{2}_m,\pi^{3}_m)\), then we use
    \begin{align*}
        \norm{\bpi_{m}-\bpi^{*}}_1= \max_{i\in \{1,2,3\}} \norm{\pi^{i}_m-\pi^{*,i}}_1
    \end{align*}
    where \(\pi^{i}_m\) represents the policy of the \(i^{th}\) population during the $m$-th iteration, and $\bpi^*$ is the optimal policy profile. 
\end{itemize}

 \subsection{The Simulator-Oracle-Based Learning}
In this section we implement simulator-based PMA procedure outlined in Algorithm \ref{algo:Simulator-based PMA Reinforcement Learning for MP-MFG NE}, and the results converge quickly to the complete information case. Since the central learner has no access to  real transitions,  it will query the simulator to estimate the transition probabilities.

\begin{figure}[H]
    \centering
\includegraphics[scale=0.25]{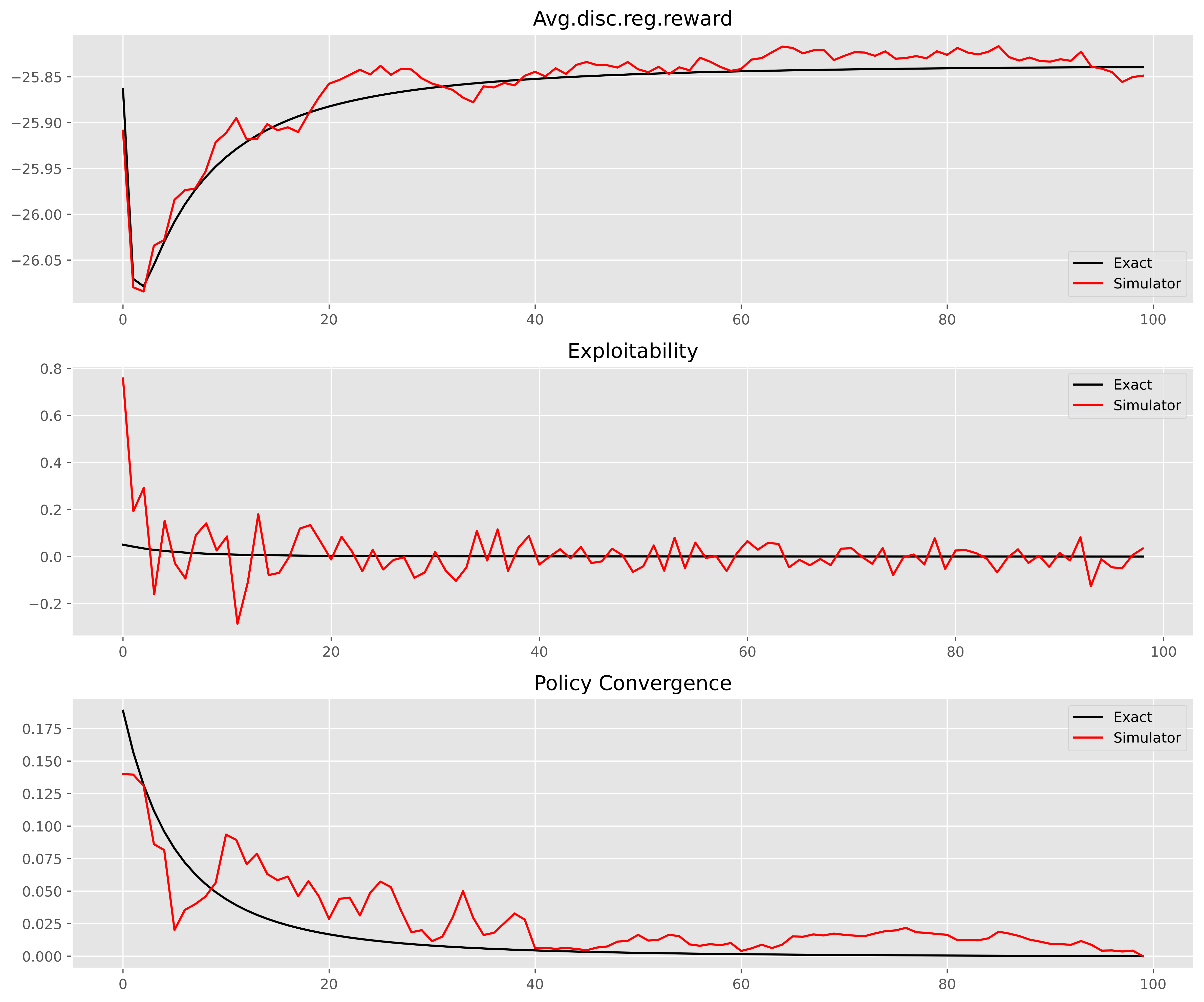}
    \caption{Simulator Result}
    \label{Sim_Result}
\end{figure}
We also specify our choice for hyperparameters in the algorithm.

\begin{longtable}{ccp{8cm}} 
\caption{Hyperparameters}\\
\hline
Hyperparameter & Value &  Comment  \\
\hline
\endfirsthead
\multicolumn{3}{c}
{\tablename\ \thetable\ -- \textit{Continued from previous page}} \\
\hline
 Hyperparameter &  Value & Comment \\
\hline
\endhead
\hline
\multicolumn{3}{r}{\textit{Continued on next page}} \\
\endfoot
\hline
\endlastfoot

        $\lambda$ & 1.0 & $\lambda$ is the scale of the regularizer, where the regularizer is defined as $h_{\lambda}(\pi)=\lambda \cdot \sum_{s\in \mathcal{S}}-\pi(s)\log(\pi(s))$.
        We tested $\lambda$ in \{0.01,0.1,1,10\} and found that 0.1 gave
reasonably stable learning that progressed sufficiently quickly.
Further optimizing this hyper-parameter may lead to better results. \\
        \hline
        N & 100 & In each iteration we use the simulator to roll out N samples to estimate the transition. We tested N in \{10,50,100,500\} and N=100 can yield a satisfactory result. \\
     
         \hline 
        $\bmu_0$ & $\mu_0^{i}(S) = 0.5 \forall i$ & The initial mean field. We set the initial state distribution be the same for each population. \\
         \hline 
    
        $\epsilon_{\pi}$ & 0.002 & The tolerance for the iteration. \\
\end{longtable}

\subsection{GGR-S PMA-CTD learning}
In this section, we implement the GGRS PMA-CTD learning method outlined in Algorithm \ref{algo:GGRS learning}.
To compare the population size's effect on the learning result, we tested the algorithm on different population sizes with $N = 50$ and $N = 500$. It is clear that larger population sizes ensure better convergence, which confirms our theoretical results. 
\begin{figure}[H]
    \centering
\includegraphics[scale=0.25]{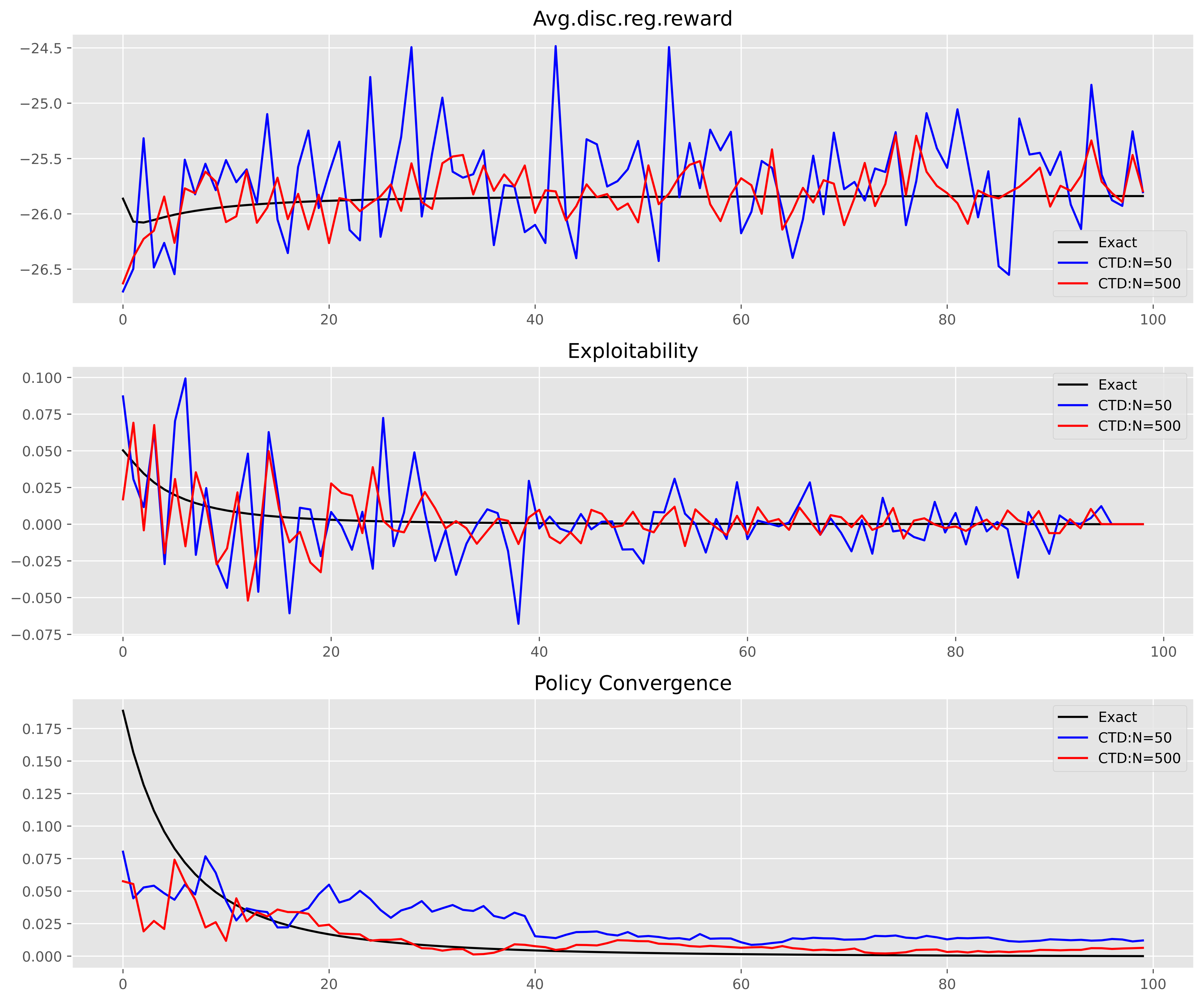}
    \caption{PMA-CTD Learning Result}
    \label{CTD_Result}
\end{figure}
\begin{longtable}{ccp{8cm}} 
\caption{Hyperparameters}\\
\hline
Hyperparameter & Value &  Comment  \\
\hline
\endfirsthead
\multicolumn{3}{c}
{\tablename\ \thetable\ -- \textit{Continued from previous page}} \\
\hline
 Hyperparameter &  Value & Comment \\
\hline
\endhead
\hline
\multicolumn{3}{r}{\textit{Continued on next page}} \\
\endfoot
\hline
\endlastfoot

$\lambda$ & 0.5 & The scale of the regularizer.\\ 
\hline 
$\gamma$ & 0.95 & Discount rate \\ 
\hline 
$I_{ctd}$ & 500 & Iteration round number for the CTD learning, we tested $I_{ctd}$ in 
   \{10,100,500,1000\}, 100 gave reasonably stable learning that progressed sufficiently quickly. \\
   \hline
   $I_{\mix}$ & 200 & Number of fictitious playing \\
   \hline
   $\{\beta_{n}\}_{n=1}^{I_{ctd}}$ & $\frac{2/(1-\gamma)}{n+t_0-1},t_0= 100$ & Learning rate for the CTD learning in each step.\\
\end{longtable}

\end{document}